\documentclass{amsart}
\pdfoutput=1
\usepackage[utf8]{inputenc}
\usepackage[T1]{fontenc}
\usepackage{lmodern}
\usepackage{amssymb,mathtools,stmaryrd,mathrsfs}
\usepackage{tikz-cd}
\usepackage[pdftitle={Consistent symmetry breaking and topological
phases},
pdfauthor={Yuezhao Li, Guo Chuan Thiang},
pdfsubject={Mathematical Physics}
]{hyperref}
\usepackage[noabbrev,capitalize]{cleveref}
\usepackage{slashed}
\usepackage{MnSymbol}
\usepackage{bbm}
\usepackage{diagbox}
\usepackage{subcaption}
\usepackage[version=4]{mhchem}

\usepackage{enumitem}
\setlist[enumerate,1]{label=\textup{(\arabic*)}}
\setlist[enumerate,2]{label=\textup{(\alph*)}}

\usepackage[lite,alphabetic,initials]{amsrefs}

\numberwithin{figure}{section}
\numberwithin{equation}{section}
\numberwithin{table}{section}
\theoremstyle{plain}
\newtheorem{theorem}[equation]{Theorem}
\newtheorem{lemma}[equation]{Lemma}
\newtheorem{proposition}[equation]{Proposition}
\newtheorem{corollary}[equation]{Corollary}
\theoremstyle{definition}
\newtheorem{definition}[equation]{Definition}
\newtheorem*{notation}{Notation}
\theoremstyle{remark}
\newtheorem{remark}[equation]{Remark}
\newtheorem{example}[equation]{Example}

\newtheoremstyle{nonumber} 
  {3pt} 
  {3pt} 
  {\upshape} 
  {} 
  {\bfseries} 
  {.} 
  {.5em} 
  {\thmname{#1}} 
\theoremstyle{nonumber}
\newtheorem{basicsetup}{Basic Setup}
\newcommand{\refbasicsetup}{\hyperref[basic_setup]{Basic Setup}}

\newcommand{\N}{\mathbb{N}}
\newcommand{\Z}{\mathbb{Z}}
\newcommand{\Q}{\mathbb{Q}}
\newcommand{\R}{\mathbb{R}}
\newcommand{\C}{\mathbb{C}}
\newcommand{\F}{\mathbb{F}}
\newcommand{\T}{\mathbb{T}}
\newcommand{\Mat}{\mathbb{M}}
\newcommand{\Cpt}{\mathbb{K}}
\newcommand{\Bdd}{\mathbb{B}}

\newcommand{\I}{\mathbb{I}} 
\newcommand{\HH}{\mathbb{H}} 
\newcommand{\K}{\relax\ifmmode\operatorname{K}\else\textup{K}\fi}
\newcommand{\KO}{\relax\ifmmode\operatorname{KO}\else\textup{KO}\fi}
\newcommand{\KR}{\relax\ifmmode\operatorname{KR}\else\textup{KR}\fi}
\newcommand{\KK}{\relax\ifmmode\operatorname{KK}\else\textup{KK}\fi}
\newcommand{\sK}{\operatorname{sK}}
\newcommand{\wK}{\operatorname{wK}}
\newcommand{\sKO}{\operatorname{sKO}}

\newcommand{\st}{\relax\ifmmode{}^*\else{}\textup{*}\fi}
\newcommand{\Cst}{\relax\ifmmode\mathrm{C}^*\else\textup{C*}\fi}
\newcommand{\Cont}{\mathrm{C}}
\newcommand{\Co}{\mathrm{C}_0}

\newcommand{\Cstl}{\Cst_\lambda} 
\newcommand{\Cstr}{\Cst_\rho} 
\newcommand{\Stab}{\operatorname{St}} 
\newcommand{\ind}{\operatorname{ind}} 
\newcommand{\supp}{\operatorname{supp}} 
\newcommand{\GL}{\mathrm{GL}}
\newcommand{\Isom}{\mathrm{Isom}}
\newcommand{\Ad}{\operatorname{Ad}} 
\newcommand{\id}{\operatorname{id}} 
\newcommand{\ev}{\operatorname{ev}} 
\newcommand{\Ball}{\mathrm{B}} 
\newcommand{\catCst}{\mathfrak{C^*alg}}
\newcommand{\catAb}{\mathfrak{Ab}}
\makeatletter
\newcommand{\colim@}[2]{%
  \vtop{\m@th\ialign{##\cr
    \hfil$#1\operator@font colim$\hfil\cr
    \noalign{\nointerlineskip\kern1.5\ex@}#2\cr
    \noalign{\nointerlineskip\kern-\ex@}\cr}}%
}
\newcommand{\colim}{%
  \mathop{\mathpalette\colim@{\rightarrowfill@\textstyle}}\nmlimits@
}
\makeatother
\newcommand{\defeq}{\mathrel{\vcentcolon=}}
\newcommand{\blank}{-}
\renewcommand{\mathbb}{\mathbbm} 
\DeclarePairedDelimiter{\set}{\lbrace}{\rbrace} 
\DeclarePairedDelimiter{\abs}{\lvert}{\rvert} 
\DeclarePairedDelimiter{\ket}{\lvert}{\rangle} 
\DeclarePairedDelimiter{\bra}{\langle}{\rvert} 
\DeclarePairedDelimiterX{\braketvert}[3]{\langle}{\rangle}{#1\,\delimsize\vert\,\mathopen{}#2\mathopen{}\,\delimsize\vert\,#3} 

\begin{document}
\title{Consistent~symmetry~breaking and topological~phases}
\author{Yuezhao Li}
\address{Max-Planck-Institut f\"ur Mathematik \\
Bonn \\
Germany}
\email{liy@mpim-bonn.mpg.de}
\author{Guo Chuan Thiang}
\address{Beijing International Center
for Mathematical Research \\ 
Peking University \\
China}
\email{guochuanthiang@bicmr.pku.edu.cn }



\begin{abstract}
Operators invariant under a symmetry group are also invariant under any
finite-index subgroup. The equivariant indices of such operators must be
consistent under symmetry breaking. We use this principle to establish a
canonical weak/strong dichotomy for equivariant indices, building on an idea
originating in the theory of topological insulators in solid-state physics.
We also study the relationship to coarse-geometric, or macroscopic,
indices. 
\end{abstract}

\maketitle
\setcounter{tocdepth}{2}
\tableofcontents

\section{Introduction}
The theoretical prediction
\cites{Fu-Kane:Topological_insulators_with_inversion_symmetry,
Fu-Kane-Mele:TI_in_3D, Moore-Balents:Topological_invariants_TRS} and
experimental discovery
\cite{Hsieh_etal:Topological_Dirac_insulator_QSH} of
so-called $\Z/2$-topological insulators in three dimensions are some of the
biggest recent breakthroughs in solid-state physics, see \cite{Hasan-Moore:3D_TI} for a review. On the theoretical
side, there is a dichotomy of $\Z/2$-valued invariants, with three of them
being ``weak'' and one of them being ``strong''. The strong one is
well-established experimentally, whereas the weak ones are generally
considered to be non-robust.

In this paper, we will provide a general theory of weak and strong indices,
applicable whenever
$X$ is a proper metric space equipped with a proper, cocompact action
of a discrete group \( \Gamma \). 
This setup is commonly encountered in geometry and index theory.
For example, if $\slashed{D}$ is a $\Gamma$-invariant, self-adjoint
Dirac-type operator on a Riemannian manifold $X$,
its \( \Gamma \)-equivariant index is defined as a certain K-theory class
\begin{equation}\label{eq:Dirac.index}
\mathrm{ind}_{\Gamma}(\slashed{D})\in \K_*(\Cst(X)^\Gamma),
\end{equation}
where $\Cst(X)^\Gamma$ is the \( \Gamma \)-equivariant Roe algebra, 
generated by
$\Gamma$-invariant, locally compact operators with finite propagation (see
\cref{dfn:Roe_algebra} and \cite{Roe:Index_theory_coarse_geometry}).

The \emph{abelian} case occurs in solid-state physics. Here, $X=\R^d$ and
$\Gamma\simeq\Z^d$ is an abelian lattice of translation symmetries of
a crystalline material. If $\mathscr{H}=-\Delta+V$ is a Schr\"{o}dinger operator with
$\Gamma$-invariant potential function $V$, and $E\in\R$ lies in a spectral
gap, then under mild assumptions, the spectral projection
$P_{\leq E}\defeq \chi_{(-\infty,E]}(\mathscr{H})$
lies in $\Cst(X)^\Gamma$, thus it defines a
K-theory class
\begin{equation}\label{eq:Gamma.equiv.invariant}
[P_{\leq E}]_\Gamma\in \K_0(\Cst(X)^{\Gamma}).
\end{equation}
The abstract Bloch transform 
(see \cref{sec:equiv_Roe_alg_transform}) maps
$P_{\leq E}$ to a family of compact projections, continuously
parametrized by the Pontryagin dual $\widehat{\Gamma}\simeq\T^d$.
This \( d \)-torus \( \widehat{\Gamma} \) is called the Brillouin
torus in solid-state physics.
The range of such a family of projections is a finite-rank vector
bundle over the Brillouin torus, hence it determines a topological
$\K$-theory class,
\begin{equation}\label{eq:topological.K.class}
[P_{\leq E}]_\Gamma\in \K^0(\T^d)\simeq \Z^{2^{d-1}}\quad (d\geq
1).
\end{equation}
When the class \eqref{eq:topological.K.class} is non-trivial, one usually
says that the insulating state described by $P_{\leq E}$ is a
\emph{topological phase}. Physicists use a dichotomy for
\eqref{eq:topological.K.class} (thus also for
\eqref{eq:Gamma.equiv.invariant}) --- classes which come from a
``stacking'' procedure are considered ``weak'', otherwise they are
``strong''. In the special case of a fixed
$G\simeq \Z^2$, this kind of dichotomy was studied in
\cites{Fu-Kane-Mele:TI_in_3D, Moore-Balents:Topological_invariants_TRS,
Ewert-Meyer:Coarse_geometry,
Prodan-SBaldes:Complex_topological_insulators}. However, this formulation depends on implicit choices and does not have an obvious generalization for non-abelian $\Gamma$.

We will introduce the general weak/strong dichotomy in a \emph{canonical} way, based on the principle of \emph{consistent} symmetry
breaking. 
Consider the collection \( \mathfrak{S} \)
of all finite-index subgroups
of $\Gamma$. For any \( G,H\in\mathfrak{S} \) with \( H\subseteq
G \), every $G$-invariant operator is automatically $H$-invariant. So there
is a canonical forgetful map
\[ 
\iota_{H\backslash G}\colon \Cst(X)^G\to\Cst(X)^H
\]
between their equivariant Roe algebras. Moreover, the canonical forgetful
maps for all such pairs of subgroups
are consistently related, in the sense that if \( K\subseteq H\subseteq G
\) is a nested sequence of (finite-index) subgroups of \( \Gamma \), then
\[ 
  \iota_{K\backslash G}=\iota_{K\backslash H}\circ\iota_{H\backslash
  G}.
\] 
We refer to these \( \iota_{H\backslash G} \), as well as the
induced K-theory maps, as \emph{symmetry breaking maps}. 

In general, both the (finite-index)
subgroup structure of $\Gamma$ and the maps $\K_*(\iota_{H\backslash G})$ are very complicated. Nevertheless, the full ``consistent symmetry
breaking'' diagram possesses a colimit \Cst-algebra \(
\Cst(X)^\mathfrak{S} \), which we call the \emph{symmetry breaking Roe
algebra}. For many \( \Gamma \), the K-theory of the symmetry breaking Roe
algebra contains \emph{divisible} summands; see
\cref{tab:K-theory_untwisted} for the abelian case 
(\( \Gamma\simeq\Z^d \)). 
Consideration of this colimit \Cst-algebra and its K-theory 
leads, at each finite-index subgroup $G\subseteq\Gamma$, to a
canonical short exact sequence of abelian groups (see \cref{dfn:weak_strong}):
\begin{equation}\label{eq:informal.weak.strong}
\begin{tikzcd}[column sep=small]
0\arrow[r] & {\rm wK}_*(\Cst(X)^G) \arrow[r] & \K_*(\Cst(X)^G)
\arrow[r] & {\rm sK}_*(\Cst(X)^G) \arrow[r] & 0.
\end{tikzcd}
\end{equation}
We call ${\rm wK}_*(\Cst(X)^G)$ the ``weak subgroup'' and ${\rm
sK}_*(\Cst(X)^G)$ the ``strong quotient''. 

Intuitively, an index class in $\K_*(\Cst(X)^\Gamma)$ becomes more robust after passing to a finite-index subgroup \(
G\subseteq \Gamma \) and viewing it as an index class in \( \K_*(\Cst(X)^G) \), because it remains
well-defined and stable under a larger class of $G$-invariant
perturbations which are not necessarily $\Gamma$-invariant.
Geometrically, the fundamental domain becomes enlarged, but remains bounded
as $G\subseteq \Gamma$ has finite index, see \cref{fig:coarse-graining}.
The passage from $\Gamma$ to $G\subseteq\Gamma$ is thus a ``coarse
graining'' procedure, and the colimit of this procedure is what survives,
informally speaking, in the ``macroscopic limit''.

\begin{figure}[ht]
\centering
\begin{tikzpicture}[
    scale=0.75,
    fine/.style={fill=black},
    coarse/.style={fill=black},
    edge/.style={densely dotted},
    block/.style={fill=gray!15,draw=gray!50,line width=0.3pt}
]


\filldraw[block] (-2.5,-2.5) rectangle (-1.5,-1.5);

\foreach \x in {-2,-1,0,1,2,3}{
            \draw[edge] (\x,-2.5)--(\x,3.5);
}
\foreach \y in {-2,-1,0,1,2,3}{
            \draw[edge] (-2.5,\y)--(3.5,\y);
}

\foreach \x in {-2,-1,0,1,2,3}{
    \foreach \y in {-2,-1,0,1,2,3}{
        \fill[fine] (\x,\y) circle (2.5pt);
    }
}

\node[below] at (0.5,-2.5) {$\mathbb Z^2$};


\draw[->,thick] (4,0)--(6.6,0)
    node[midway,above,font=\small]{coarse-graining};
\draw[->,thick] (0,4)--(2,6)
    node[midway,left,font=\small]{coarse-graining};
\draw[->,thick] (0,-4)--(2,-6)
    node[midway,left,font=\small]{coarse-graining};
\draw[->,thick] (9,6)--(11,4)
    node[midway,right,font=\small]{coarse-graining};
\draw[->,thick] (9,-6)--(11,-4)
    node[midway,right,font=\small]{coarse-graining};


\begin{scope}[xshift=5cm,yshift=6.5cm]

\filldraw[block] (-2.5,-2.5) rectangle (-0.5,-0.5);

\foreach \x in {-2,-1,0,1,2,3}{
            \draw[edge] (\x,-2.5)--(\x,3.5);
}
\foreach \y in {-2,-1,0,1,2,3}{
            \draw[edge] (-2.5,\y)--(3.5,\y);
}

\foreach \x in {-2,0,2}{
    \foreach \y in {-2,0,2}{
        \fill[coarse] (\x,\y) circle (2.5pt);
    }
}

\foreach \x in {-2,-1,0,1,2,3}{
    \foreach \y in {-2,-1,0,1,2,3}{
        \fill[fine] (\x,\y) circle (1.5pt);
    }
}

\node[below] at (0.5,-2.5) {$2\mathbb Z^2$};
\end{scope}


\begin{scope}[xshift=5cm,yshift=-6.5cm]

\filldraw[block] (-2.5,-2.5) rectangle (-0.5,0.5);

\foreach \x in {-2,-1,0,1,2,3}{
            \draw[edge] (\x,-2.5)--(\x,3.5);
}
\foreach \y in {-2,-1,0,1,2,3}{
            \draw[edge] (-2.5,\y)--(3.5,\y);
}

\foreach \x in {-2,0,2}{
    \foreach \y in {-2,1}{
        \fill[coarse] (\x,\y) circle (2.5pt);
    }
}

\foreach \x in {-2,-1,0,1,2,3}{
    \foreach \y in {-2,-1,0,1,2,3}{
        \fill[fine] (\x,\y) circle (1.5pt);
    }
}

\node[below] at (0.5,-2.5) {$2\mathbb Z\times 3\Z$};
\end{scope}


\begin{scope}[xshift=10cm]

\filldraw[block] (-2.5,-2.5) rectangle (3.5,3.5);

\foreach \x in {-2,-1,0,1,2,3}{
            \draw[edge] (\x,-2.5)--(\x,3.5);
}
\foreach \y in {-2,-1,0,1,2,3}{
            \draw[edge] (-2.5,\y)--(3.5,\y);
}

\fill[coarse] (-2,-2) circle (2.5pt);

\foreach \x in {-2,-1,0,1,2,3}{
    \foreach \y in {-2,-1,0,1,2,3}{
        \fill[fine] (\x,\y) circle (1.5pt);
    }
}

\node[below] at (0.5,-2.5) {$6\mathbb Z^2$};
\end{scope}

\end{tikzpicture}
\caption{In the left diagram, the large black circles indicate an orbit for
the standard $\Z^2$-action on $\R^2$, while the shaded $1\times 1$ square
indicates a fundamental domain. Symmetry breaking to $2\Z^2$ is depicted by
the top diagram ---  the shaded $2\times 2$ square indicates a
``coarse-grained'' fundamental domain containing a single point of the
$2\Z^2$-orbit (large black circle). Points of the $\Z^2$-orbit (small black
circles) are no longer individually distinguished, but are lumped together
in fours. We may also break $\Z^2$ to $2\Z\times 3\Z$ (bottom) or $6\Z^2$
(right); the latter breaking can be done in two stages in several different
ways (upper arrows or lower arrows).}
\label{fig:coarse-graining}
\end{figure}

Physically measurable properties should not depend on artificial choices
made in the theoretical model used to predict them. For example, when
using indices in $\K_*(\Cst(X)^G)$ to study properties of a physical
system, different choices of (presentations of) $G$ should give
consistent predictions. Specifically, as will be explained in
\cref{sec:weak_strong_phases}, consistency with symmetry breaking
means that any $G$-dependent prediction must factor through the
above-mentioned colimit.
For the
dichotomy \eqref{eq:informal.weak.strong}, weak classes 
will be defined as those that
become arbitrarily divisible in the colimit, whereas the strong ones remain
``quantized'' in the colimit. 

\emph{Quantized} (i.e., non-divisible) \emph{macroscopic}
phenomena can depend only on the (colimit of) \emph{strong} quotients. The
motivating examples from experimental physics are the Hall response of a 2D
insulating material, and the time-reversal invariant
$\Z/2$-topological insulators. In these examples, one encounters some novelties, not usually considered in the mathematical literature. In
the first case, one is dealing with \emph{magnetic} Schr\"{o}dinger
operators which are only invariant under a \emph{projective} representation
of $\Gamma$. In the second case, the relevant topological invariants are
\emph{real} $\K$-theory classes. Our formalism applies in
such twisted or real settings as well.

We will provide several explicit computations to illustrate our
theory. For the abelian lattices $\Gamma\simeq\Z^d$, we work out the
symmetry breaking diagrams for both real and complex $\K$-theory, as well
as the twisted complex case. We also study the non-abelian Heisenberg group
lattice in the complex case, as an example of how the coarse Baum--Connes
conjecture may be exploited in such computations. Here, for finite index subgroups $G\subseteq \Gamma$, we use the
comparison maps
\begin{equation}\label{eq:comparison.to.coarse}
\K_*(\Cst(X)^G)\to \K_*(\Cst(X))
\end{equation}
which forget the $G$-equivariance altogether, where the range \( \Cst(X) \)
is the non-equivariant Roe algebra of \( X \). In
\cite{Ewert-Meyer:Coarse_geometry}, coarse indices in $\K_*(\Cst(X))$
were studied for the case of a fixed abelian lattice, and used to formulate
a notion of ``strong'' indices. In our general setting, we do not assume an
identification between the concepts of ``strong'' and ``coarse'' indices.

In the abelian case
$\Gamma\simeq\Z^d$, we show that the comparison maps
\eqref{eq:comparison.to.coarse}, descended to the strong quotients, are isomorphisms, so our formalism is consistent
with \cite{Ewert-Meyer:Coarse_geometry}. The same holds for $\Gamma$ the
non-abelian integer Heisenberg group (\cref{sec:Heisenberg}).
However, we give an example of a $\Gamma$ for which the comparison map does
not surject onto $\K_*(\Cst(X))$  (see
\cref{sec:crystallographic.generalization}), thus the concepts of
``strong'' and ``coarse'' do not coincide in general. Nevertheless, we show
that for a large class of groups, the comparison map from the
symmetry breaking \emph{colimit} of strong quotients to $\K_*(\Cst(X))$ is
surjective (\cref{cor:existence.strong.colimit}).

\subsection*{Acknowledgements}
Part of the work was started when YL was a PhD student at Leiden University supported by 
NWO project 613.009.142
``\href{https://www.nwo.nl/projecten/613009142}{Noncommutative index theory
of discrete dynamical systems}''. YL thanks Leiden University and Max Planck Institute for Mathematics for their financial support, and Beijing International Center for Mathematical Research for their hospitality. GCT thanks M.~Ludewig for helpful discussions, as well as A.~Ferreri and Y.~Chen for working out a related concept of ``band folding'' after Bloch--Floquet transform. YL thanks Th.~Schick for various illuminating suggestions, as well as C.~Bourne and B.~Mesland for their
proofreading and valuable comments.

\section{Equivariant Roe algebras and abstract Bloch transform}
\label{sec:equiv_Roe_alg_transform}

In this section, we recall the construction of \( \Cst(X,\mathcal{H})^G \), 
the equivariant Roe algebra defined by an \( X \)-\( G \)-module \(
\mathcal{H} \).  A standard reference is
\cite{Willett-Yu:Higher_index_theory}. We allow for a few modifications:
the $G$-action can be non-free, and the \( X \)-\( G \)-module \(
\mathcal{H} \) can be real or twisted by a 2-cocycle.  We also relate
these equivariant Roe algebras to the (twisted) group $\Cst$-algebras of
$G$.

\begin{notation}
Let \(\mathcal{H}\) be a Hilbert space.
We write $\Bdd(\mathcal{H})$
(respectively $\Cpt(\mathcal{H})$) for the $\Cst$-algebra of bounded
operators (respectively compact operators) on
$\mathcal{H}$, \(\mathrm{U}(\mathcal{H})\) for the group of unitary operators on \(\mathcal{H}\), 
and $1_\mathcal{H}$ for the identity operator on $\mathcal{H}$. 
We also write \(\mathrm{U}(1)\) for the group of unitary complex numbers.
\end{notation}

\subsection{Twisted regular representations and group
\Cst-algebras}
A projective unitary representation of a discrete countable group $G$ on a complex Hilbert space
$\mathcal{H}$ comprises the following data:
\begin{itemize}
\item a \( \mathrm{U}(1) \)-valued group \( 2 \)-cocycle \( \sigma
\); that is, a map \( \sigma\colon G\times G\to\mathrm{U}(1) \)
satisfying
\begin{equation}\label{eq:2_cocycle_condition}
\begin{aligned}
\sigma(g_1,g_2)\sigma(g_1g_2,g_3)&=\sigma(g_1,g_2g_3)\sigma(g_2,g_3),\\
\sigma(g,1)&=\sigma(1,g)=1,
\end{aligned}
\end{equation}
for all \( g,g_1,g_2,g_3\in G \).
\item a collection of elements \(
U_g\in\mathrm{U}(\mathcal{H}) \) for all \( g\in G \), satisfying
\[ 
  U_gU_h=\sigma(g,h)U_{gh},\qquad g,h\in G.
\]
\end{itemize}
Note that
\begin{equation}\label{eq:twisted.inverse.formula}
U_g^{-1}=U_g^*=\overline{\sigma(g^{-1},g)}\cdot U_{g^{-1}}=\overline{\sigma(g,g^{-1})}\cdot U_{g^{-1}}.
\end{equation}

We recall how to construct projective unitary
representations of \( G \) on \( \ell^2(G) \), called \( \sigma \)-twisted
regular representations, as follows.
For \( g\in G \), let
\[ 
    \delta_g\colon G\to\C,\quad \delta_g(g')=\begin{cases}
    1 & g=g'; \\
    0 & g\neq g',
    \end{cases}  
\]
so \( \{\delta_g\}_{g\in G} \) is a basis for \( \ell^2(G) \). 
The \emph{left/right $\sigma$-twisted regular representations} of $G$ 
on $\ell^2(G)$ are defined by the unitary operators
\begin{equation}\label{eq:twisted_regular_rep}
\begin{aligned}
\lambda_G^\sigma(g)\colon \delta_{g^\prime}&\mapsto \sigma(g,g^\prime)\delta_{gg^\prime},\\
\rho_G^\sigma(g)\colon \delta_{g^\prime}&\mapsto \sigma(g^\prime g^{-1},g)\delta_{g^\prime g^{-1}},\qquad g,g^\prime\in G.
\end{aligned}
\end{equation}

We write \( \Cstl(G,\sigma) \) or \( \Cstr(G,\sigma) \) for the \emph{reduced group
\Cst-algebras}, generated inside $\Bdd(\ell^2(G))$ by the left or right $\sigma$-regular
representation of \( G \), respectively. If $\sigma= 1$, we just write $\Cstl(G)$ and $\Cstr(G)$.

\begin{definition}\label{dfn:cohomologous_cocycles}
Two 2-cocycles $\sigma,\sigma^\prime\colon G\times G\to\mathrm{U}(1)$ are
\emph{cohomologous} if there exists a map $\omega\colon G\to \mathrm{U}(1)$ such that 
\begin{equation}\label{eq:cohomologous_cocycles}
\sigma^\prime(g,h)=\sigma(g,h)\cdot\frac{\omega(g)\omega(h)}{\omega(gh)},\qquad g, h\in G.
\end{equation}
\end{definition}
When working with projective representations, we allow a modification of the $U_g$ by some ($g$-dependent) phases. For example, we say that a unitary operator $V$ \emph{intertwines} the projective representations $\{U_g\}_{g\in G}$ and $\{U^\prime_g\}_{g\in G}$, if
\begin{equation}\label{eq:projective_intertwine}
VU_gV^*=\omega(g)U^\prime_g,\qquad g\in G
\end{equation}
for some phases $\omega(g)$. In this case, the respective 2-cocycles $\sigma$ and $\sigma^\prime$ are cohomologous.

\subsection{$X$-$G$-modules}\label{sec:XG.module}
Let \( \Isom(\R^d) \) be the isometry group of the \( d
\)-dimensional Euclidean space, i.e.~the group generated by translations,
rotations and reflections of \( \R^d \). Then it acts on \( L^2(\R^d) \)
via unitaries,
\begin{equation}\label{eq:pullback.rep} 
U_g\xi(x)=\xi(g^{-1}\cdot
x),\qquad g\in \Isom(\R^d),\, x\in X,\, \xi\in L^2(\R^d).  
\end{equation}
The Hilbert space \( L^2(\R^d) \) is the space of wavefunctions in solid-state physics.
The $\Cst$-algebra $\Co(\R^d)$ acts on \( L^2(\R^d) \) by pointwise multiplication,
which is compatible with the unitary representation of \(
\Isom(\R^d) \) given in \eqref{eq:pullback.rep} above.

The free Hamiltonian on \( \R^d \), i.e.~the Laplacian
\[
  -\Delta=-\sum_{i=1}^d\frac{\partial^2}{\partial x_i^2}
\]
is invariant under the action of \( \Isom(\R^d) \).
In the presence of a periodic potential {\( V \)} (e.g.~from an atomic crystal), the symmetry group is reduced to a crystallographic group,
i.e.~a discrete cocompact subgroup of $\Isom(\R^d)$.
Disregarding point group symmetries, the maximal abelian subgroup
of such a crystallographic group is a 
rank-$d$ lattice $\Gamma$ in the translation
subgroup $(\R^d,+)$.
Then one uses the isomorphisms
\[ 
  L^2(\R^d)\simeq \ell^2(\Gamma)\otimes L^2(\R^d/\Gamma)\simeq
  L^2(\widehat{\Gamma})\otimes L^2(\R^d/\Gamma), 
\] 
(see
\cref{prop:V_D_intertwine}, \cref{rem:compose.with.Fourier}) to analyze the $\Gamma$-invariant Hamiltonian \( -\Delta+V \).  Under mild assumptions, the
spectral projections for gapped $\Gamma$-invariant Hamiltonians lie in a
($\Gamma$-equivariant) Roe algebra
\cite{Ewert-Meyer:Coarse_geometry}*{Proposition 1}, and their $\K$-theory
classes serve as invariants distinguishing distinct ``topological insulator
phases''. The situation is abstracted and generalized below. 

\medskip

Let \( X \) be a locally compact Hausdorff space, and \( G \) act on \( X
\) by homeomorphisms, i.e.~there is a group homomorphism \( \alpha\colon
G\to\operatorname{Homeo}(X) \).
The action \( \alpha \) pulls back
to an action on \( \Co(X) \), which we still denote by \( \alpha \): 
\[ 
  \alpha_g(f)(x)\defeq f(\alpha^{-1}_g(x)),\quad f\in\Co(X),\,x\in X.
\]

\begin{definition}[\cite{Willett-Yu:Higher_index_theory}*{Definition 4.5.1
and 4.5.2}]\label{dfn:XG_module}
Let \( X \) be a locally compact Hausdorff space
and \( \alpha\colon G\to\operatorname{Homeo}(X) \) 
be a proper\footnote{This means that for
any compact subset $B\subseteq X$, the set $\{g\in G\,:\,gB\cap B\neq
\emptyset\}$ is finite.} action of \( G \) on \( X \).
\begin{enumerate}
\item An \emph{\( X \)-\( G \)-module} is a Hilbert space $\mathcal{H}$ equipped with:
\begin{itemize}
\item A non-degenerate \st-representation 
\[
\varrho\colon \Co(X)\to\Bdd(\mathcal{H});
\]
\item A projective unitary representation
\[
  U\colon G\to\mathrm{U}(\mathcal{H}),
\]
such that:
\begin{equation}\label{eq:equiv_Xmod}
U_g\varrho(f)U_g^*=\varrho(\alpha_g(f)),\quad
\text{for}\ f\in\Co(X),\, x\in X,\, g\in G.
\end{equation}
\end{itemize}
An \( X \)-\( G \)-module $\mathcal{H}$ is \emph{ample} if 
\[
f\in\Co(X),\,\varrho(f)\in\Cpt(\mathcal{H})\quad\Rightarrow\quad f=0.
\]
\item An \emph{isomorphism} between \( X \)-\( G \)-modules \(
\mathcal{H},\mathcal{H}' \)  is a
unitary operator \( V\colon \mathcal{H}\to\mathcal{H}' \)
which intertwines the representation of \( \Co(X) \) and the (projective) 
unitary representations of \( G \) on them.
\end{enumerate}
\end{definition}

\begin{remark}
We use the terminology
\emph{$\sigma$-twisted $X$-$G$-module} if we wish to emphasize $\sigma$.
We refer to \cref{sec:magnetic_translations} 
for concrete naturally occurring examples
involving magnetic translations.

In physics applications, one also needs to consider \emph{real} Hilbert spaces
(more generally, complex Hilbert spaces equipped with real or quaternionic
structures). All the definitions and results of
\cref{sec:equiv_Roe_alg_transform,sec:symmetry_breaking} hold with
real Hilbert spaces instead of complex Hilbert spaces. In the real setting,
we can allow for ${\rm O}(1)$-valued 2-cocycles, although we will only
consider trivial cocycles in our concrete examples and computations.
\end{remark}

\begin{example}\label{ex:standard_equivariant_ample_modules}
Take $X$ to be the underlying set of $G$, and fix an integer \(N\geq 1 \).
The \( \Cst \)-algebra \( \Cont_0(G) \) is represented on \(
\mathcal{H}=\ell^2(G)\otimes \C^N \) by
pointwise multiplication,
\[
\varrho(f)\cdot\xi(g)\defeq f(g)\xi(g),\quad \xi\in\ell^2(G), f\in\Co(G).
\]
Let $U=\lambda_G^\sigma\otimes 1$ be the left $\sigma$-regular
representation tensored with the identity operator on $\C^N$. Then $\mathcal{H}$ is
an $X$-$G$-module, but it is not ample.

We may replace $\C^N$ in the previous example by any separable Hilbert
space to get an ample $X$-$G$-module;
see~\cite{Willett-Yu:Higher_index_theory}*{Example 4.5.3}.
\end{example}

\medskip
The representation \( \varrho\colon \Co(X)\to\Bdd(\mathcal{H}) \)
extends to the \Cst-algebra of bounded Borel functions on \( X \).
In particular, for $D\subseteq X$ 
a Borel subset and $\chi_D$ its characteristic function, 
 \( \chi_D \) is a bounded Borel function with values in \(
\{0,1\} \). Hence \( \varrho(\chi_D) \) is an orthogonal 
projection on \( \mathcal{H}
\) and we write 
\[
  \mathcal{H}_D\defeq\varrho(\chi_D)\mathcal{H}.
\]
The covariance relation 
\eqref{eq:equiv_Xmod} gives
\begin{equation}\label{eq:spatial_unitary}
U_g\varrho(\chi_D)U_g^*=\varrho(\alpha_g(\chi_D))
=\varrho(\chi_{gD}),\quad
g\in G,
\end{equation}
so that $U_g$ restricts to an isomorphism of Hilbert spaces
$\mathcal{H}_D\xrightarrow{\sim}\mathcal{H}_{gD}$.

\begin{definition}\label{def:fundamental_domain}
A fundamental domain for an \(X\)-\(G\)-module \(\mathcal{H}\) is a Borel subset \( D\subseteq X \) such that
\[
    X=\bigcup_{g\in G}g\overline{D}\quad\text{and}\quad
    \sum_{g\in G}\varrho(\chi_{gD})=1_{\mathcal{H}}.
\]
\end{definition}

Note that \cref{def:fundamental_domain} implies that \(\varrho(\chi_{gD}),
g\in G\), are mutually orthogonal projections on \(\mathcal{H}\). For
convenience, we often omit the representations \( \alpha \) and 
\( \varrho \), and just write $\chi_D$ for $\varrho(\chi_D)$ and
$\chi_{gD}$ for $\varrho(\alpha_g(\chi_D))$. 

\begin{remark}
We note that \cref{def:fundamental_domain} differs from the 
standard definition in the literature (see,
e.g.~\cite{Willett-Yu:Higher_index_theory}*{Definition 5.3.3}), where a
fundamental domain \( D\subseteq X \) is usually
defined as a Borel subset such that \( X=\coprod_{g\in G}gD \).
Instead, we allow the translates
\( gD \) to have non-empty intersection with \( D \); 
the essential requirement is that 
these intersections have zero support on the \( X \)-\( G
\)-module \( \mathcal{H} \). 
Thus we also allow non-free actions \( G\curvearrowright X \), such as that
of a crystallographic group on Euclidean space (see
\cref{sec:crystallographic.generalization}).
\end{remark}

\medskip
Let \( \mathcal{H} \) be an \( X \)-\( G \)-module,
with a fundamental domain \( D \).
By \eqref{eq:spatial_unitary} and \cref{def:fundamental_domain}, 
we obtain a unitary isomorphism of Hilbert spaces,
\begin{equation}\label{eq:V_D}
\begin{aligned}
V_{D}\colon\mathcal{H}=\bigoplus_{g\in G} \mathcal{H}_{gD}
&\overset{\sim}{\longrightarrow}\ell^2(G)\otimes\mathcal{H}_{D},\\
\xi&\longmapsto
\sum_{g\in G}\delta_g\otimes\chi_{D}U_g^*\xi\\
U_g\zeta &\longmapsfrom \delta_g\otimes\zeta.
\end{aligned}
\end{equation}

\begin{proposition}\label{prop:V_D_intertwine}
Let \( \mathcal{H} \) be a $\sigma$-twisted \( X \)-\( G \)-module, and
$D$ be a fundamental domain. The unitary $V_D$ of \eqref{eq:V_D}
intertwines the $\sigma$-projective unitary representation on $\mathcal{H}$
with $\lambda_G^\sigma\otimes 1$.
\end{proposition}
\begin{proof}
Let $U$ be the projective unitary representation on $\mathcal{H}$ with 2-cocycle $\sigma$.
For $g, g^\prime\in G$ and $\zeta\in \mathcal{H}_D$, we have
\begin{align*}
V_D\circ U_{g}\circ V_D^*(\delta_{g^\prime}\otimes \zeta)
&=\sum_{g^{\prime\prime}}\delta_{g^{\prime\prime}}\otimes\chi_D
U_{g^{\prime\prime}}^*U_gU_{g^\prime}\chi_D\zeta & \\
&=\delta_{g g^\prime}\otimes\chi_D U_{g
g^\prime}^*\sigma(g,g^\prime)U_{gg^\prime}\zeta &
(\mathrm{only}\;g^{\prime\prime}=gg^\prime\;\mathrm{contributes})\\
&=(\lambda_G^\sigma(g)\otimes 1)(\delta_{g^\prime}\otimes\zeta).
\end{align*}
This verifies that $\Ad_{V_D}$ converts $U_g$ to
$\lambda_G^\sigma(g)\otimes 1$.
\end{proof}

\begin{remark}\label{rem:compose.with.Fourier}
If $G$ is abelian, we can compose $V_D$ with the Fourier transform on the
$\ell^2(G)$ tensor factor. In the typical solid-state physics context with
$X\simeq\R^d$ and $G\simeq\Z^d$, this gives the (magnetic)
\emph{Bloch--Floquet transform} $\mathcal{H}\to L^2(\T^d)\otimes
\mathcal{H}_D$. For general $G$, we may think of \eqref{eq:V_D} as an
``abstract Bloch transform''.
\end{remark}

\subsection{Transforming equivariant Roe algebras to group algebras}
\label{sec:mapping_equiv_to_group_alg}
We shall define equivariant Roe algebras in this section, and provide
explicit isomorphisms with group \Cst-algebras, following
\cite{Willett-Yu:Higher_index_theory}. To define the equivariant Roe
algebra, we must be able to talk about the propagation of operators on \( X
\)-\( G \)-modules. This can be done when \( X \) is a metric space; or
more generally, if \( X \) carries a coarse structure (see
\cite{Roe:Lectures_on_coarse_geometry}).

Recall that a metric space is \emph{proper} if the closed unit ball is
compact.

\begin{definition}[\cite{Willett-Yu:Higher_index_theory}*{Definition 4.1.8,
5.1.1 and 5.2.1}]\label{dfn:Roe_algebra}
Let \( X \) be a proper metric space,
\( \mathcal{H} \) be an \( X \)-\( G \)-module, and
\( T\in\Bdd(\mathcal{H}) \).
\begin{itemize}
\item The \emph{support} of \( T \)  
denoted by \( \supp(T) \), is the collection of all points \( (x,y)\in
X\times X \), such that \( \chi_AT\chi_B\neq 0 \)
holds for all open neighbourhoods \( A \) of \( x \) and \( B \) of
\( y \).
\item We say \( T \) is \emph{locally compact}, 
if for any compact subset \( K\subseteq X \), the operators
\( \chi_KT \) and \( T\chi_K \) are both compact;
\item We say \( T \) is \emph{controlled} (or \emph{has finite
propagation}), if
\[ 
    \operatorname{prop}(T)\defeq\sup\{\operatorname{dist}(x,y)\,:\,
    (y,x)\in\supp(T)\}<+\infty.
\]
\item We say \( T \) is \emph{\( G \)-invariant}, if \( U_gT=TU_g \)
for all \( g\in G \). 
\end{itemize}
The \emph{equivariant Roe algebra} \( \Cst(X,\mathcal{H})^G \) is
the \Cst-algebra generated by all locally compact, controlled and \( G
\)-invariant operators on $\mathcal{H}$. 
If \( \mathcal{H} \) is ample and $G$ is the trivial group $\{1\}$, 
then
\[
\Cst(X,\mathcal{H})\defeq \Cst(X,\mathcal{H})^{\{1\}}
\]
is called the (non-equivariant) \emph{Roe algebra}.
\end{definition}

\begin{basicsetup}\label{basic_setup}
Subsequently, we will work under the following \emph{basic
setup}:
\begin{enumerate}
\item \( X \) is a proper metric space.
\item \( \Gamma \) is a countable discrete group, 
acting properly on \( X \).\footnote{Here we use the notation $\Gamma$
instead of $G$. From
\cref{sec:symmetry_breaking} onwards, we will start from a fixed
$\Gamma$ and $X$-$\Gamma$-module $\mathcal{H}$, then regard $\mathcal{H}$
as an $X$-$G$-module for any finite-index subgroup $G\subseteq\Gamma$.}
\item $\mathcal{H}$ is a (non-zero) $X$-$\Gamma$-module which admits a 
\emph{bounded} fundamental domain.
\end{enumerate}
\end{basicsetup}

It is known that a countable discrete group $\Gamma$ admits a proper
left-invariant metric, see
\cite{Smith:Asymptotic_dimension_abelian_groups}.  Thus given any
countable discrete group \( \Gamma \) and a 2-cocycle \( \sigma\colon
\Gamma\times \Gamma\to\mathrm{U}(1) \), there  exists a proper
metric space \( X \) and a \( \sigma \)-twisted \( X \)-\( \Gamma
\)-module \( \mathcal{H}\) satisfying \refbasicsetup, say, \( X=\Gamma \)
(with any choice of proper left-invariant metric) and \(
\mathcal{H}=\ell^2(\Gamma) \) equipped with the \( \sigma \)-twisted
regular representation (this is \cref{ex:standard_equivariant_ample_modules}).

The boundedness of fundamental domains is crucial for
the following twisted generalization of
\cite{Willett-Yu:Higher_index_theory}*{Proposition 5.3.4}; see also
\cref{rmk:equivariant_coarse_equivalence}.

\begin{proposition}\label{prop:twisted_Willett}
Let $\mathcal{H}$ be a
$\sigma$-twisted $X$-$\Gamma$-module satisfying \refbasicsetup, 
and $D$ be a bounded fundamental domain. Then the unitary $V_D$ of
\eqref{eq:V_D}
induces an isomorphism of $\Cst$-algebras,
\begin{equation}\label{eq:Willett_isomorphism_twisted}
\Ad_{V_D}\colon \Cst(X,\mathcal{H})^\Gamma\xrightarrow{\simeq}
\Cstr(\Gamma,\overline{\sigma})\otimes\Cpt(\mathcal{H}_{D}).
\end{equation}
Here, $\overline{\sigma}$ denotes the 2-cocycle which is complex conjugate to $\sigma$.
\end{proposition}
\begin{proof}
Let $T$ be a locally compact, controlled, $\Gamma$-invariant operator.
Applying formula \eqref{eq:V_D} for $V_D$, for $g^\prime\in \Gamma,
\zeta\in \mathcal{H}_D$, we have
\begin{align*}
\Ad_{V_D}(T)(\delta_{g^\prime}\otimes\zeta)
&=V_D T U_{g^\prime}\zeta\\
&=V_D U_{g^\prime} T\zeta & (TU_{g^\prime}=U_{g^\prime}T)\\
&=(\lambda^\sigma_\Gamma(g^\prime)\otimes 1)(V_D(T\zeta)) &
(\text{\cref{prop:V_D_intertwine}})\\
&=\sum_{g\in \Gamma}\sigma(g^\prime,g)\delta_{g^\prime g}\otimes \chi_D
U_g^* T\zeta & \eqref{eq:twisted_regular_rep}\\
&=\sum_{g\in \Gamma}\overline{\sigma(g^\prime
g,g^{-1})}\sigma(g,g^{-1})\delta_{g^\prime g}\otimes U_g^*\chi_{gD}  T\zeta
& \eqref{eq:2_cocycle_condition}\\
&=\sum_{g\in
\Gamma}\rho_\Gamma^{\overline{\sigma}}(g^{-1})(\delta_{g^\prime})\otimes
U_{g^{-1}}\chi_{gD} T\chi_D\zeta. &
\eqref{eq:twisted.inverse.formula},\eqref{eq:twisted_regular_rep}
\end{align*}
Due to properness of $X$ and the $\Gamma$-action, and $T$ being controlled,
$\chi_{gD}T\chi_D$ is nonzero only for finitely many $g\in \Gamma$, so the
above sum is a finite sum. Furthermore, $T$ is locally
compact, so the above formula for $\Ad_{V_D}(T)$ shows that it lies in
$\mathbb{C}[\Gamma,\overline{\sigma}]\otimes \Cpt(\mathcal{H}_D)$, where
$\mathbb{C}[\Gamma,\overline{\sigma}]$ is the algebra generated by the
$\rho_\Gamma^{\overline{\sigma}}$ operators.
It follows that for $T$ in the norm closure $\Cst(X,\mathcal{H})^\Gamma$,
we have $\Ad_{V_D}(T)$ being in the norm-closure
$\Cstr(\Gamma,\overline{\sigma})\otimes \Cpt(\mathcal{H}_D)$.

Conversely, we check that for
$\rho_\Gamma^{\overline{\sigma}}(g)\otimes A\in
\mathbb{C}[\Gamma,\overline{\sigma}]\otimes \Cpt(\mathcal{H}_D)$, the
operator
\[
  \Ad_{V_D^*}(\rho_\Gamma^{\overline{\sigma}}(g)\otimes A)
\]
is locally compact, and controlled. Furthermore,
$\rho_\Gamma^{\overline{\sigma}}(g)\otimes A$ commutes with
$\lambda_\Gamma^{\sigma}\otimes 1$, so by \cref{prop:V_D_intertwine},
$\Ad_{V_D^*}(\rho_\Gamma^{\overline{\sigma}}(g)\otimes A)$ is
$\Gamma$-invariant. Then we pass to the norm closures.
\end{proof}

Quite generally, if \( V\colon \mathcal{H}_1\to\mathcal{H}_2\) is an
isometry between Hilbert
spaces, then the $\st$-homomorphism \( \id\otimes\Ad_{V}\colon
A\otimes\Cpt(\mathcal{H}_1)\to A\otimes\Cpt(\mathcal{H}_2) \) induces
isomorphisms in $\K$-theory,
\begin{equation}\label{eq:isometry.stabilization}
\Stab=\Stab^{(A)}\colon \K_i(A\otimes\Cpt(\mathcal{H}_1))\overset{\sim}{\to} \K_i(A\otimes\Cpt(\mathcal{H}_2)).
\end{equation}
In particular, if $\mathcal{H}_1$ is one-dimensional, $\Ad_V$ is a corner
embedding \( \C\hookrightarrow \Cpt(\mathcal{H}_2) \), 
and we can view
$\Stab$ as being induced by tensoring with a(ny) rank-1 projection on
$\mathcal{H}_2$.

For the case of $\Cpt=\Cpt(\mathcal{H}_D)$ appearing in
\eqref{eq:Willett_isomorphism_twisted}, we shall write
\[
\Stab_{\mathcal{H}_D}\colon \K_i(\Cstr(G,\overline{\sigma}))\xrightarrow{\sim}\K_i(\Cstr(G,\overline{\sigma})\otimes\Cpt(\mathcal{H}_D))
\]
to keep track of the dependence on $\mathcal{H}_D$. Combining the stabilization isomorphism with \eqref{eq:Willett_isomorphism_twisted}, we have the isomorphism
\begin{equation}\label{eq:nat.trans.symm.break}
\alpha_G\defeq \Stab_{\mathcal{H}_D}^{-1}\circ \K_i(\Ad_{V_{D}})\colon \K_i(\Cst(X,\mathcal{H})^{G}))\xrightarrow{\sim} \K_i(\Cstr(G,\overline{\sigma})).
\end{equation}

Notice that we did not include the dependence on $D$ in the notation $\alpha_G$. This is because of the following result, which may be compared with \cite{Willett-Yu:Higher_index_theory}*{Theorem 5.3.2}.

\begin{proposition}\label{prop:equiv_Roe_decomp_K_canonical}
Let $\mathcal{H}$ be a $\sigma$-twisted $X$-$\Gamma$-module satisfying
\refbasicsetup. The
isomorphism \eqref{eq:nat.trans.symm.break}
is independent of the choice of bounded fundamental domain $D$.
\end{proposition}

\begin{proof}
Let $D,D^\prime$ be two choices of bounded fundamental domains. 
For each $k\in \Gamma$,
\[
\mathcal{H}_{D^\prime}=\chi_{D^\prime}\bigoplus_{k\in \Gamma}\chi_{kD}\mathcal{H}=\bigoplus_{k\in \Gamma}\chi_{D^\prime\cap kD}\mathcal{H}=\bigoplus_{k\in \Gamma}\mathcal{H}_{D^\prime\cap kD},
\]
and similarly,
\[
\mathcal{H}_D=\bigoplus_{k\in \Gamma}\mathcal{H}_{k^{-1}D^\prime\cap D}.
\]
For the $k$-th summands, recall from \eqref{eq:spatial_unitary} that $U_k$ restricts to unitary isomorphisms
\[
U_k\colon \mathcal{H}_{k^{-1} D^\prime\cap D}\to \mathcal{H}_{D^\prime \cap kD}.
\]
We may combine these isomorphisms into a single unitary isomorphism
\[
S\colon \mathcal{H}_{D}=\bigoplus_{k\in \Gamma}\mathcal{H}_{k^{-1}D^\prime \cap D}\to \bigoplus_{k\in \Gamma}\mathcal{H}_{D^\prime \cap kD}=\mathcal{H}_{D^\prime}.
\]

Denote $1_{\ell^2(\Gamma)}$ by $1$. For $g\in \Gamma$ and $\zeta\in \mathcal{H}_{D^\prime}$, we have
\begin{equation}
\label{eq:sum_of_multipliers}
\begin{aligned}
(1\otimes S)\circ V_{D}\circ V_{D^\prime}^*(\delta_g\otimes \zeta)
&=(1\otimes S)\Big(\sum_{k\in \Gamma}\delta_k\otimes \chi_DU_k^*U_g\zeta\Big)\\
&=\sum_{k\in\Gamma}\delta_k\otimes S\chi_D\overline{\sigma(k,k^{-1})}\sigma(k^{-1},g)U_{k^{-1}g}\zeta\\
&=\sum_{k\in\Gamma}\overline{\sigma(k,k^{-1}g)}\delta_k\otimes S\chi_D U_{k^{-1}g}\zeta\\
&=\sum_{k\in\Gamma}\overline{\sigma(gk^{-1},k)}\delta_{gk^{-1}}\otimes S\chi_D U_k\chi_{D^\prime}\zeta\\
&=\sum_{k\in\Gamma}\big(\rho_\Gamma^{\overline{\sigma}}(k)\otimes S\chi_{D\cap kD^\prime}U_k\big)(\delta_g\otimes\zeta).
\end{aligned}
\end{equation}
As $X$ is a proper metric space, $\overline{D}, \overline{D^\prime}$ are
compact, so $\overline{D}\cup\overline{D^\prime}$ is compact. By properness
of the $\Gamma$-action,
\[
D\cap kD^\prime\subseteq \overline{D}\cap k\overline{D^\prime}\subseteq (\overline{D}\cup\overline{D^\prime})\cap k(\overline{D}\cup\overline{D^\prime})
\]
is nonempty only for finitely many $k\in\Gamma$. So the sum in
\eqref{eq:sum_of_multipliers} is a finite sum. Therefore, $(1\otimes
S)\circ V_{D}\circ V_{D^\prime}^*$ is a unitary multiplier for the algebra
\mbox{$\Cstr(\Gamma,\overline{\sigma})\otimes\Cpt(\mathcal{H}_{D^\prime})$},
and it is a standard result that its conjugation action induces the
identity map on
$\K_i(\Cstr(\Gamma,\overline{\sigma})\otimes\Cpt(\mathcal{H}_{D^\prime}))$;
see, e.g., \cite{Willett-Yu:Higher_index_theory}*{Proposition 2.7.5}.
Therefore,
\begin{align*}
\Stab_{\mathcal{H}_{D^\prime}}^{-1}\circ\K_i(\Ad_{V_{D^\prime}})&=\big(\Stab_{\mathcal{H}_D}^{-1}\circ\K_i(\Ad_{1\otimes S^*})\big)\circ\big(\K_i(\Ad_{1\otimes S})\circ \K_i(\Ad_{V_{D}})\big)\\
&=\Stab_{\mathcal{H}_D}^{-1}\circ \K_i(\Ad_{V_{D}}).\qedhere
\end{align*}
\end{proof}

\begin{remark}\label{rmk:equivariant_coarse_equivalence}
We may understand
\cref{prop:twisted_Willett,prop:equiv_Roe_decomp_K_canonical}
from the viewpoint of coarse geometry.

The \emph{coboundedness} assumption in \refbasicsetup\, says that \( X \) is
equivariantly coarsely equivalent to \( \Gamma \), where \( \Gamma \) is
equipped with any choice of proper, left \( \Gamma \)-invariant metric.
In fact, any two such metrics are coarsely equivalent, as explained in
\cite{Smith:Asymptotic_dimension_abelian_groups}.

If \( \mathcal{H} \) is \emph{ample} 
and $\sigma=1$, then the equivariant Roe algebra \(
\Cst(X,\mathcal{H})^\Gamma \) depends on \( X \) only up to equivariant
coarse equivalence, and is independent of the choice of the \( X \)-\(
\Gamma \)-module \( \mathcal{H} \) defining it. 
Indeed, an isomorphism between two ample \( X \)-\(
\Gamma \)-modules is given by a (non-canonical) \( \Gamma \)-equivariant
unitary $V\colon \mathcal{H}\to\mathcal{H}^\prime$, see
\cite{Willett-Yu:Higher_index_theory}*{Proposition 4.5.16}.

In the twisted case, suppose $\mathcal{H},\mathcal{H}^\prime$ are
ample $\sigma$-twisted and $\sigma^\prime$-twisted $X$-$\Gamma$-modules
respectively, with $\sigma^\prime$ cohomologous to $\sigma$. The
equivariant Roe algebra \( \Cst(X,\mathcal{H}^\prime)^\Gamma \) is
unchanged when the operators $U^\prime_g$ are modified by phases, so it depends on the group cohomology class of $\sigma$. The same construction of a
$\Gamma$-equivariant $V\colon \mathcal{H}\to \mathcal{H}^\prime$ gives an
isomorphism \(\Cst(X,\mathcal{H})^\Gamma \simeq
\Cst(X,\mathcal{H}^\prime)^\Gamma\). 

Thus if
\( X \) is \(\Gamma\)-equivariantly coarsely equivalent to \( Y \) and
\( \mathcal{H}_X, \mathcal{H}_Y\) are ample modules for them that satisfy
\refbasicsetup, with cohomologous 2-cocycles $\sigma,\sigma^\prime$,
then \(
\Cst(X,\mathcal{H}_X)^\Gamma\simeq\Cst(Y,\mathcal{H}_Y)^\Gamma \). This is
a special case of \cref{prop:twisted_Willett}, since $\Cstr(\Gamma,\overline{\sigma})$ and
$\Cstr(\Gamma,\overline{\sigma^\prime})$ are isomorphic (see
\eqref{eq:cohom_iso_algebra}).

In the non-ample setting, the equivariant Roe algebra
\( \Cst(X,\mathcal{H})^\Gamma \) depends also on \( \mathcal{H} \),
because \( \mathcal{H}_D \) could have any finite dimension.
Nevertheless, \( \Cpt(\mathcal{H}_D) \)
is always Morita equivalent to \( \C \), and a Morita equivalence induces
an isomorphism in K-theory. Thus the induced K-theory maps of
\eqref{eq:Willett_isomorphism_twisted} are independent of the \( X \)-\(
\Gamma \)-module, as shown in \cref{prop:equiv_Roe_decomp_K_canonical}.

Therefore, under \refbasicsetup, the K-theory of the 
equivariant Roe algebra \(\Cst(X,\mathcal{H})^\Gamma \) is
a coarse-geometric invariant of \( \Gamma \) together with 
the cohomology class of
$\sigma$. If we further assume that 
\( \mathcal{H} \) is ample, then the isomorphism class of the
equivariant Roe algebra is already an invariant.
\end{remark}

\section{Symmetry breaking}\label{sec:symmetry_breaking}
A central idea of the present work is to describe the symmetry breaking
process as a direct colimit, with its universal property understood as
consistency conditions.
The relevant categorical terminology is recalled in
\cref{app:categorical_terminology}.
The main result of this section is \cref{thm:main_computation_theorem}, stating that the isomorphisms \eqref{eq:nat.trans.symm.break} give a natural identification between two pictures of symmetry breaking --- the equivariant Roe algebra picture \eqref{eq:iota_HG} and the group $\Cst$-algebra picture \eqref{eq:phi_HG}.

\subsection{Symmetry breaking Roe algebra}
\begin{definition}
Let $\Gamma$ be a discrete countable group. The \emph{symmetry breaking
poset} for $\Gamma$ is the poset \( (\mathfrak{S},\prec) \), where 
\[ 
    \mathfrak{S}=\set{\text{finite index subgroups of}\;\Gamma}
\]
and \( G\prec H \) iff \( H \) is a subgroup of \( G \).    
\end{definition}

Note that if $G,H$ are finite index subgroups of $\Gamma$, then so is
$G\cap H$, thus $\mathfrak{S}$ is directed. 

An (ample) $X$-$\Gamma$-module \( \mathcal{H} \) is also an (ample)
$X$-$G$-module for any subgroup $G\subseteq \Gamma$. Thus, whenever $G\prec H$ in $\mathfrak{S}$, there is an obvious inclusion homomorphism
\begin{equation}\label{eq:iota_HG}
\iota_{H\backslash G}\colon\Cst(X,\mathcal{H})^{G}
\hookrightarrow\Cst(X,\mathcal{H})^{H},
\end{equation} 
given by regarding a $G$-invariant operator as
a $H$-invariant operator.

\begin{definition}\label{DfnSymmBreakingAlg}
Let $X$ be a proper metric space $X$, and \( \mathcal{H} \) be an \( X
\)-\( \Gamma \)-module.  
The \emph{symmetry breaking Roe algebra} of \( \mathcal{H} \), 
denoted $\Cst(X,\mathcal{H})^\mathfrak{S}$,
is defined to be the colimit of the directed system of \Cst-algebras whose
connecting morphisms are the inclusions \eqref{eq:iota_HG}.
\end{definition}

The symmetry breaking Roe algebra is concretely realized as
\begin{equation}\label{eq:concrete.symm.break.alg}
\Cst(X,\mathcal{H})^\mathfrak{S}\defeq
\colim_{G\in\mathfrak{S}}\Cst(X,\mathcal{H})^{G} \simeq
\overline{\bigcup_{G\in\mathfrak{S}}\Cst(X,\mathcal{H})^{G}},
\end{equation}
and is equipped with the universal inclusions
\begin{equation}\label{eq:universal.inclusion.to.symm.break}
\iota_{\mathfrak{S}\backslash G}\colon \Cst(X,\mathcal{H})^{G}\hookrightarrow \Cst(X,\mathcal{H})^\mathfrak{S},\qquad G\in\mathfrak{S}.
\end{equation}

\subsection{Symmetry breaking with group \Cst-algebras}
We would like to explicitly compute the induced K-theory maps of \(
\iota_{H\backslash G} \) as in \eqref{eq:iota_HG}. This will be done by
passing to certain homomorphisms between the $\K$-theory of group
$\Cst$-algebras, which will, moreover, form a well-defined directed system. 

A 2-cocycle $\sigma$ on $G$ restricts to a 2-cocycle on any subgroup
$H\subseteq G$. We continue to denote the latter 2-cocycle on $H$ with the
same symbol $\sigma$. 

Write \( H\backslash G \) for the 
set of right cosets.
A \emph{section} of the quotient map \( G\to H\backslash G \) is a map \( s\colon H\backslash G\to G \) such that $Hs(Hg)=Hg$ for all cosets $Hg\in H\backslash G$.
With respect to $s$, there is a bijection
\begin{equation*}
\begin{aligned}
G & \xlongrightarrow{\sim} H\times H\backslash G\\
g & \longmapsto \big(g s(Hg)^{-1}, Hg\big)
\end{aligned}
\end{equation*}
and a corresponding unitary isomorphism,
\begin{equation}\label{eq:unitary_l2G_l2H}
\begin{aligned}
\eta_s^\sigma\colon \ell^2(G) &\longrightarrow \ell^2(H)\otimes \ell^2(H\backslash
G)\\
\delta_{g}&\longmapsto \rho_G^\sigma(s(Hg))(\delta_g)\otimes \delta_{Hg}\\
\rho_G^\sigma(s(Hg))^*(\delta_h)&\longmapsfrom \delta_{h}\otimes \delta_{Hg}.
\end{aligned}
\end{equation}

\begin{proposition}\label{prop:i_GH}
Let $H\subseteq G$ be a finite index subgroup of a discrete countable group
$G$, and $\sigma$ be a 2-cocycle for $G$.
\begin{enumerate}
\item For each section
$s\colon H\backslash G\to G$, conjugation by the unitary \(
\eta_s^{\sigma} \) in \eqref{eq:unitary_l2G_l2H} restricts to an injective \st-homomorphism
\[
\Ad_{\eta_{s}^\sigma}\colon\Cstr(G,\sigma)\hookrightarrow\Cstr(H,\sigma)\otimes \Cpt(\ell^2(H\backslash G)),
\]
\item The induced \K-theory map (referred to as a \emph{symmetry breaking} map)
\begin{equation}\label{eq:phi_HG}
\phi_{H\backslash G}^{s,\sigma}\defeq\K_i(\Ad_{\eta_s^\sigma})\colon
\K_i(\Cstr(G,\sigma))\to\K_i(\Cstr(H,\sigma))
\end{equation}
is independent of the section \( s\colon H\backslash G\to G \).

\end{enumerate}
\end{proposition}

\begin{remark}\label{rem:symm.break.stabilization}
We note that since \( H \) has finite index in \( G \),   
\[
  \Cpt(\ell^2(H\backslash G))=\Bdd(\ell^2(H\backslash
  G))\simeq\Mat_{[G:H]}(\C). 
\]
In \eqref{eq:phi_HG}, we implicitly identify \( \K_i(\Cstr(H,\sigma)) \) with \(
\K_i\big(\Cstr(H,\sigma)\otimes\Cpt(\ell^2(H\backslash G))\big) \) using any corner
embedding \( \C\hookrightarrow \Cpt(\ell^2(H\backslash G)) \).
\end{remark}

\begin{proof}
Fix $g\in G$. We use the formulae \eqref{eq:unitary_l2G_l2H} and \eqref{eq:twisted.inverse.formula} to determine what $\Ad_{\eta_s^\sigma}(\rho_G^\sigma(g))$ is.
For any $h\in H$ and coset $i\in H\backslash G$, we have
\begin{equation}\label{eq:Ad_on_regular_translation_calculation}
\begin{aligned}
&\eta_s^\sigma\circ \rho_G^\sigma(g)\circ(\eta_s^\sigma)^*(\delta_h\otimes\delta_{i})\\
=&\eta_s^\sigma\circ \rho_G^\sigma(g)\circ\rho_G^\sigma(s(i))^*(\delta_{h})\\
=&\sigma(g,s(i)^{-1})\cdot\overline{\sigma(s(i)^{-1},s(i))}\cdot \eta_s^\sigma\circ\rho_G^\sigma(gs(i)^{-1})(\delta_h)\\
=&\overline{\sigma(gs(i)^{-1},s(i))}\cdot\rho_G^\sigma(s(Hhs(i)g^{-1}))\circ\rho_G^\sigma(gs(i)^{-1})(\delta_h)\otimes\delta_{Hhs(i)g^{-1}}\\
=&\overline{\sigma(gs(i)^{-1},s(i))}\cdot\sigma\big(s(Hs(i)g^{-1}),gs(i)^{-1}\big)\\
&\qquad\qquad\qquad \cdot\rho_G^\sigma\big(s(Hs(i)g^{-1})gs(i)^{-1}\big)(\delta_h)\otimes\delta_{Hs(i)g^{-1}}\\
=&\colon\epsilon^{s,\sigma}_{g;i}\cdot \rho_G^\sigma(h^s_{g;i})(\delta_h)\otimes\delta_{Hs(i)g^{-1}},
\end{aligned}
\end{equation}
where in the last line, we collected all the phase factors into a single
term $\epsilon^{s,\sigma}_{g;i}$,
and also abbreviated $s(Hs(i)g^{-1})gs(i)^{-1}$ to $h^s_{g;i}$.
Importantly, $\epsilon^{s,\sigma}_{g;i}$ is $h$-independent, and
$h^s_{g;i}\in H$. Therefore, the result of
\eqref{eq:Ad_on_regular_translation_calculation} can be written as the operator equality
\begin{equation}\label{eq:Ad_on_regular_translation}
\Ad_{\eta_{s}^\sigma}\big(\rho_G^\sigma(g)\big)=\sum_{i\in H\backslash G}
\rho_H^\sigma(h^s_{g;i})\otimes
\ket{\epsilon^{s,\sigma}_{g;i}\cdot\delta_{Hs(i) g^{-1}}}\bra{\delta_i},
\end{equation}
which is a finite sum of terms in
$\Cstr(H,\sigma)\otimes\Cpt(\ell^2(H\backslash G))$. Since $g\in G$ was
arbitrary, and since \( \Cstr(G,\sigma) \) is generated by \(
\rho_G^\sigma(g) \) for all 
\( g\in G \), we conclude that $\Ad_{\eta_{s}^\sigma}$ maps
\( \Cstr(G,\sigma) \) into \(
\Cstr(H,\sigma)\otimes\Cpt(\ell^2(H\backslash G)) \).

Now let \( t\colon H\backslash G\to G \) be another section. We
use \eqref{eq:unitary_l2G_l2H} and \eqref{eq:twisted.inverse.formula} to
compute the unitary isomorphism
\[ 
  \eta_{t,s}^{\sigma}\defeq \eta_{t}^\sigma\circ(\eta_{s}^\sigma)^{-1}\colon
  \ell^2(G)\otimes\ell^2(H\backslash
  G)\to\ell^2(G)\otimes\ell^2(H\backslash G).
\]
Similarly to \eqref{eq:Ad_on_regular_translation_calculation}, we have
  (note that \(
t(i)s(i)^{-1}\in H \) for any \( i\in H\backslash G \)):
\begin{align*}
\eta_{t,s}^{\sigma}(\delta_h\otimes\delta_i)
=&\eta_{t}^{\sigma}\circ\rho_G^\sigma(s(i))^*(\delta_h) \\
=&\overline{\sigma(s(i)^{-1},s(i))}\cdot\sigma(t(i),s(i)^{-1})
\cdot\rho_G^\sigma(t(i)s(i)^{-1})(\delta_h)\otimes \delta_i\\
=&\colon\omega_i^{s,t,\sigma}\cdot\rho_H^\sigma(t(i)s(i)^{-1})(\delta_h)\otimes\delta_i
\end{align*}
where the phase factor \( \omega_{i}^{s,t,\sigma} \) is independent of \(
h \). It follows that
\[ 
  \eta^{\sigma}_{t,s}=\sum_{i\in H\backslash G}
  \rho_H^\sigma(t(i)s(i)^{-1})\otimes\ket{\omega_{i}^{s,t,\sigma}\cdot\delta_i}\bra{\delta_i}
\]
belongs to $\Cstr(H,\sigma)\otimes
\Cpt(\ell^2(H\backslash G))$.
Therefore, the conjugation \( \Ad_{\eta^{\sigma}_{t,s}} \) induces the identity map on
\( \K_i(\Cstr(H,\sigma)) \). Then
\[ 
  \underbrace{\K_i(\Ad_{\eta_{t}^\sigma})}_{\phi^{t,\sigma}_{H\backslash G}}=
  \underbrace{\K_i(\Ad_{\eta^{\sigma}_{t,s}})}_{\id}\circ
  \underbrace{\K_i(\Ad_{\eta_{s}^\sigma})}_{\phi^{s,\sigma}_{H\backslash
G}}.\qedhere
\]
\end{proof}

\begin{corollary}\label{cor:exists.equivalent.directed.system}
Let \( \Gamma \) be a discrete countable group, \( \sigma\colon
\Gamma\times\Gamma\to\mathrm{U}(1) \) a
2-cocycle for \( \Gamma \), and \( \mathfrak{S} \) its symmetry breaking
poset. Then there is a directed system of abelian groups \(
\{\K_i(\Cstr(G,\sigma))\}_{G\in\mathfrak{S}} \), whose connecting
morphisms are given by the symmetry breaking maps \eqref{eq:phi_HG}.
\end{corollary}

\begin{proof}
We must show that the symmetry breaking maps \(
\phi^{s,\sigma}_{H\backslash G} \) satisfy
\begin{equation}\label{eq:composition.of.symm.break}
  \phi^{t,\sigma}_{K\backslash H}\circ\phi^{s,\sigma}_{H\backslash G}=
  \phi^{r,\sigma}_{K\backslash G},
\end{equation}
where, by \cref{prop:i_GH}, we may use any choices of sections
\[ 
  s\colon H\backslash G\to G,\quad t\colon K\backslash H\to H,\quad r\colon
  K\backslash G\to G.
\]
for the maps in \eqref{eq:composition.of.symm.break}.
In particular, given sections $s,t$, we may choose $r$ to be
\begin{equation}\label{eq:iterated_section}
r\defeq t\star s\colon K\backslash G \to G,\quad Kg \mapsto t\big(Kg
s(Hg)^{-1}\big)s(Hg).
\end{equation}
To see that $r=t\star s$ is indeed a section, first note that for any $g\in G$, we have $gs(Hg)^{-1}\in H$, so the formula
\eqref{eq:iterated_section} is well-defined. Then, as
$t\colon H\to K\backslash H$ is a section, we may write $t\big(Kg s(Hg)^{-1}\big)=k g s(Hg)^{-1}$ for some $k\in K$, so
\[
K(t\star s)(Kg)= Kt\big(Kg s(Hg)^{-1}\big)s(Hg)=Kk g s(Hg)^{-1} s(Hg)=Kg.
\]

Next, let us construct a unitary map $\mu_{t,s}^\sigma\colon
\ell^2(K\backslash G)\overset{\sim}{\to}\ell^2(K\backslash
H)\otimes\ell^2(H\backslash G)$ such that
\begin{equation}\label{eq:iterated_section_iso}
\begin{tikzcd}[column sep=1.5cm]
  \ell^2(G)\ar[r,"\eta_s^\sigma"]\ar[drr,"\eta_{t\star s}^\sigma", bend right=10] & \ell^2(H)\otimes \ell^2(H\backslash G) \ar[r,"\eta_t^\sigma\otimes 1_{\ell^2(H\backslash G)}"] & \ell^2(K)\otimes\ell^2(K\backslash H)\otimes\ell^2(H\backslash G)\\
  & &\ell^2(K)\otimes \ell^2(K\backslash G)\ar[u,"1_{\ell^2(K)}\otimes\mu_{t,s}^\sigma"]
\end{tikzcd}
\end{equation}
commutes. For the horizontal composition in
\eqref{eq:iterated_section_iso}, we have, for each $g\in G$,
\begin{equation}\label{eq:iterated_coset_computation}
\begin{aligned}
&\big(\eta_t^\sigma\otimes 1_{\ell^2(H\backslash G)}\big)\circ \eta_s^\sigma(\delta_g)\\
=&\big(\eta_t^\sigma\otimes 1_{\ell^2(H\backslash G)}\big)\big(\rho_G^\sigma(s(Hg))(\delta_g)\otimes\delta_{Hg}\big)\\
=&\rho_H^\sigma(t(Kgs(Hg)^{-1}))\circ\rho_G^\sigma(s(Hg))(\delta_g)\otimes \delta_{Kgs(Hg)^{-1}}\otimes\delta_{Hg}\\
=&\sigma(t(Kgs(Hg)^{-1}),s(Hg))\cdot\rho_G^\sigma(t(Kgs(Hg)^{-1})s(Hg))(\delta_g)\otimes \delta_{Kgs(Hg)^{-1}}\otimes\delta_{Hg}\\
=&\sigma(t(Kgs(Hg)^{-1}),s(Hg))\cdot\rho_G^\sigma(t\star s(Kg))(\delta_g)\otimes \delta_{Kgs(Hg)^{-1}}\otimes\delta_{Hg}.
\end{aligned}
\end{equation}
For the other composition in \eqref{eq:iterated_coset_computation}, we have
\begin{equation}\label{eq:iterated_coset_computation.2}
\eta_{t\star s}^\sigma(\delta_g)=\rho_G^\sigma(t\star s(Kg))(\delta_g)\otimes \delta_{Kg}.
\end{equation}
To compare \eqref{eq:iterated_coset_computation} with \eqref{eq:iterated_coset_computation.2}, we define the unitary
\begin{equation*}\label{eq:iterated_coset_space_bijection}
\begin{aligned}
\mu_s\colon \ell^2(K\backslash G)&\overset{\sim}{\to}\ell^2(K\backslash H)\otimes\ell^2(H\backslash G)\\
\delta_{Kg}&\mapsto \delta_{Kgs(Hg)^{-1}}\otimes \delta_{Hg}
\end{aligned}
\end{equation*}
and compose it with the multiplication operator by the following phase function
\begin{equation*}\label{eq:extra.phase.function}
Kg\mapsto \sigma(t(Kgs(Hg)^{-1}),s(Hg)),
\end{equation*}
to obtain a unitary $\mu_{t,s}^\sigma\colon \ell^2(K\backslash
G)\overset{\sim}{\to}\ell^2(K\backslash H)\otimes\ell^2(H\backslash G)$.
Then it is clear, by construction, that \eqref{eq:iterated_coset_computation.2} and
\eqref{eq:iterated_coset_computation} are related by the commutative diagram
\eqref{eq:iterated_section_iso}. It follows that 
\[ 
\underbrace{\K_i(\Ad_{1_{\ell^2(K)}\otimes\mu_{t,s}^\sigma})}_{\id}\circ
  \underbrace{\K_i(\Ad_{\eta_{t\star
  s}^\sigma})}_{\phi^{t\star s,\sigma}_{K\backslash G}}=
  \underbrace{\K_i(\Ad_{\eta_{t}^\sigma\otimes 1_{\ell^2(H\backslash
  G)}})}_{\phi^{t,\sigma}_{K\backslash
  H}}\circ\underbrace{\K_i(\Ad_{\eta_s^\sigma})}_{\phi^{s,\sigma}_{H\backslash
G}}.
\]
Above, we implicitly use the stabilization isomorphisms (\cref{rem:symm.break.stabilization}) to identify the $\K_i(\Ad_{(-)})$ with their respective symmetry-breaking maps. Under these identifications, we also have $\K_i(\Ad_{1_{\ell^2(K)}\otimes\mu_{t,s}^\sigma})=\id$ because of \eqref{eq:isometry.stabilization}.
\end{proof}

\begin{notation}
Subsequently, we will just write
$\phi^\sigma_{H\backslash G}$ instead of $\phi^{s,\sigma}_{H\backslash G}$.
We will also often drop the superscript $\sigma$ to simplify notation.
\end{notation}

\subsection{Natural isomorphism of symmetry breaking directed systems}
Let $D$ be a fundamental domain for an $X$-$G$-module $\mathcal{H}$ in
the sense of \cref{def:fundamental_domain}. It follows that for any distinct $g, g^\prime \in G$,
\begin{equation}\label{eq:sum_disjoint}
\begin{aligned}
\varrho(\chi_{gD\cup g^\prime D})&=\varrho\left(\chi_{gD}+\chi_{g^\prime D}-\chi_{gD\cap g^\prime D}\right)\\
&=\varrho\left(\chi_{gD}+\chi_{g^\prime D}\right)-\varrho(\chi_{gD})\varrho(\chi_{g^\prime D})\\
&=\varrho\left(\chi_{gD}+\chi_{g^\prime D}\right).
\end{aligned}
\end{equation}
Now let $H\subseteq G$ be a subgroup, 
and $s\colon H\backslash G\to G$ be a section, then
\begin{align*}
1_{\mathcal{H}} = \sum_{g\in G}\varrho(\chi_{gD})
=\sum_{h\in H}\sum_{i\in H\backslash G}\varrho(\chi_{hs(i)D})
&=\sum_{h\in H}\varrho\left(\sum_{i\in H\backslash G}\chi_{hs(i)D}\right)\\
&=\sum_{h\in H}\varrho\Big(\chi_{\bigcup_{i\in H\backslash G} hs(i)D}\Big) & ({\rm By}\;\eqref{eq:sum_disjoint})\\
&= \sum_{h\in H}\varrho\Big(\chi_{h\bigcup_{i\in H\backslash G}s(i)D}\Big).
\end{align*}
This says that
\begin{equation}\label{eq:section.act.on.fund.domain}
s\cdot D\defeq \bigcup_{i\in H\backslash G} s(i) D
\end{equation}
is a fundamental domain for $\mathcal{H}$ regarded as an $X$-$H$-module.  
Consequently, \eqref{eq:V_D} also gives a unitary
\[
V_{s\cdot D}\colon \mathcal{H}\xrightarrow{\sim} \ell^2(H)\otimes \mathcal{H}_{s\cdot D}.
\]
Also recall that for each coset $i\in H\backslash G$, $\mathcal{H}_{s(i)D}$
is identified with $\mathcal{H}_D$ via $U_{s(i)}$, so there is a
unitary isomorphism
\[
W_{s\cdot D}\colon\mathcal{H}_{s\cdot D}\xrightarrow{\sim} 
\ell^2(H\backslash G)\otimes \mathcal{H}_D.
\]

\begin{lemma}\label{lem:VsD.related.to.eta.s}
Let $D$ be a fundamental domain for a \( \sigma \)-twisted
$X$-$G$-module $\mathcal{H}$. Let
$H\subseteq G$ be a finite index subgroup, and $s\colon H\backslash G\to G$
be a section. Then there is a commutative diagram, whose arrows are all
unitary isomorphisms of Hilbert spaces:
\begin{equation}\label{eq:VD.compatible.with.eta}
\begin{tikzcd}[column sep=large]
\mathcal{H} \arrow[r, "1_\mathcal{H}", "\sim"'] \arrow[dd, "V_D"', "\sim"] &
\mathcal{H} \arrow[d, "V_{s\cdot D}", "\sim"']  \\
& 
\ell^2(H)\otimes\mathcal{H}_{s\cdot D} \arrow[d, "1_{\ell^2(H)}\otimes
W_{s\cdot D}", "\sim"'] \\
\ell^2(G)\otimes\mathcal{H}_D \arrow[r, "\eta_s^{\overline{\sigma}}\otimes
1_{\mathcal{H}_D}", "\sim"'] &
\ell^2(H)\otimes\ell^2(H\backslash G)\otimes\mathcal{H}_D.
\end{tikzcd} 
\end{equation}
\end{lemma}
\begin{proof}
The formula \eqref{eq:V_D} applied to (the inverses of) 
$V_{s\cdot D}$ and \( W_{s\cdot D} \) 
gives
\[ 
\begin{tikzcd}[row sep=0, column sep=large]
\mathcal{H} \arrow[r, "V_{s\cdot D}"] &
\ell^2(H)\otimes\mathcal{H}_{s\cdot D} 
\arrow[r, "1_{\ell^2(H)}\otimes W_{s\cdot D}"] & 
\ell^2(H)\otimes\ell^2(H\backslash G)\otimes \mathcal{H}_D \\
U_hU_{s(Hg)}\zeta & \delta_h\otimes U_{s(Hg)}\zeta \arrow[l, maps to] &
\delta_h\otimes \delta_{Hg}\otimes\zeta \arrow[l, maps to]. 
\end{tikzcd}
\]
At the same time, formula \eqref{eq:unitary_l2G_l2H} for $\eta^{\overline{\sigma}}_s$ and formula \eqref{eq:V_D} for $V_D$ give
\begin{equation*}
\begin{tikzcd}[row sep=0]
      \mathcal{H} \ar[r,"V_{D}"]  &  \ell^2(G)\otimes \mathcal{H}_{D} \ar[r,"\eta^{\overline{\sigma}}_{s}\otimes 1_{\mathcal{H}_D}"] & \ell^2(H)\otimes\ell^2(H\backslash G)\otimes\mathcal{H}_{D}  \\
 \sigma(h,s(Hg))U_{hs(Hg)}\zeta & \sigma(h,s(Hg))\delta_{hs(Hg)}\otimes\zeta \ar[l,mapsto] & \delta_h\otimes\delta_{Hg}\otimes\zeta \ar[l,mapsto] .
\end{tikzcd}
\end{equation*}
Since $U$ is a $\sigma$-twisted representation of $G$, the above two formulas verify that the diagram
\eqref{eq:VD.compatible.with.eta} commutes.
\end{proof}

\begin{corollary}\label{cor:VsD.related.to.eta.s}
Let \( \mathcal{H} \) be a \( \sigma \)-twisted \( X \)-\( G \)-module
satisfying \refbasicsetup, and let \( D \) be a bounded fundamental domain for
\( \mathcal{H} \).
Let $H\subseteq G$ be a finite index subgroup, and $s\colon H\backslash
G\to G$ be a section. Then the following diagram commutes:
\begin{equation}\label{eq:transferring_symmetry_breaking_maps}
\begin{tikzcd}[column sep=large]
\Cst(X,\mathcal{H})^G \arrow[r, "\iota_{H\backslash G}"] 
\arrow[dd, "\Ad_{V_D}"', "\sim"] &
\Cst(X,\mathcal{H})^H \arrow[d, "\Ad_{V_{s\cdot D}}", "\sim"']  \\
& 
\Cstr(H,\overline{\sigma})\otimes\Cpt(\mathcal{H}_{s\cdot D}) \arrow[d, "\Ad_{1_{\ell^2(H)}\otimes W_{s\cdot D}}", "\sim"'] \\
\Cstr(G,\overline{\sigma})\otimes\Cpt(\mathcal{H}_D) \arrow[r, "\Ad_{\eta_s^{\overline{\sigma}}\otimes
1_{\mathcal{H}_D}}"] &
\Cstr(H,\overline{\sigma})\otimes\Cpt(\ell^2(H\backslash G))\otimes\Cpt(\mathcal{H}_D).
\end{tikzcd} 
\end{equation}
\end{corollary}

\begin{proof}
By \cref{prop:twisted_Willett}, the conjugations \( \Ad_{V_D} \) and \( \Ad_{V_{s\cdot D}} \) restrict to isomorphisms
\[ 
\Cst(X,\mathcal{H})^G\xrightarrow{\sim}\Cstr(G,\overline{\sigma})\otimes\Cpt(\mathcal{H}_D),\quad
\Cst(X,\mathcal{H})^H\xrightarrow{\sim}\Cstr(H,\overline{\sigma})\otimes\Cpt(\mathcal{H}_{s\cdot
D}).
\]
It follows from \cref{prop:i_GH} that 
the conjugation \( \Ad_{\eta_{s}^{\overline{\sigma}}\otimes1_{\mathcal{H}_D}} \) maps \(
\Cstr(G,\overline{\sigma})\otimes\Cpt(\mathcal{H}_D)\) into \(
\Cstr(H,\overline{\sigma})\otimes\Cpt(\ell^2(H\backslash G))\otimes\Cpt(\mathcal{H}_D) \).
Then commutativity of the diagram \eqref{eq:transferring_symmetry_breaking_maps} follows from the commutativity of
\eqref{eq:VD.compatible.with.eta}.  
\end{proof}

\begin{theorem}\label{thm:main_computation_theorem}
Let \( \mathcal{H} \) be a
$\sigma$-twisted \( X \)-\( \Gamma \)-module as in the
\refbasicsetup, and \( \mathfrak{S} \) be the symmetry breaking poset
of \( \Gamma \). Then the abelian group isomorphisms 
of \eqref{eq:nat.trans.symm.break},
\[
  \{\alpha_G\colon
  \K_i(\Cst(X,\mathcal{H})^G)\xrightarrow{\sim}\K_i(\Cstr(G,\overline{\sigma}))\}_{G\in\mathfrak{S}},
\]
yield a natural isomorphism of the following
directed systems of abelian groups:
\begin{itemize}
\item The directed system \(
\{\K_i(\Cst(X,\mathcal{H})^G)\}_{G\in\mathfrak{S}} \) whose
connecting morphisms \( \K_i(\iota_{H\backslash G}) \) are induced by the
inclusion \st-homomorphisms \( \iota_{H\backslash G} \) \textup{(}see
\eqref{eq:iota_HG}\textup{)}. 
\item The directed system \( \{\K_i(\Cst(G,\overline{\sigma}))\}_{G\in\mathfrak{S}} \)
whose connecting morphisms are the symmetry breaking maps
\( \phi^{\overline{\sigma}}_{H\backslash G} \) \textup{(}see \eqref{eq:phi_HG}\textup{)}.
\end{itemize}
Consequently,
\begin{equation}\label{eq:inductive_limit_substitution}
\K_i\big(\Cst(X,\mathcal{H})^{\mathfrak{S}}\big)\simeq\colim_{G\in\mathfrak{S}}\K_i(\Cst(X,\mathcal{H})^G)\simeq
\colim_{G\in\mathfrak{S}}\K_i(\Cstr(G,\overline{\sigma})).
\end{equation}

\end{theorem}

\begin{proof}
Let \( G,H\in\mathfrak{S} \) with \( G\prec H \). Pick any section $s\colon H\backslash G\to G$, and any bounded fundamental domain $D$ for the $X$-$G$-module $\mathcal{H}$.
We claim that the following diagram commutes:
\begin{equation}\label{eq:transferring_symmetry_breaking_maps_K-theory}
\begin{tikzcd}[row sep=large, column sep=huge]
\K_i(\Cst(X,\mathcal{H})^G) \arrow[r, "\K_i(\iota_{H\backslash G})"]
\arrow[d, "\K_i(\Ad_{V_D})"', "\sim"] \arrow[dd, bend right=90,
"\alpha_G"'] &
\K_i(\Cst(X,\mathcal{H})^H)
\arrow[d, "\K_i\big(\Ad_{1_{\ell^2(H)\otimes W_{s\cdot D}}}\big)\circ \K_i(\Ad_{V_{s\cdot
D}})"', "\sim"] \arrow[dd, bend left=100, "\alpha_H"] \\ 
\K_i(\Cstr(G,\overline{\sigma})\otimes\Cpt(\mathcal{H}_D)) \arrow[r, "\K_i(\Ad_{\eta_s^{\overline{\sigma}}\otimes
1_{\mathcal{H}_D}})", "\sim"'] &
\K_i\left(\begin{matrix}
\Cstr(H,\overline{\sigma})\otimes\Cpt(\ell^2(H\backslash G))\\
\otimes\Cpt(\mathcal{H}_D)
\end{matrix}\right)\\
\K_i(\Cstr(G,\overline{\sigma})) \arrow[u, "\Stab_{\mathcal{H}_D}", "\sim"'] \arrow[r, "\K_i(\Ad_{\eta_s^{\overline{\sigma}}})"] &
\K_i(\Cstr(H,\overline{\sigma})) \arrow[u, "\Stab_{\mathcal{H}_D}", "\sim"'].
\end{tikzcd} 
\end{equation}
Commutativity of the top square follows from
commutativity of the diagram
\eqref{eq:transferring_symmetry_breaking_maps} of \cref{lem:VsD.related.to.eta.s}. The stablization map \(
\Stab_{\mathcal{H}_D} \) is induced by the \st-homomorphism 
\( A\to A\otimes\Cpt(\mathcal{H}_D) \) for either \( A=\Cstr(G,\overline{\sigma}) \) or \(
\Cstr(H,\overline{\sigma})\otimes\Cpt(\ell^2(H\backslash G)) \),
sending \( a\mapsto a\otimes p \)
for any rank-one projection \( p\in\Cpt(\mathcal{H}_D) \), thus the bottom
square commutes as well. (As usual, \( \K_i(\Cstr(H,\overline{\sigma})) \) is identified with \(
\K_i\big(\Cstr(H,\overline{\sigma})\otimes\Cpt(\ell^2(H\backslash G))\big) \), see \cref{rem:symm.break.stabilization}). 
The commutativity of
\eqref{eq:transferring_symmetry_breaking_maps_K-theory} now follows from
\begin{align*}
\alpha_G&=\Stab_{\mathcal{H}_D}^{-1}\circ\K_i(\Ad_{V_D}), \\
\alpha_H&=\Stab_{\mathcal{H}_{s\cdot D}}^{-1}\circ\K_i(\Ad_{V_{s\cdot D}})=\Stab_{\mathcal{H}_D}^{-1}\circ \K_i\big(\Ad_{1_{\ell^2(H)\otimes W_{s\cdot D}}}\big) \circ\K_i(\Ad_{V_{s\cdot D}}),
\end{align*}
see \eqref{eq:isometry.stabilization}--\eqref{eq:nat.trans.symm.break} and
\cref{prop:equiv_Roe_decomp_K_canonical}.

Now, the bottom horizontal map of
\eqref{eq:transferring_symmetry_breaking_maps_K-theory} is, by definition, the
symmetry breaking map $\phi^{\overline{\sigma}}_{H\backslash G}$, see
\eqref{eq:phi_HG} of \cref{prop:i_GH}. Thus, the commutativity of the outer part of \eqref{eq:transferring_symmetry_breaking_maps_K-theory} is precisely the statement that $\phi^{\overline{\sigma}}_{H\backslash G}$ is naturally
identified with $\K_i(\iota_{H\backslash G})$.
\end{proof}

As the following proposition will show, the
directed system of abelian groups \(
\{\K_i(\Cstr(G,\sigma))\}_{G\in\mathfrak{S}} \) depends, up to isomorphism,
only on the group cohomology class of the 2-cocycle \( \sigma \). This
allows us to modify the 2-cocycle \( \sigma \) to a more convenient one,
for the computation of
$\colim_{G\in\mathfrak{S}}\K_i(\Cstr(G,\overline{\sigma}))$ on the right
side of \eqref{eq:inductive_limit_substitution}.

\begin{proposition}\label{prop:cohomologous_symm_breaking}
Let $\sigma,\sigma^\prime$ be cohomologous 2-cocycles on a discrete
countable group $\Gamma$. Let $\mathfrak{S}$ be the symmetry breaking poset
for $\Gamma$. Then there is a natural isomorphism between the
symmetry breaking directed systems whose connecting morphisms are given by
\[
\begin{aligned}
    \phi_{H\backslash G}^\sigma\colon \K_i(\Cstr(G,\sigma))& \to \K_i(\Cstr(H,\sigma)),\\
    \phi_{H\backslash G}^{\sigma^\prime}\colon
    \K_i(\Cstr(G,\sigma^\prime))& \to \K_i(\Cstr(H,\sigma^\prime)),
\end{aligned}
\]
thus an induced isomorphism between their colimits.
\end{proposition}
\begin{proof}
Let $\omega\colon \Gamma\to\mathrm{U}(1)$ be a map implementing the equivalence $\sigma\sim\sigma^\prime$, as in \eqref{eq:cohomologous_cocycles}. For each $G\in\mathfrak{S}$, the map
\begin{equation}\label{eq:cohom_iso_algebra}\begin{aligned}
    \Omega_G\colon \Cstr(G,\sigma)&\to\Cstr(G,\sigma^\prime)\\
    \lambda^\sigma_{G}(g) & \mapsto \overline{\omega(g)}\cdot\lambda^{\sigma^\prime}_{G}(g),\qquad g\in G,
\end{aligned}
\end{equation}
is an isomorphism of \Cst-algebras. 
For \( G,H\in\mathfrak{S} \) with \( G\prec H \),
pick any section $s\colon H\backslash G\to G$, and define the unitary
\begin{align*}
W_s\colon \ell^2(H\backslash G) &\to\ell^2(H\backslash G)\\
\delta_{i} &\mapsto \omega(s(i))\cdot\delta_{i}.
\end{align*}
Using the formula \eqref{eq:Ad_on_regular_translation} for $\Ad_{\eta_s^\sigma}$, we may check that the following diagram commutes,
\begin{equation*}\label{eq:cohomologous_morphism_diagram}
\begin{tikzcd}
\Cstr(G,\sigma)\ar[r,"\Ad_{\eta_s^\sigma}"]\ar[d,"\Omega_G","\sim"'] &
\Cstr(H,\sigma)\otimes \Cpt(\ell^2(H\backslash G))\ar[d,"\Omega_H\otimes
\Ad_{W_s}","\sim"']\\
\Cstr(G,\sigma^\prime)\ar[r,"\Ad_{\eta_s^{\sigma^\prime}}"] &
\Cstr(H,\sigma^\prime)\otimes\Cpt(\ell^2(H\backslash G)).
\end{tikzcd}
\end{equation*}
Applying the $\K_i$-functor, composing with stabilization isomorphisms, and
noting that $\K_i(1\otimes\Ad_{W_s})$ is the identity map, we deduce that the diagram
\begin{equation}\label{eq:cohomologous_symm_breaking}
\begin{tikzcd}
\K_i(\Cstr(G,\sigma))\ar[r,"\phi_{H\backslash G}^\sigma"]\ar[d,"\K_i(\Omega_G)","\sim"'] & \K_i(\Cstr(H,\sigma))\ar[d,"\K_i(\Omega_H)","\sim"']\\
\K_i(\Cstr(G,\sigma^\prime))\ar[r,"\phi_{H\backslash G}^{\sigma^\prime}"] & \K_i(\Cstr(H,\sigma^\prime))
\end{tikzcd}
\end{equation}
commutes for every \( G\prec H \).
So $\{\K_i(\Omega_G)\}_{G\in\mathfrak{S}}$ is a natural isomorphism (depending on $\omega$), as claimed.
\end{proof}

\section{Symmetry breaking in Euclidean space: the untwisted case}
\label{sec:Euclidean_untwisted}
A simple example of a group action \( \Gamma\curvearrowright X \)
satisfying \refbasicsetup\, is given by
 \( X \) the \( d
\)-dimensional Euclidean space, and \( \Gamma \) a lattice in ({the translation group of)} \( X \).
The goal of this section is to compute the K(O)-theory of
the symmetry breaking algebra of an \emph{untwisted}, real or complex \( X
\)-\( \Gamma \)-module. The result is summarised in
\cref{tab:K-theory_untwisted}.

The computations in both the complex and real cases follow essentially the
same strategy. We analyze the symmetry breaking maps at the group
$\Cst$-algebra level, as justified by \cref{thm:main_computation_theorem}.
In \cref{sec:symmetry breaking_1-dim} we use explicit generators to compute
the \( 1 \)-dimensional symmetry breaking maps \( \phi_{H\backslash G} \).
The general computation in dimensions \( d\geq 2 \), as in
\cref{sec:d-dim}, is done by passing to a cofinal subset of the symmetry
breaking poset, so that a version of the K\"unneth Theorem may be applied.

The technicality in the real case is that Schochet's K\"unneth Theorem for
complex \Cst-algebras \cite{Schochet:Kuenneth_formula} does not pass to
real \Cst-algebras, as the KO-theory of real \Cst-algebras typically
contains torsion elements. Instead, we apply a version of the K\"unneth Theorem, due
to Boersema \cite{Boersema:Real_Cst-algebras_Kuenneth_formula}. Such a
computation involves describing the KO-theory of real \Cst-algebras using
their \( \KO_*(\R) \)-module structure, as recalled in
\cref{sec:real_Cst-algebras}.

Computation of the real case is not only more involved, but also physically
more interesting, see \cref{sec:real_case_motivation_Kane_Mele}. As a
consequence, we shall provide a detailed computation for the real case in
\cref{sec:real_d-dim} and only sketch its complex counterpart in
\cref{sec:complex_d-dim}.

\begin{notation}
We will use subscripts $\R$ and $\C$ when we want to emphasize whether the
$\Cst$-algebra is real or complex.
\end{notation}

\begin{table}[h!]
\centering
\resizebox{\textwidth}{!}{
\begin{tabular}{c|ccccccccc}
\hline
\diagbox{\( i \)}{\( d \)} & 0 & 1 & 2 & 3 & 4 & 5 & 6 & 7 & 8 \\
\hline
\( \K_0 \)  & \( \Z \)  & \( \Q \) & \( \Q\oplus\Z \) & \( \Q^{4} \) & \( \Q^{7}\oplus\Z \) & \( \Q^{16} \) & \( \Q^{31}\oplus\Z \) & \( \Q^{64} \) & \( \Q^{127}\oplus\Z \) \\
\( \K_1 \)  & \( 0 \) & \( \Z \) & \( \Q^{2} \) & \( \Q^{3}\oplus\Z \) & \( \Q^{8} \) & \( \Q^{15}\oplus\Z \) & \( \Q^{32} \) & \( \Q^{63}\oplus\Z \) & \( \Q^{128} \) \\
\hline\hline
\diagbox{\( i \)}{\( d \)} & 0 & 1 & 2 & 3 & 4 & 5 & 6 & 7 & 8 \\
\hline
\( \KO_0 \) & \( \Z \)  & \( \Q \) & \( \Q \) & \( \Q \) & \( \Q\oplus\Z \) & \( \Q^{6}
\) & \( \Q^{16}\oplus\Z/2 \) & \( \Q^{36}\oplus\Z/2 \) & \( \Q^{71}\oplus\Z \) \\ 
\( \KO_1 \) & \( \Z/2 \)  & \( \Z \) & \( \Q^{2} \) & \( \Q^{3} \) & \( \Q^{4} \) & \( \Q^{5}\oplus\Z \) & \( \Q^{12} \) & \( \Q^{28}\oplus\Z/2 \) & \( \Q^{64}\oplus\Z/2 \) \\
\( \KO_2 \) & \( \Z/2 \) & \( \Z/2 \) & \( \Z \) & \( \Q^{3} \) & \( \Q^{6} \) & \( \Q^{10} \) & \( \Q^{15}\oplus\Z \) & \( \Q^{28} \) & \( \Q^{56}\oplus\Z/2 \) \\
\( \KO_3 \) & 0 & \( \Z/2 \) & \( \Z/2 \) & \( \Z \) & \( \Q^{4} \) & \( \Q^{10} \) & \( \Q^{20} \) & \( \Q^{35}\oplus\Z \) & \( \Q^{64} \) \\
\( \KO_4 \) & \( \Z \) & \( \Q \) & \( \Q\oplus\Z/2 \) & \( \Q\oplus\Z/2 \) & \( \Q\oplus\Z \) & \( \Q^{6} \) & \( \Q^{16} \) & \( \Q^{36} \) & \( \Q^{70}\oplus\Z \) \\
\( \KO_5 \) & 0 & \( \Z \) & \( \Q^{2} \) & \( \Q^{3}\oplus\Z/2 \) & \( \Q^{4}\oplus\Z/2 \) & \( \Q^{5}\oplus\Z \) & \( \Q^{12} \) & \( \Q^{28} \) & \( \Q^{64} \) \\
\( \KO_6 \) & 0 & 0 & \( \Z \) & \( \Q^{3} \) & \( \Q^{6}\oplus\Z/2 \) & \( \Q^{10}\oplus\Z/2 \) & \( \Q^{15}\oplus\Z \) & \( \Q^{28} \) & \( \Q^{56} \) \\
\( \KO_7 \) & 0 & 0 & 0 & \( \Z \) & \( \Q^{4} \) & \( \Q^{10}\oplus\Z/2 \) &
\( \Q^{20}\oplus\Z/2 \) & \( \Q^{35}\oplus\Z \) & \( \Q^{64} \) \\
\hline
\end{tabular}
}
\caption{K-theory of \( \Cst(\R^d)^{\mathfrak{Z}^d}_\C \) and $\KO$-theory
of \( \Cst(\R^d)^{\mathfrak{Z}^d}_\R \) for
low-dimensional Euclidean spaces $\R^d$.}
\label{tab:K-theory_untwisted}
\end{table}

\subsection{Background and motivation}
\label{sec:real_case_motivation_Kane_Mele}
Let $X$ be $d$-dimensional Euclidean space, i.e., the Riemannian manifold
underlying the additive Lie group $(W,+)$ of a $d$-dimensional inner
product space $W$. Let \( \Gamma \) be a lattice, i.e.~a discrete cocompact
subgroup of $(W,+)$. Then \( \Gamma \) is a free abelian group of rank $d$.
Upon choosing a basis, $\Gamma$ is identified with $\Z^d$, and $W$ is
identified with $\R^d$. After picking an origin, $X$ is then identified
with $\R^d$. (But note that $\R^d$ is not equipped with the metric from the
standard inner product as $\Gamma$ is not necessarily a standard square
lattice.)

An identification $\Gamma\simeq\Z^d$ induces an identification of the
symmetry breaking poset \( \mathfrak{S} \) of \( \Gamma \) with the poset
$(\mathfrak{Z}^d,\prec)$, where: 
\begin{itemize}
\item \( \mathfrak{Z}^d \) consists of \emph{all} finite-index subgroups of
\( \Z^d \); 
\item for \( G,H\in \mathfrak{S} \), we set \( G\prec H \) iff \( H
\) is a subgroup of \( G \).
\end{itemize}
Then
\[ 
  \K_i(\Cst(X,\mathcal{H})^\mathfrak{S}_\C)\simeq
  \K_i(\Cst(\R^d)^{\mathfrak{Z}^d}_\C),\quad
  \KO_i(\Cst(X,\mathcal{H})^\mathfrak{S}_\R)\simeq
  \KO_i(\Cst(\R^d)^{\mathfrak{Z}^d}_\R).
\]
 
The following two cases (see \cref{thm:KO-theory_real}) are of
particular importance in physics,
\begin{equation}\label{eq:physical.KO.cases}
  \KO_4(\Cst(\R^2)^{\mathfrak{Z}^2}_\R)\simeq\Q\oplus\Z/2,\qquad
  \KO_4(\Cst(\R^3)^{\mathfrak{Z}^3}_\R)\simeq\Q\oplus\Z/2.
\end{equation}
In the first case, the \( \Z/2 \)-summand is the ``symmetry breaking colimit'' of the \( \Z/2 \)
topological invariant for 2-dimensional Quantum spin Hall systems introduced by physicists in
\cite{Kane-Mele:Quantum_spin_Hall_effect}. In the second case, the \( \Z/2 \)-summand 
is the ``symmetry breaking colimit'' of the ``strong'' 3-dimensional
topological insulator invariant 
\cites{Fu-Kane:Topological_insulators_with_inversion_symmetry,Fu-Kane-Mele:TI_in_3D}.
This should be contrasted with KO-theory of the equivariant Roe algebra,
\begin{equation}\label{eq:3D.TI.all.invariants}
\KO_4(\Cst(\R^3)^{\Z^3}_\R)\simeq \Z\oplus (\Z/2)^3\oplus\Z/2,
\end{equation}
where there are \emph{four} copies of $\Z/2$. We will also see that the $\Z$-summand in \eqref{eq:3D.TI.all.invariants} gets mapped into the $\Q$-summand in \eqref{eq:physical.KO.cases}.
The four $\Z/2$ invariants in \eqref{eq:3D.TI.all.invariants} were introduced in
\citelist{\cite{Moore-Balents:Topological_invariants_TRS}
\cite{Fu-Kane-Mele:TI_in_3D}}, and three of them are usually called ``weak'' invariants. We refer to
\cites{dNittis-Gomi:Classification_quaternionic,Gomi-Thiang:Real_gerbes}
for a precise presentation in terms of \KR-theory of the involutive
$d$-torus $\widehat{\Z}^3$. 

The degree $4$ appearing in
\eqref{eq:physical.KO.cases}, \eqref{eq:3D.TI.all.invariants} is due to the
presence of a \emph{fermionic time-reversal symmetry} $\mathcal{T}$,
i.e.~an anti-unitary operator such that \( \mathcal{T}^2=-1
\).
Mathematically, in the presence of such a symmetry $\mathcal{T}$, we should work with a
real Hilbert space which carries an action by the quaternion
algebra commuting with the representation of the relevant real
\Cst-algebra \( A=\Cst(X)^{\Gamma} \). 
Topological phases will therefore belong to the \emph{quaternionic
\K-theory group} \( \operatorname{KH}_0(A) \), which is isomorphic to
\( \KO_4(A) \), see
\citelist{\cite{Karoubi:K-theory}
\cite{Boersema-Schochet:Real_K-theory}*{Corollary 7.11}}. 

\subsection{The \(1\)-dimensional symmetry breaking map}
\label{sec:symmetry breaking_1-dim}
Let \( \F\in\{\C,\R\} \). We write \( \Cstr(\Gamma)_\F \) for the
reduced group \Cst-algebra of \( \Gamma \), i.e.~the \Cst-algebra
generated by (right) regular representations \( \rho_\Gamma\colon
\Gamma\to\mathrm{U}(\ell^2(\Gamma)_\F) \). Then \(
\Cstr(\Gamma)_\C\simeq\Cstr(\Gamma)_\R\otimes_\R\C \) is the
complexification of \( \Cstr(\Gamma)_\R \). 

Thus we may identify \( \Cstr(\Gamma)_\R \) as a real \st-subalgebra of \(
\Cstr(\Gamma)_\C\). For the case \( \Gamma\simeq\Z \), the Gelfand duality
\begin{equation}\label{eq:real_Fourier}
    \Cstr(\Z)_\C\simeq\Cont(\T)_\C,\; \rho_\Z(n)\mapsto z^n
\end{equation}
restricts to an isomorphism \( \Cstr(\Z)_\R\simeq\Cont(\T,\tau_\T) \),
where \( \T \) denotes the unit circle in the complex plane and
\[
\Cont(\T,\tau_\T)\defeq\big\{f\in\Cont(\T)\,:\,\overline{f(z)}=f(\overline{z})\big\}
\]
is the real fixed point algebra associated to the involution \(
\tau_\T(z)\defeq
\overline{z} \).

It is convenient to describe the KO-theory of \( \Cstr(\Z)_\R \) using
its \( \KO_*(\R) \)-module structure (see \cref{sec:real_Cst-algebras}), as
follows.

\begin{lemma} \label{lem:KO_Z_generators}
The \( \KO_*(\R) \)-module \( \KO_*( \Cstr(\Z)_\R) \) is
freely generated by a generator in
degree \( 0 \) together with a generator in degree \( 1 \).
\end{lemma}

\begin{proof}
We identify \( \Cstr(\Z)_\R \) with \( \Cont(\T,\tau_\T) \) using the
isomorphism \eqref{eq:real_Fourier}.
Then \( -1\in\T \) is invariant under the involution \(
\tau_{\T} \), so \( f(-1)\in\R \) for any \(
f\in\Cont(\T,\tau_\T) \). The kernel of the evaluation map
\[ 
  \ev_{-1}\colon \Cont(\T,\tau_\T)\to\R,\quad f\mapsto f(-1)
\]
consists of functions in \( \Cont(\T,\tau_\T) \) vanishing at \( -1 \).
Define\footnote{We may choose \( \I \) to be any open interval. The choice
\( \I\defeq (-\pi,\pi) \) is for convenience.} \( \I\defeq (-\pi,\pi) \)
with involution \( \tau_{\I}(x)=-x \) , and 
\[
\Co(\I,\tau_\I)\defeq\big\{f\in\Co(\I)\,:\,\overline{f(x)}=f(-x)\big\}.
\]
Thus we have a split exact sequence of real \Cst-algebras
\[
\begin{tikzcd}
0 \arrow[r]  &\Co(\I,\tau_\I) \arrow[r, "j"] & \Cont(\T,\tau_\T) 
\arrow[r,"\ev_{-1}"] & \R \arrow[r] \arrow[l, bend left, "s"]  & 0,
\end{tikzcd}
\]
where
\[
j\colon\Co(\I,\tau_\I)\to\Cont(\mathbb{T},\tau_\mathbb{T}),\quad
  j(f)(\mathrm{e}^{ix})\defeq f(x).
\]
and the splitting map \( s\colon \R\to\Cont(\T,\tau_\T) \) sends a real
number to the corresponding constant function on \( \T \).

The split exactness of KO-theory implies that there is 
an isomorphism of \( \KO_*(\R) \)-modules
\[
  \KO_*(\Cont(\T,\tau_\T))\simeq\KO_*(\R)\oplus\KO_*(\Co(\I,\tau_\I)).
\]
Recall from \cref{lem:KO-theory_suspension_desuspension} that 
\( \KO_*(\Co(\I,\tau_\I)) \) is isomorphic
to \( \KO_*(\R) \) with a degree shift \( -1 \). 
Thus there is a free
generator \( e \) of degree \( +1 \) that yields the following 
isomorphism of free \( \KO_*(\R) \)-modules
\[
  \KO_*(\R)\xrightarrow{\sim}\KO_{*+1}(\Co(\I,\tau_\I)),\quad
  \xi\mapsto \xi\cdot e,\quad \xi\in\KO_*(\R).\qedhere
\]
\end{proof}

Let us provide explicit representatives of the free \( \KO_*(\R) \)-module
generators in \cref{lem:KO_Z_generators} coming from 
the unique free generators of \( \KO_*(\R) \) and \(
\KO_*(\Co(\mathbb{I},\tau_{\mathbb{I}})) \).

\begin{itemize}
\item Degree 0. $\KO_*(\R)$ is freely generated by $[1]\in
\KO_0(\R)$ where \(1\in\R\) is viewed as a rank-1 projection matrix. 
Under the splitting map $s$, the generator $[1]$ is mapped to a 
generator of \( \KO_0(\Cont(\T, \tau_\T))) \), which is 
represented by the
class $[1_{\ell^2(\Z)_\R}]$ of the identity operator on \( \ell^2(\Z)_\R
\) under the isomorphism \eqref{eq:real_Fourier}.

\item Degree 1. $\KO_*(\Cont(\T,\tau_\T))$ is generated by \( e
\in \KO_1(\Co(\I,\tau_\I))\simeq\Z \). For a real
\Cst-algebra \( A \), a class in \( \KO_1(A) \) may be described by a
unitary \( u\in\Mat_n(A^+) \) in the unitary picture of KO-theory,
see~\cite{Boersema-Loring:K-theory_unitaries}. In particular, $e$ is
represented by the function $\widehat{w}_1\colon z\mapsto z$ in
$\Cont(\T,\tau_\T)=(\Co(\I,\tau_\I))^+$, which is mapped to
the unitary shift operator $w_1\defeq\rho_\Z(1)$ 
on \( \ell^2(\Z)_\R \) under the isomorphism \eqref{eq:real_Fourier}.
\end{itemize}

Generators of \( \KO_*(\Cstr(m\Z)_\R) \) are similarly described as
follows. View \( m\Z \) as a subgroup of \( \Z \), freely generated
by \( m \). This induces an isomorphism of real \Cst-algebras 
\[
\Cst(\Z)_\R\xrightarrow{\sim}\Cst(m\Z)_\R, \qquad
1_{\ell^2(\Z)}\mapsto 1_{\ell^2(m\Z)}, \quad w_1 \mapsto w_m, 
\]
where \( \Cst(m\Z) \) is realized as a concrete \Cst-algebra on \(
\ell^2(m\Z)_{\R} \), and \( w_m\defeq \rho_{m\Z}(m) \) is the shift
operator on
\( \ell^2(m\Z)_\R \). Thus \( \KO_*(\Cstr(m\Z)_\R) \) is the free \(
\KO_*(\R) \)-module generated by
\begin{equation}\label{eq:real_1D_KO_generators}
[1_{\ell^2(m\Z)_\R}]\quad\text{and}\quad [w_m] 
\end{equation} 
with degrees $0$ and $1$ respectively,

\begin{lemma}\label{lem:SB_map_1-dim}
For \( m\mid n \), the \K-theory symmetry breaking map \eqref{eq:phi_HG}
\begin{multline*}
\phi_{n\Z\backslash m\Z}\colon \KO_i(\Cstr(m\Z)_\R)
\xrightarrow{\KO_i(\Ad_{\eta_s})}
\KO_i\big(\Cstr(n\Z)_\R\otimes\Cpt(\ell^2(n\Z\backslash m\Z)_\R)\big)\\
\xrightarrow{\mathrm{St}^{-1}} 
\KO_i(\Cstr(n\Z)_\R)
\end{multline*}
is given on the free generators by
\begin{equation}\label{eq:KO.symm.breaking.d1}
[1_{\ell^2(m\Z)_\R}]\mapsto (n/m)\cdot[1_{\ell^2(n\Z)_\R}],
\qquad [w_m]\mapsto [w_n].
\end{equation}
\end{lemma}

\begin{proof}
Up to replacing \( n\Z\backslash m\Z \) by
\( \frac{n}{m}\Z\backslash \Z \), it suffices to compute the
maps \( \phi_{N\Z\backslash \Z} \) for \( N\in\N \).
To this end, pick the section \( s\colon N\Z\backslash \Z\to \Z \) given by
\begin{equation}\label{eq:section_1D_case}
  s(N\Z+i)\defeq i,\quad i=1,\ldots,N,
\end{equation}
which induces the unitary isomorphism of Hilbert spaces 
(see \eqref{eq:unitary_l2G_l2H}):
\begin{equation}\label{eq:break.N.Hilbert.iso}
\begin{aligned}
  \eta_s\colon\ell^2(\Z)_\R&\xrightarrow{\sim}\ell^2(N\Z)_\R\otimes\ell^2(N\Z\backslash
  \Z)_\R\\
  \quad \delta_{Nk+i}&\mapsfrom \delta_{Nk}\otimes\delta_{N\Z+i}. 
\end{aligned}
\end{equation}
In turn, we have an isomorphism
\[
\Ad_{\eta_s}\colon \Bdd(\ell^2(\Z)_\R)\to \Bdd(\ell^2(N\Z)_\R\otimes\ell^2(N\Z\backslash
  \Z)_\R)\simeq \Mat_N(\Bdd(\ell^2(N\Z)_\R)),
\]
where in the second isomorphism, 
we identify \( \ell^2(N\Z\backslash \Z)_\R\simeq\R^N \) via the basis \(
\{\delta_{N\Z+i}\}_{i=1,\ldots,N} \). 
Thus, for a general $T\in\Bdd(\ell^2(\Z))$, the operator
\[
\Ad_{\eta_s}(T)\defeq (\Ad_{\eta_s}(T)_{i,j})_{i,j=1}^N,
\]
is an $N\times N$ matrix with $\Bdd(\ell^2(N\Z)_\R)$-valued entries. 
Explicitly,
for \( i,j\in\{1,\dots,N\} \), the $i,j$-entry
of \( \Ad_{\eta_s}(T) \) is the operator given by 
\begin{equation}\label{eq:1-dim_matrix_form}
\begin{aligned}
  \braketvert{\delta_{Nk}}{\Ad_{\eta_s}(T)_{i,j}}{\delta_{N\ell }}
  &=\braketvert{\delta_{Nk}\otimes\delta_{N\Z+i}}{\Ad_{\eta_s}(T)}{\delta_{N\ell }\otimes \delta_{N\Z+j}}\\
  &=\braketvert{\eta_s^*(\delta_{Nk}\otimes\delta_{N\Z+i})}{T}{\eta_s^*(\delta_{N\ell }\otimes \delta_{N\Z+j})}\\
  &=\braketvert{\delta_{Nk+i}}{T}{\delta_{N\ell + j}}, & k,\ell\in\Z.
  \end{aligned}
\end{equation}

In particular,
\[
\Ad_{\eta_s}(1_{\ell^2(\Z)_\R})=\mathrm{diag}(\underbrace{1_{\ell^2(N\Z)_\R},\ldots,1_{\ell^2(N\Z)_\R}}_{N\ \mathrm{copies}}),
\]
thus \( \phi_{N\Z\backslash \Z}[1_{\ell^2(\Z)_\R}]=N\cdot[1_{\ell^2(N\Z)_\R}] \).

As for the degree-1 generator, we have
\begin{equation}\label{eq:1D.translation.broken.matrix}
    \Ad_{\eta_s}(w_1)=\begin{pmatrix}
    0 & 1_{\ell^2(N\Z)_\R} & 0 &\cdots & 0 \\
    0 &  & 1_{\ell^2(N\Z)_\R} & & 0 \\
      \vdots & & & \ddots & \vdots \\
    0 & & & & 1_{\ell^2(N\Z)_\R}\\
    w_N & 0 & 0 & \cdots & 0
    \end{pmatrix}\in\Mat_N(\Cstr(N\Z)_\R),
\end{equation}
where we recall that \( w_1\defeq \rho_{\Z}(1) \) and
$w_N\defeq \rho_{N\Z}(N)$ are the left shift operators on
\( \ell^2(\Z)_\R \) and \( \ell^2(N\Z)_\R \), respectively.
Via a permutation matrix, \(
\Ad_{\eta_s}(w_1) \) is unitarily equivalent in 
$\Mat_N(\Cstr(N\Z)_\R)$ to the diagonal operator
\[ 
  w_N\oplus
  \underbrace{1_{\ell^2(N\Z)_\R}\oplus \cdots \oplus
  1_{\ell^2(N\Z)_\R}}_\text{\( N-1 \)\ copies}.
\]
Therefore \( \phi_{N\Z\backslash \Z}[w_1]=[w_N] \).
\end{proof}

\subsection{The \(d\)-dimensional real and complex cases}
\label{sec:d-dim}
Write \( \mathfrak{Z}^d \) for the symmetry breaking
poset of \( \Z^d \), that is, the directed poset of all finite-index
subgroups of \( \Z^d \). For \(d\geq 2\), the poset $\mathfrak{Z}^d$ is
significantly more complicated than the $d=1$ case. For example, the range
of any \( M\in\Mat_{d}(\Z) \) that is invertible over \( \Q \)
(i.e.~\( M \) also belongs to \( \GL_{d}(\Q) \)) is a finite-index
subgroup of \( \Z^d \). Computing the K-theory symmetry breaking maps for
such a subgroup involves explicit description of their K-theory generators.
Nevertheless, we claim that the symmetry breaking Roe algebra \(
\Cst(\R^d)^{\mathfrak{Z}^d} \) (thus also its $\K$-theory) can be computed
using a much smaller poset that is cofinal in \(
\mathfrak{Z}^d \). 

\subsubsection{The real case}
\label{sec:real_d-dim}

Recall from \cref{dfn:cofinal} that 
a subset \( J \) of a directed poset \(
(I,\prec) \) is called \emph{cofinal}, if for every \(
i\in I \), there exists \( j\in J \) such that \(
i\prec j \).

\begin{lemma}\label{lem:cofinal_subseteq_R}
The collection of finite-index subgroups of \( \Z^d
\) of the form
\[ 
  m\Z^d \defeq\{(mx_1,\dots,mx_d)\,:\,m\in\N,\,(x_1,\dots,x_d)\in\Z^d\},
\]
is cofinal in $\mathfrak{Z}^d$. Thus there are natural isomorphisms
\[ 
  \Cst(X,\mathcal{H})^{\mathfrak{Z}^d}\simeq\colim_{G\in\mathfrak{Z}^d}\Cst(X,\mathcal{H})^G\simeq
  \colim_{m\in\N}\Cst(X,\mathcal{H})^{m\Z^d}.
\]
for any \( X \)-\( \Z^d \)-module \( \mathcal{H} \) satisfying
\refbasicsetup. 
\end{lemma}

\begin{proof}
Let $G\in\mathfrak{Z}^d$ be a finite-index subgroup of $\mathbb{Z}^d$. Then
\( G\backslash \Z^d  \) has order \( m \) for some $m\in\mathbb{N}$.
Then for any \( x\in\Z^d \),
its image \( [x]\in G\backslash \Z^d \) satisfies \( m[x]=[0] \). This implies that
\( mx\in G \) for every \( x\in\Z^d \), i.e., \( m\Z^d\subseteq G \).  Thus
every finite-index subgroup \( G\subseteq \Z^d \) contains a subgroup of
the form \( m\Z^d \). We may now apply \cref{lem:cofinal_isomorphism} to deduce the second claim of \cref{lem:cofinal_subseteq_R}. 
\end{proof} 

In the \refbasicsetup, let $\Gamma\simeq\Z^d$ be a free abelian group of
rank $d$. As a consequence of 
\cref{thm:main_computation_theorem,lem:cofinal_subseteq_R}, we have
\begin{equation}\label{eq:real.symm.break.as.colimit}
\KO_i(\Cst(X,\mathcal{H})_{\R}^{\mathfrak{Z}^d})\simeq\colim_{m\in\N}\KO_i(\Cst(X,\mathcal{H})_{\R}^{m\Z^d})\simeq\colim_{m\in\N}\KO_i(\Cstr(m\Z^d)_{\R}).
\end{equation}
So our task reduces to the computation of the K-theory symmetry breaking maps
\begin{equation}\label{eq:real.symm.break.maps}
  \phi_{n\Z^d\backslash m\Z^d}\colon
  \KO_i(\Cstr(m\Z^d)_\R)\to\KO_i(\Cstr(n\Z^d)_\R),\quad \text{for \( m\mid
  n\)},
\end{equation}
which will be facilitated by the following version of the K\"unneth formula for real \Cst-algebras, due to Boersema 
\citelist{\cite{Boersema:Real_Cst-algebras_Kuenneth_formula}
\cite{Boersema-Schochet:Real_K-theory}*{Theorem 7.3}}:
\begin{lemma}\label{lem:real_Kuenneth_theorem}
Let \( A \) and \( B \) be real \Cst-algebras, 
where \( A \) is in the pre-bootstrap category 
\textup{(}which includes all commutative real \Cst-algebras\textup{)}, 
and \( B \) has the
property that \( \K_*(B_\C) \) is torsion-free, where \( B_\C \) is the
complexification of \( B \). 
Then the external product map induces an isomorphism of 
graded \( \KO_*(\R) \)-modules 
\[ 
  \KO_*(A)\otimes_{\KO_*(\R)}\KO_*(B)\xrightarrow{\sim}\KO_*(A\otimes B).
\]
\end{lemma}

There is a canonical isomorphism of real \Cst-algebras,
\begin{equation}\label{eq:basis_induce_tensor}
\Cstr(m\Z^d)_\R\simeq\bigotimes_{j=1}^d \Cstr(m\Z)_\R.
\end{equation}
The complexification 
$\Cstr(m\Z^d)_\R\otimes_\R\C=\Cstr(m\Z^d)_\C$ has
torsion-free $\K$-theory.
So Boersema's K\"unneth Theorem applies,
and gives an isomorphism of graded \( \KO_*(\R) \)-modules: 
\begin{equation}\label{eq:KO_mZd_tensor_product}
  \KO_*( \Cstr(m\Z^d)_{\R})\simeq
  \underbrace{\KO_*(
  \Cstr(m\Z)_{\R})\otimes_{\KO_*(\R)}\cdots\otimes_{\KO_*(\R)}\KO_*( \Cstr(m\Z)_{\R})}_{\text{\(
  d\)-copies of \( \KO_*( \Cstr(m\Z)_{\R}) \)} }.
\end{equation}

By \cref{lem:KO_Z_generators}, \( \KO_*(\Cstr(m\Z)_{\R}) \) is a free \(
\KO_*(\R) \)-module generated by \( [1_{\ell^2(m\Z_\R)}] \) in degree \( 0
\) and \( [w_m] \) in degree \( 1 \). 
The right-hand side of \eqref{eq:KO_mZd_tensor_product} is
a \( d \)-fold tensor product of this free module. Define
\begin{equation}\label{eq:d-dim_KO-generators}
  w_{m,I}\defeq x_1\otimes\dots\otimes x_d,\quad x_j\defeq \begin{cases}
  [1_{\ell^2(m\Z)_\R}] & \text{if \( j\in I\);} \\
  [w_m] & \text{if \( j\notin I \),} 
  \end{cases}
\end{equation}
where
\[
  I=\{i_1,\dots,i_{\abs{I}}\},\quad 1\leq i_1<i_2<\dots<i_{\abs{I}}\leq d,
\]
is a multi-index. Then the right-hand side of
\eqref{eq:KO_mZd_tensor_product} is a free \( \KO_*(\R) \)-module generated
by \( w_{m,I} \) for all multi-indices \( I\subseteq \{1,\dots,d\} \). The
generator \( w_{m,I} \) has
degree \( d-\abs{I} \) because 
$w_{m,I}$ has $\abs{I}$ factors of $[1_{\ell^2(m\Z^d)}]$ and \( d-\abs{I}
\) factors of \( [w_m] \).  

We claim that the symmetry breaking maps \( \phi_{n\Z^d\backslash m\Z^d} \) 
are functorial under the isomorphism in \eqref{eq:KO_mZd_tensor_product}. 
This means the following:

\begin{lemma}\label{lem:real_Kunneth_breaking_formula}
The following diagram of isomorphisms commutes:
\[
\begin{tikzcd}
\KO_*\left(\bigotimes\limits_{j=1}^d\Cstr(m\Z)_{\R}\right) \ar[r,"\textnormal{K\"unneth}"]\ar[d,"\KO_*\big(\bigotimes\limits_{j=1}^d\Ad_{\eta_{s}}\big)"] & \bigotimes_{j=1}^d\KO_*(\Cstr(m\Z)_{\R}) \ar[d,"\otimes_{j=1}^d\KO_*(\Ad_{\eta_{s}})"]\\
\KO_*\left(\bigotimes\limits_{j=1}^d \big(\Cstr(n\Z)_{\R}\otimes \Cpt(\ell^2(n\Z\backslash m\Z)_{\R})\big)\right)\ar[r,"\textnormal{K\"unneth}"]\ar[d,"\Stab^{-1}"] & \bigotimes\limits_{j=1}^d\KO_*\left(\Cstr(n\Z)_{\R}\otimes \Cpt(\ell^2(n\Z\backslash m\Z)_{\R})\right)\ar[d,"\Stab^{-1}"] \\
\KO_*\left(\bigotimes\limits_{j=1}^d\Cstr(n\Z)_{\R}\right)
\ar[r,"\textnormal{K\"unneth}"] & \bigotimes\limits_{j=1}^d\KO_*(\Cstr(n\Z)_{\R}).
\end{tikzcd}
\]
\end{lemma}
\begin{proof}
Pick a section $s\colon n\Z\to m\Z$. The K\"unneth isomorphism for
KO-theory is functorial for \st-homomorphisms between real \Cst-algebras.
Consequently, the top half of the above diagram of $\KO_*(\R)$-modules
commutes.  Likewise, the lower half commutes because the stabilization
isomorphisms are induced by the \st-homomorphism of tensoring with any
rank-1 projection.
\end{proof}

Consequently, \( \phi_{n\Z^d\backslash m\Z^d} \) 
can be identified with a \( d \)-fold tensor product of
symmetry breaking maps in dimension one using the natural isomorphism of
\eqref{eq:KO_mZd_tensor_product}. That is,
\begin{equation}\label{eq:KO_break_map_tensor}
\phi_{n\Z^d\backslash m\Z^d}\defeq\phi_{n\Z\backslash
m\Z}\otimes\cdots\otimes\phi_{n\Z\backslash m\Z},\quad
\text{for \( m\mid n \).} 
\end{equation}

\begin{corollary}
In terms of the K\"unneth isomorphism \eqref{eq:KO_mZd_tensor_product} and
the \KO-theory generators \eqref{eq:d-dim_KO-generators}, the
symmetry breaking map \( \phi_{n\Z^d\backslash m\Z^d} \) for \( m\mid n \) 
is the $\KO_*(\R)$-module homomorphism given by the formula
\begin{equation}\label{eq:KO.breaking.formula}
\begin{aligned}
\phi_{n\Z^d\backslash m\Z^d}\colon 
\KO_{d-|I|}(\Cstr(m\Z^d)_\R) & \to \KO_{d-|I|}(\Cstr(n\Z^d)_\R)\\
w_{m,I} &\mapsto (n/m)^{\abs{I}}w_{n,I}.
\end{aligned}
\end{equation}
\end{corollary}

\begin{theorem}\label{thm:KO-theory_real}
Let $\Gamma$ be a free abelian group on $d$ generators. Let \( \mathcal{H}
\) be a untwisted,
real \( X \)-\( \Gamma \)-module as in \refbasicsetup, and \(
\mathfrak{S} \) its symmetry breaking poset. Let \(
\Cst(X,\mathcal{H})_\R^\mathfrak{S} \) be the corresponding real
symmetry breaking Roe algebra. Then
\begin{align*} 
\KO_i(\Cst(X,\mathcal{H})_\R^{\mathfrak{S}})&\simeq\colim_{G\in\mathfrak{Z}^d}\KO_i(\Cstr(G)_\R)\\
  &\simeq
  \begin{cases}
  \Q^{q(d,i)-1}\oplus\Z, & \text{if}\ i-d\equiv 0,4\hphantom{,4,4}\mod 8;\\
  \Q^{q(d,i)}\oplus\Z/2, & \text{if}\ i-d\equiv 1,2\hphantom{,1,2}\mod 8;\\
  \Q^{q(d,i)}, & \text{if}\ i-d\equiv 3,5,6,7\mod 8;
  \end{cases}
\end{align*}
where
\begin{equation}\label{eq:q.coefficient}
q(d,i)=\sum_{j\in\Z}\binom{d}{i+8j}+\binom{d}{i+4+8j}.
\end{equation}
\end{theorem}
Here \( \binom{m}{n} \) is understood as \( 0 \) unless \( 0\leq n\leq
m \). Thus the sum in the formula \eqref{eq:q.coefficient} for \( q(d,i) \) is finite.

\begin{proof}
We need to compute the right side of
\eqref{eq:real.symm.break.as.colimit}, i.e., the colimit of the
symmetry breaking maps $\phi_{n\Z^d\backslash m\Z^d}$ of
\eqref{eq:real.symm.break.maps}. The formula for $\phi_{n\Z^d\backslash
m\Z^d}$ has a simple ``diagonal'' form, given by
\eqref{eq:KO.breaking.formula} of
\cref{lem:real_Kunneth_breaking_formula}. 

\noindent{\itshape The ``strong'' summand}.
There is a unique \( \KO_*(\R) \)-module generator
\begin{equation}\label{eq:strong.basis}
  w_{m,\emptyset}=[w_m]\otimes[w_m]\otimes\dots\otimes [w_m]
\end{equation}
that has degree \( d \). This \(
w_{m,\emptyset} \) generates a \emph{``strong summand''} in \(
\KO_i(\Cstr(m\Z^d)_{\R}) \)
after multiplying it with a coefficient
\[
  \xi\in\KO_{i-d}(\R).
\]
By Bott periodicity, it suffices to consider \( i-d\in\{0,1,\dots,7\}
\) as we may replace \( \xi\in\KO_{i-d}(\R) \) by \(
\xi\beta^{j}\in\KO_{i+8j-d}(\R) \) for any \( j\in\Z \). 
The ``strong'' summand is isomorphic to \( 0 \), \( \Z/2 \) or \( \Z \)
depending
on \( i-d \):   
\begin{itemize}
\item For \( i-d\equiv 1,2\mod 8 \), the coefficient \( \xi \) belongs to
\( \KO_{i-d}(\R)\simeq\Z/2 \), so $w_{m,\emptyset}$ generates a ``strong''
summand isomorphic to $\Z/2$.
\item 
For \( i-d\equiv 0,4\mod 8 \), the coefficient \( \xi \) belongs to \(
\KO_{i-d}(\R)\simeq\Z \), so $w_{m,\emptyset}$ generates a ``strong''
summand in  \(
\KO_{i}(\Cstr(m\Z^d)_{\R}) \) isomorphic to $\Z$.
\item For \( i-d\equiv 3,5,6,7\mod 8 \), we have $\KO_{i-d}(\R)=0$, so
$w_{m,\emptyset}$ does not generate any non-trivial summand in  \(
\KO_{i}(\Cstr(m\Z^d)_{\R}) \).
\end{itemize}
Since the symmetry breaking map \( \phi_{n\Z^d\backslash m\Z^d}
\) is a \( \KO_*(\R) \)-module homomorphism, we have
\[ 
\phi_{n\Z^d\backslash m\Z^d}(\xi\cdot w_{m,\emptyset})=\xi\cdot \phi_{n\Z^d\backslash m\Z^d}(w_{m,\emptyset})=\xi\cdot w_{n,\emptyset},
\]
i.e., it restricts to isomorphisms on the ``strong''
summands. Therefore, the colimit abelian group \(
\colim\KO_i(\Cstr(m\Z^d)_{\R}) \) will contain a
``strong'' summand isomorphic to \( 0 \), \( \Z/2 \) or \( \Z \) according
to \( i-d \mod 8\).
\vspace*{1em}

\noindent{\itshape The ``weak'' summands}.
Now we consider the generators \( w_{m,I} \) where 
\(I\neq \emptyset\). Then \( w_{m,I} \) has degree
$d-|I|$. As before, \( w_{m,I} \) generates a summand in \(
\KO_i(\Cstr(m\Z^d)_{\R}) \) after multiplication with a coefficient \(
\xi\in\KO_{i-d+|I|}(\R) \) and we call this a ``weak'' summand.
Recalling the description of the graded ring $\KO_*(\R)$ in \eqref{eq:KO-ring_of_R}), we have the following cases. 
\begin{itemize}
\item 
For \( i-d+|I|\equiv 1,2\mod 8 \), then \( \KO_{i-d+|I|}(\R)\simeq\Z/2 \)
is generated by a torsion-two element of the form \( \eta\beta^k \)
or \( \eta^2\beta^k \) for some \( k\in\Z \).
So $w_{m,I}$ contributes a summand in \( \KO_i(\Cstr(m\Z^d)_{\R}) \)
isomorphic to \( \Z/2 \). We claim that such a
summand cannot survive in the colimit. 
To this end, we note that \eqref{eq:KO.breaking.formula} gives
\begin{equation}\label{eq:KO.weak.breaking.formula}
\begin{aligned}
  \phi_{n\Z^d\backslash m\Z^d}(\eta\beta^k\cdot
  w_{m,I})&=\eta\beta^k\cdot\phi_{n\Z^d\backslash m\Z^d}(w_{m,I})\\ 
  &=\eta\beta^k\cdot
  (n/m)^{|I|}w_{n,I}=(n/m)^{|I|}\eta\beta^k\cdot w_{n,I}.
\end{aligned}
\end{equation}
Since $|I|>0$ and $\eta\beta^k$ is 2-torsion, 
the right side vanishes when $n=2m$.
Thus \( \phi_{2m\Z^d\backslash m\Z^d} \) restricts to the zero map on the summands generated by \( \eta\beta^k\cdot w_{m,I} \)
and \( \eta^2\beta^k\cdot w_{m,I} \).
By the universal property, these 2-torsion summands must map to zero in the
colimit.

\item 
For \( i-d+|I|\equiv 0,4\mod 8 \), \( \KO_{i-d+|I|}(\R)\simeq\Z \) is
generated by a free element. Thus $w_{m,I}$ contributes a summand in \(
\KO_i(\Cstr(m\Z^d)_{\R}) \) isomorphic to \( \Z \). By \eqref{eq:KO.breaking.formula} of \cref{lem:real_Kunneth_breaking_formula}, the symmetry breaking maps restricted to these $I$-th summands are
\[
\Z\overset{(\frac{n}{m})^{|I|}}\longrightarrow\Z,\qquad m\mid n,
\]
and the colimit is isomorphic to $\Q$ since $|I|\neq \emptyset$; more
precisely,
\begin{equation}\label{eq:basic_mult_colimit}
\begin{tikzcd}[column sep=2cm]
\overbrace{\Z}^{m\text{-th}}\ar[r,"(\frac{n}{m})^{|I|}"]\ar[dr,"\frac{1}{m^{|I|}}"'] & \overbrace{\Z}^{n\text{-th}}\ar[d,"\frac{1}{n^{|I|}}"] \\
& \Q
\end{tikzcd}
\end{equation}
For the full symmetry breaking maps, the number of such $\Q$ summands
appearing in the colimit is equal to the number of non-empty 
multi-indices $I$ satisfying
\[
i-d+|I|\equiv 0,4\mod 8,\qquad 0<|I|\leq d.
\]
If \( i-d\not\equiv0,4\mod 8 \), then every multi-index \( I \) satisfying
\begin{equation}\label{eq:weak.multiindex}
  \abs{I}=d-i\mod 8\quad\text{or}\quad\abs{I}=d-i-4\mod 8
\end{equation}
is non-empty. The number of such multi-indices is given by $q(d,i)$ of
\eqref{eq:q.coefficient}.

If \( i-d\equiv 0,4\mod 8 \), the empty multi-index satisfies
\eqref{eq:weak.multiindex}, but as analyzed previously, $w_{m,\emptyset}$ generates a ``strong''
summand in the colimit, rather than a $\Q$ summand. \qedhere
\end{itemize}
\end{proof}

\subsubsection{The complex case}
\label{sec:complex_d-dim}

The complex case is conceptually the same as the real case, but much
simpler. All the results in
\cref{sec:real_d-dim} before \cref{thm:KO-theory_real}
hold in
the complex case, up to replacing the real \Cst-algebra \( \Cstr(m\Z^d)_\R
\), their KO-theory \( \KO_*(\Cstr(m\Z^d)_\R) \) and the base ring \(
\KO_*(\R) \) by their corresponding complex counterparts. Schochet's K\"unneth Theorem for complex \Cst-algebras
(see \citelist{\cite{Schochet:Kuenneth_formula}
\cite{Blackadar:Operator_algebras}*{Theorem 23.1.3}}) is the same as
\cref{lem:real_Kuenneth_theorem} if we replace \( \KO_*(A) \) and \(
\KO_*(B) \) by \( \K_*(A) \) and \( \K_*(B) \), as well as
\( \KO_*(\R) \) by \( \K_*(\C) \). This is because a \( \K_*(\C) \)-module
is the same data as a \( \Z/2 \)-graded abelian group, see
\cref{rmk:K-ring_of_C}. 
The symmetry breaking formula \eqref{eq:KO.breaking.formula} in the complex
case is now the \( \K_*(\C) \)-module homomorphism
\[
\begin{aligned}
\phi_{n\Z^d\backslash m\Z^d}\colon 
\K_{d-\abs{I}}(\Cstr(m\Z^d)_\C) & \to \K_{d-\abs{I}}(\Cstr(n\Z^d)_\C)\\
w_{m,I} &\mapsto (n/m)^{\abs{I}}w_{n,I}.
\end{aligned}
\]

The complex analogue of the computation in \cref{thm:KO-theory_real} is simpler because the \( \Z/2 \)-graded commutative ring \( \K_*(\C) \) is
much simpler, having no torsion summand. The classes $w_{m,\emptyset}$
generate a strong summand $\Z$, while the remaining $2^d-1$ classes
$w_{m,I}, I\neq \emptyset$ generate ``weak'' summands isomorphic to $\Q$. The
result is summarized below.
\begin{theorem}\label{thm:K-theory_complex_twisted}
Let $\Gamma$ be a free abelian group on $d$ generators. Let \( \mathcal{H}
\) be a untwisted, complex
\( X \)-\( \Gamma \)-module as in \refbasicsetup, and \(
\mathfrak{S} \) its symmetry breaking poset. Let \(
\Cst(X,\mathcal{H})_\C^\mathfrak{S} \) be the corresponding complex
symmetry breaking Roe algebra. Then
\begin{align*}
  \K_i(\Cst(X,\mathcal{H})_\C^{\mathfrak{S}})
  &\simeq\colim_{G\in\mathfrak{Z}^d}\K_i(\Cstr(G)_\C) \\
  &\simeq\begin{cases}
  \Q^{2^{d-1}-1}\oplus\Z, & \text{if \( i-d\equiv 0\mod 2 \);} \\
  \Q^{2^{d-1}}, & \text{if \( i-d\equiv 1\mod 2 \).}
  \end{cases}
\end{align*}
\end{theorem}

\begin{remark}\label{rem:micro.geom.irrelevant}
As explained in \cref{rmk:equivariant_coarse_equivalence}, the K(O)-theory
of \( \Cst(X,\mathcal{H})^{\mathfrak{S}} \) 
is a (equivariant, coarse-geometric)
invariant of the free abelian group \( \Gamma \). Thus  
\cref{thm:KO-theory_real} applies, in particular, to $X\simeq\R^d$ the
Euclidean $d$-space (continuum models), and to $X=\Gamma\simeq\Z^d$ (tight-binding
models). In the $X\simeq\R^d$ case, we can actually allow the manifold
$\R^d$ to have ``microscopically varying geometry'' which is only
$\Gamma$-invariant. In other words, microscopic geometric details of each fundamental
domain are not important for the result of \cref{thm:KO-theory_real}.
\end{remark}

\section{Weak and strong topological phases}
\label{sec:weak_strong_phases}
In the proof of \cref{thm:KO-theory_real}, the adjectives ``strong'' and
``weak'' referred to the summands in the $\K$-theory of
$\Cstr(\Gamma)_\F\simeq\Cstr(\Z^d)_\F$ arising from generators 
$w_{1,I}$,
according to whether the degree 
$d-|I|$ equals $d$ or not. This is
consistent with terminology in physics
(e.g.~\cites{Fu-Kane-Mele:TI_in_3D,Fu-Kane:Topological_insulators_with_inversion_symmetry}
in the real case and \cites{Prodan-SBaldes:Complex_topological_insulators}
in the complex case).  The strong summands (see
\cref{tab:strong_symmetry breaking_colimits}; the terms \( \sK_i \) and \(
\sKO_i \) will be explained in \cref{sec:weak_vs_strong}) depend only on \(
i-d\mod 2 \) (in the complex case) or \( i-d\mod 8 \) (in the real case).
They are the same abelian groups which appear in the Kitaev's famous
periodic table of topological insulators and
superconductors \cite{Kitaev:Periodic_table}; in fact, this table comprises the homotopy groups of the
stable classical groups, whose periodicity was discovered by Bott
\cite{Bott:Periodicity}.  In \cite{Kitaev:Periodic_table}*{Page 28--29},
the table was argued to either classify topological insulators and
superconductors in Euclidean space with continuous untwisted
$\R^d$-symmetry, or the ``strong'' summand in the case of untwisted
$\Z^d$-symmetry.

\begin{table}[h!]
\centering
\begin{tabular}{c|ccccccccc}
\hline
\diagbox{\( i \)}{\( d \)} & 0 & 1 & 2 & 3 & 4 & 5 & 6 & 7 & 8 \\
\hline
\( {\rm sK}_0 \) & \( \Z \) & 0 & \( \Z \) & 0 & \( \Z \) & 0 & \( \Z \) & 0 & \( \Z \) \\
\( {\rm sK}_1 \) & 0 & \( \Z \) & 0 & \( \Z \) & 0 & \( \Z \) & 0 & \( \Z \) & 0 \\ 
\hline\hline
\diagbox{\( i \)}{\( d \)} & 0 & 1 & 2 & 3 & 4 & 5 & 6 & 7 & 8 \\
\hline
\( {\rm sKO}_0 \) & \( \Z \) & \( 0 \) & \( 0 \) & \( 0 \) & \( \Z \) & \( 0 \) & \( \Z/2 \) & \( \Z/2 \) & \( \Z \) \\ 
\( {\rm sKO}_1 \) & \( \Z/2 \) & \( \Z \) & \( 0 \) & \( 0 \) & \( 0 \) & \( \Z \) & \( 0 \) & \( \Z/2 \) & \( \Z/2 \) \\
\( {\rm sKO}_2 \) & \( \Z/2 \) & \( \Z/2 \) & \( \Z \) & \( 0 \) & \( 0 \) & \( 0 \) & \( \Z \) & \( 0 \) & \( \Z/2 \) \\
\( {\rm sKO}_3 \) & 0 & \( \Z/2 \) & \( \Z/2 \) & \( \Z \) & \( 0 \) & \( 0 \) & \( 0 \) & \( \Z \) & \( 0 \) \\
\( {\rm sKO}_4 \) & \( \Z \)  & \( 0 \) & \( \Z/2 \) & \( \Z/2 \) & \( \Z \) & \( 0 \) & \( 0 \) & \( 0 \) & \( \Z \) \\
\( {\rm sKO}_5 \) & 0 & \( \Z \) & \( 0 \) & \( \Z/2 \) & \( \Z/2 \) & \( \Z \) & \( 0 \) & \( 0 \) & \( 0 \) \\
\( {\rm sKO}_6 \) & 0 & 0 & \( \Z \) & \( 0 \) & \( \Z/2 \) & \( \Z/2 \) & \( \Z \) & \( 0 \) & \( 0 \) \\
\( {\rm sKO}_7 \) & 0 & 0 & 0 & \( \Z \) & \( 0 \) & \( \Z/2 \) & \( \Z/2 \) & \( \Z \) & \( 0 \) \\
\hline
\end{tabular}
\caption{The ``strong'' summand of K(O)-theory of \(
\Cst(\R^d)^{\mathfrak{Z}^d} \)}
\label{tab:strong_symmetry breaking_colimits}
\end{table}

However, it is not clear, a priori, that this notion of strong/weak
summands is independent of the choice of basis of $\Gamma$, that is, an
isomorphism $\Gamma\simeq\Z^d$; or is consistent with the passage to general
finite-index subgroups $G\subseteq \Gamma$. (We only verified this for the
case of $n\Z^d\subseteq m\Z^d$). Moreover, this kind of definition does not
generalize to non-abelian $\Gamma$, which have the further complication
that finite-index subgroups $G,G^\prime\subseteq \Gamma$ need not even be
isomorphic to each other. 

In particular, the notion of a ``strong summand'' relies on choosing a
special family of generator{s} $w_{m,\emptyset}$ for each \(
\Cstr(m\Z^d) \) (see
\eqref{eq:strong.basis}), built from a tensor product
decomposition $\Cstr(m\Z^d)_\F\simeq
\Cstr(m\Z)_\F^{\otimes d}$. It turned out that for each \( m \),
\( w_{m,\emptyset} \) is a free K(O)-theory generator, and that this generator does
not get rescaled under symmetry breaking from $m\Z^d$ to $n\Z^d$ (that
is, \( \phi_{n\Z^d\backslash m\Z^d}(w_{m,\emptyset})=w_{n,\emptyset} \)).
However, several choices are involved in constructing such a
family of K(O)-theory generators, and there are other possible generators
with this scale-invariance property. For example, in
$d=2$, we could pick any non-zero integer $k$ and set, for each $m$,
\[
w^\prime_{m,\emptyset}\defeq mkw_{m,\{1,2\}} + w_{m,\emptyset}.
\]
With respect to the bases $\{w_{m,\{1,2\}}, w^\prime_{m,\emptyset}\}$ for
$\K_0(\Cstr(\Z^2))\simeq\Z^2$, the symmetry breaking map
$\phi_{n\Z^2\backslash m\Z^2}$ is represented by the $2\times 2$ integer
matrix
\[
\begin{pmatrix}
1 & nk \\ 0 & 1
\end{pmatrix}
\begin{pmatrix}
n/m & 0 \\ 0 & 1
\end{pmatrix}
\begin{pmatrix}
1 & -mk \\ 0 & 1
\end{pmatrix}
=\begin{pmatrix}
n/m & 0 \\ 0 & 1
\end{pmatrix},
\]
which is the same as the matrix with respect to the bases $\{w_{m,\{1,2\}},
w_{m,\emptyset}\}$. In other words, the family of K(O)-theory
generators $\{w^\prime_{m,\emptyset}\}_{m\in\N}$ 
is also scale-invariant under symmetry breaking.

The above discussion indicates that ``strong'' K(O)-theory generators
cannot be chosen completely canonically.  We will provide an
\emph{intrinsic} definition of ``weak'' and ``strong'' K(O)-theory groups,
which applies to arbitrary \( \Gamma \).
We will see that ``weak'' elements form a \emph{canonical subgroup}, and
that after modding out the ``weak subgroup'', we obtain a ``strong''
\emph{canonical quotient} with an important universal property. When
$\Gamma\simeq \Z^d$, this
quotient group splits onto a direct summand of the K(O)-theory in a
non-canonical way; the image is a ``strong'' summand  as used in the
physical literature and in the proof of \cref{thm:KO-theory_real}.

\subsection{Weak versus strong via symmetry breaking colimits}
\label{sec:weak_vs_strong}

\medskip
Recall that an element $a$ of an abelian group $A$ is \emph{divisible}, if for each positive integer $n$, there exists $b\in A$ such that $a=nb$.
A subgroup of \( A \) is called divisible if it consists of divisible
elements.

There is a unique maximal divisible subgroup of each abelian group \(
  A \).  Let \( W \) be the maximal divisible subgroup of \( A \), and \(
R\defeq A/W \) be its quotient. Then \( R \) is a \emph{reduced} abelian
group, i.e.~possessing no non-trivial divisible elements. There exists
a splitting
\begin{equation}\label{eq:divisible.splitting}
A\simeq W\oplus R.
\end{equation}
However, the splitting
\eqref{eq:divisible.splitting}, thus the inclusion $R\hookrightarrow
  W\oplus R\simeq A$ as a reduced abelian subgroup of \( A \), is
\emph{not} canonical.

As displayed in \cref{tab:K-theory_untwisted},
K(O)-groups of symmetry breaking Roe algebras may have non-trivial divisible
subgroups, arising from symmetry breaking for ``weak'' summands.
Intuitively, such divisible subgroups should be irrelevant for describing
\emph{quantized} observables. Let us spell out this idea in detail. For
convenience, we will drop the reference to the
$X$-$\Gamma$-module $\mathcal{H}$, and work with complex $\K$-theory.
Exactly the same definitions work for real $\K$-theory.

\begin{definition}\label{dfn:weak_strong}
Under the \refbasicsetup, let $\mathfrak{S}$ denote the symmetry breaking
poset of $\Gamma$. We define the \emph{weak
(symmetry breaking) colimit} ${\rm wK}_i(\Cst(X)^\mathfrak{S})$ to be the
maximal divisible subgroup of $\K_i(\Cst(X)^\mathfrak{S})$, and the
\emph{strong (symmetry breaking) colimit} to be the (reduced) quotient
\begin{equation}\label{eq:strong_K}
{\rm sK}_i(\Cst(X)^\mathfrak{S})\defeq \K_i(\Cst(X)^\mathfrak{S})/{\rm wK}_i(\Cst(X)^\mathfrak{S}).
\end{equation}
For each $G\in \mathfrak{S}$, let $\K_i(\iota_{\mathfrak{S}\backslash
G})\colon \K_i(\Cst(X)^G)\to \K_i(\Cst(X)^\mathfrak{S})$ be the universal map into the symmetry breaking colimit.
We define the \emph{weak subgroup} of $\K_i(\Cst(X)^G)$ to be the preimage
\[
{\rm wK}_i(\Cst(X)^G)\defeq \big(\K_i(\iota_{\mathfrak{S}\backslash G})\big)^{-1}\big({\rm wK}_i(\Cst(X)^\mathfrak{S})\big),
\]
and the \emph{strong quotient} of $\K_i(\Cst(X)^G)$ to be
\begin{equation}\label{eq:strong.quotient}
{\rm sK}_i(\Cst(X)^G)\defeq \K_i(\Cst(X)^G)/{\rm wK}_i(\Cst(X)^G),
\end{equation}
with quotient maps denoted by
\[
q_G\colon \K_i(\Cst(X)^G) \to {\rm sK}_i(\Cst(X)^G).
\]
\end{definition}

\begin{remark}
We adopt the same weak/strong terminology and notation for 
\K-theory of the group \Cst-algebra
$\K_i(\Cstr(G,\sigma))$, based on the directed systems of abelian groups in \cref{cor:exists.equivalent.directed.system}. 
Actually, the weak/strong
dichotomy makes sense without the assumptions of properness of the
$\Gamma$-action and boundedness of the fundamental domain in the
\refbasicsetup. These assumptions are only needed to obtain the natural
isomorphisms (\cref{thm:main_computation_theorem}) between
$\K_i(\Cst(X)^G)$ and $\K_i(\Cstr(G,\sigma))$, therefore between their
weak subgroups and strong quotients as well.
\end{remark}

We denote by \(
\wK_i(\iota_{H\backslash G})\) the restriction of \(
\K_i(\iota_{H\backslash G}) \) to the weak subgroup \( \wK_i(\Cst(X)^G) \). Then \(
\wK_i(\iota_{H\backslash G}) \) maps \( \wK_i(\Cst(X)^G) \) into \(
\wK_i(\Cst(X)^H) \) by
definition.
So we get an $\mathfrak{S}$-shaped diagram whose connecting morphisms
are given by \( \wK_i(\iota_{H\backslash G}) \). 
Also, each \( \K_i(\iota_{H\backslash G}) \) induces a group
homomorphism between the quotients
\begin{equation}\label{eq:strong.diagram}
{\rm sK}_i(\iota_{H\backslash G})\colon {\rm sK}_i(\Cst(X)^G)\to {\rm
sK}_i(\Cst(X)^H).
\end{equation}
These are the connecting morphisms of 
an $\mathfrak{S}$-shaped diagram of strong quotients. Altogether,
we have a short exact sequence of 
$\mathfrak{S}$-shaped diagrams:
for every \( G,H\in\mathfrak{S} \) with \( G\prec H \), 
the following diagram commutes, and its rows are short exact sequences,
\[
\begin{tikzcd}[column sep=0.8cm]
0\ar[r] & {\rm wK}_i(\Cst(X)^G) \ar[d,"\wK_i(\iota_{H\backslash G})"]\ar[r] & \K_i(\Cst(X)^G)\ar[d,"\K_i(\iota_{H\backslash G})"]\ar[r,"q_G"] & {\rm sK}_i(\Cst(X)^G)\ar[d,"{\rm sK}_i(\iota_{H\backslash G})"]\ar[r] & 0\\
0\ar[r] & {\rm wK}_i(\Cst(X)^H)\ar[r]
 & \K_i(\Cst(X)^H)\ar[r,"q_H"]
 & {\rm sK}_i(\Cst(X)^H) \ar[r] & 0.
\end{tikzcd}
\]
It follows from
\cref{lem:direct_colimit_Ab_exact} that
\[
0\longrightarrow \colim_{G\in \mathfrak{S}}{\rm wK}_i(\Cst(X)^G)\overset{\alpha}{\longrightarrow} \overbrace{\colim_{G\in \mathfrak{S}}\K_i(\Cst(X)^G)}^{\simeq \K_i(\Cst(X)^\mathfrak{S})}\overset{q_\mathfrak{S}}{\longrightarrow} \colim_{G\in \mathfrak{S}}{\rm sK}_i(\Cst(X)^G)\longrightarrow 0
\]
is a short exact sequence. 
Here, $\alpha,q_\mathfrak{S}$ are the uniquely determined homomorphisms
coming out of the respective colimits. By construction, the image of
$\alpha$ is ${\rm wK}_i(\Cst(X)^\mathfrak{S})$. Thus, there is a canonical
isomorphism of abelian groups,
\[
\alpha\colon \colim_{G\in \mathfrak{S}}{\rm wK}_i(\Cst(X)^G)\xrightarrow{\sim}{\rm wK}_i(\Cst(X)^\mathfrak{S}).
\]
This also shows that
\[
\colim_{G\in \mathfrak{S}}{\rm sK}_i(\Cst(X)^G)\simeq \K_i(\Cst(X)^\mathfrak{S})/{\rm wK}_i(\Cst(X)^\mathfrak{S})\overset{\eqref{eq:strong_K}}{=}{\rm sK}_i(\Cst(X)^\mathfrak{S}).
\]
This justifies the terminology ``weak/strong symmetry breaking
\emph{colimits}'' in \cref{dfn:weak_strong} of ${\rm sK}_i(\Cst(X)^\mathfrak{S})$ and ${\rm wK}_i(\Cst(X)^\mathfrak{S})$.

\medskip
Note that the notions of weak subgroup and strong quotient of
$\K_i(\Cst(X)^G)$ provided by \cref{dfn:weak_strong} depend only on (the symmetry breaking poset of) $\Gamma$ and $\sigma$. In view of \cref{thm:main_computation_theorem}, the choice of $X$-$\Gamma$-module is also unimportant. In particular, they do not rely on any particular presentation of $G$ or of $\K_i(\Cst(X)^G)\simeq \K_i(\Cstr(G,\sigma))$. 

\begin{example}
Suppose $G_d$ is free abelian on $d$ generators, and we decompose $G_d\simeq
G_{d-1}\oplus G_\perp$ for some rank-$(d-1)$ sublattice $G_{d-1}$ and some
rank-1 lattice $G_\perp$. The inclusion $G_{d-1}\to G$
induces a group homomorphism

\begin{equation}\label{eq:inclusion_lower_sublattice}
\iota\colon \mathrm{KO}_i(\Cstr(G_{d-1})_{\R})\to\mathrm{KO}_i(\Cstr(G_d)_{\R}).
\end{equation}

We may pick any basis for $G_{d-1}$ and $G_\perp$, so that $G_d\simeq
\Z^{d-1}\times \Z$ and \( G_{d-1} \) embeds as the \( \Z^{d-1}
\)-component. Then the computation of \cref{thm:KO-theory_real} shows
that $\iota$ maps
$\mathrm{KO}_i(\Cstr(G_{d-1}))_{\R})\simeq\mathrm{KO}_i(\Cstr(\Z^{d-1})_{\R})$ into the weak subgroup ${\rm wKO}_i(\Cstr(\Z^d)_{\R})\simeq{\rm
wKO}_i(\Cstr(G_d))_{\R})$.

An element $x\in \KO_d(\Cstr(G_d)_{\R})$ is non-trivial in the strong
quotient iff $x$ is equal to some nonzero multiple of $w_{\emptyset}$
(constructed with respect to any choice of basis for $G$) plus some weak
element. Such an $x$ cannot lie in
$\iota(\mathrm{KO}_i(\Cstr(G_{d-1})_{\R}))$ for any choice of lower-rank
lattice $G_{d-1}$. Similarly for the complex case. A coarse-geometric
version of this result can be found in
\cite{Ewert-Meyer:Coarse_geometry}*{Proposition 10}.
\end{example}

\subsection{Universal property of strong colimit}
\label{sec:universal_strong_colimit}
The strong colimit ${\rm sK}_i(\Cst(X)^\mathfrak{S})$ has an important
universal property, relevant for its physical interpretation. Let $R$ be a
reduced abelian group. Typical examples are $R\simeq\Z^n$ or $R\simeq\Z/2$,
the latter case corresponding to existence/non-existence properties such as
that of hosting ``gapless'' Dirac cone edge states in topological
insulators, which may be ``gapped out'' in pairs
\cites{Fu-Kane-Mele:TI_in_3D, Gomi-Thiang:Real_gerbes,Hasan-Moore:3D_TI}. 

We say that a family
\begin{equation}\label{eq:maps.to.reduced} 
\{f_G\colon \K_i(\Cst(X)^G)\to R\}_{G\in \mathfrak{S}}
\end{equation}
of abelian group homomorphisms is \emph{consistent with symmetry breaking}, if
\begin{equation}\label{eq:consistent.symm.break}
f_H\circ\K_i(\iota_{H\backslash
G})=f_G,\quad G,H\in\mathfrak{S},\;G\prec H.
\end{equation}

By the universal property of $\K_i(\Cst(X)^\mathfrak{S})$ as a colimit, 
there is a unique group homomorphism
\[
F\colon \K_i(\Cst(X)^\mathfrak{S})\to R
\]
such that the factorizations $f_G=F\circ \K_i(\iota_{\mathfrak{S}\backslash
G})$ hold for all $G\in\mathfrak{S}$. As $R$ is reduced, this $F$
must be trivial when restricted to the divisible subgroup ${\rm wK}_i(\Cst(X)^\mathfrak{S})$, thus $F={\rm s}F\circ q_\mathfrak{S}$
for a unique group homomorphism
\[
{\rm s}F\colon {\rm sK}_i(\Cst(X)^\mathfrak{S})\to R.
\]
In fact, we have the
commutative diagram,
\begin{equation}\label{eq:weak.strong.colimit.diagram}
\begin{tikzcd}[column sep=1.2cm]
     & & \K_i(\Cst(X)^\mathfrak{S})\ar[ddd,"q_\mathfrak{S}"']\ar[r,"F"] & R\\
    \K_i(\Cst(X)^G)\ar[r,"\K_i(\iota_{H\backslash G})"']\ar[d,"q_G"]\ar[urr,"\K_i(\iota_{\mathfrak{S}\backslash G})"]\ar[urrr,"f_G", bend left=30] & \K_i(\Cst(X)^H)\ar[d,"q_H"]\ar[ur,"\K_i(\iota_{\mathfrak{S}\backslash H})"']\ar[urr,"{\quad\;}f_H"', bend right=16, crossing over] & & \\
    {\rm sK}_i(\Cst(X)^G)\ar[r,"{\rm sK}_i(\iota_{H\backslash G})"]\ar[drr,"\mathrm{s}_{\mathfrak{S}\backslash G}"'] & {\rm sK}_i(\Cst(X)^H)\ar[dr,"\mathrm{s}_{\mathfrak{S}\backslash H}"] & &\\ 
    & & {\rm sK}_i(\Cst(X)^\mathfrak{S})\ar[uuur,"{\rm s}F"', bend right=20] &
\end{tikzcd},
\end{equation}
where $\mathrm{s}_{\mathfrak{S}\backslash G}$ are the universal maps to the strong colimit, and
\[
q_\mathfrak{S}\circ\K_i(\iota_{\mathfrak{S}\backslash G})=\mathrm{s}_{\mathfrak{S}\backslash G}\circ q_G
\]
hold for all $G\in\mathfrak{S}$ by the universal property of $\K_i(\Cst(X)^\mathfrak{S})$.
It follows that the consistent family \eqref{eq:maps.to.reduced} descends canonically to a family of maps
\begin{equation}\label{eq:strong.consistent.family}
\big\{\mathrm{s}f_G\colon {\rm sK}_i(\Cst(X)^G)\to R\big\}_{G\in\mathfrak{S}},\qquad \mathrm{s}f_G \defeq {\rm s}F\circ\mathrm{s}_{\mathfrak{S}\backslash G},
\end{equation}
which is consistent with 
symmetry breaking in the sense that 
\begin{equation}\label{eq:consistent.symm.break.strong}
{\rm s}f_H\circ{\rm sK}(\iota_{H\backslash G})={\rm s}f_G,\quad 
G,H\in\mathfrak{S},\;G\prec H.
\end{equation}
In other words, the diagrams
\[
\begin{tikzcd}[column sep=10]
    \K_i(\Cst(X)^G)\ar[dr,"q_G"]\ar[rrrr,"\K_i(\iota_{H\backslash G})"]\ar[ddrr,"f_G"',bend right=25] & {} & {} & {} & \K_i(\Cst(X)^H)\ar[dl,"q_H"']\ar[ddll,"f_H",bend left=25]\\
    {} & {\rm sK}_i(\Cst(X)^G)\ar[dr,"{\rm s}f_G"]\ar[rr,"{\rm sK}_i(\iota_{H\backslash G})"] & {} & {\rm sK}_i(\Cst(X)^H)\ar[dl,"{\rm s}f_H"'] & {}\\
    {} & {} & R & {} & {}
\end{tikzcd}
\]
commute whenever $G,H\in\mathfrak{S}$ satisfy \( G\prec H \).

\begin{remark}
In physics, an insulating
system with $\Gamma$-symmetry has a topological invariant in
$\K_i(\Cst(X)^\Gamma)$. For a reduced abelian group $R$, one would like the
system's $R$-valued ``quantized'' property to be
controlled by the invariant $\K_i(\Cst(X)^\Gamma)$ via some homomorphism
$f_\Gamma\colon \K_i(\Cst(X)^\Gamma)\to R$. For example, the
integer-quantized Hall conductance is predicted from
$\K_0(\Cst(X)^{\Gamma})$ where \( \Gamma\simeq\Z^2
\) and \( X=\R^2 \) or \( X=\Z^2 \); 
see \cref{sec:Landau_level}, \cref{rem:physics.Chern.Hall} for an explicit
Landau operator model. A more abstract example is $R=\K_i(\Cst(X))$ with the map $\K_i(\Cst(X)^\Gamma)\to \K_i(\Cst(X))$ induced by inclusion, which will be studied in Section \ref{sec:comparison.with.non.equivariant}.

We are free to regard
the same system as possessing only $G$ symmetry, for any choice of
finite-index subgroup $G\subseteq \Gamma$. So we must actually have a whole
\emph{family} \eqref{eq:maps.to.reduced} of homomorphisms to $R$, which is
consistent with symmetry breaking in the sense of
\eqref{eq:consistent.symm.break}. Then the commutative diagram
\eqref{eq:weak.strong.colimit.diagram} says that only the strong quotients
contribute to quantized properties, consistently in the sense of \eqref{eq:consistent.symm.break.strong}. In contrast, the weak summands are
irrelevant. It is therefore of interest to understand when strong quotients are non-trivial.
\end{remark}

\subsection{Strong-to-coarse comparison maps}
\label{sec:comparison.with.non.equivariant}

Non-equivariant Roe algebras were introduced as observable \Cst-algebras
for topological insulators by Kubota
\cite{Kubota:Controlled_topological_phases}, and Ewert--Meyer
\cite{Ewert-Meyer:Coarse_geometry} and have already found many
applications
\cites{Ludewig-Thiang:Gapless_hyperbolic_half-plane,Ludewig-Thiang:Large-scale_obstructs_localization,Ludewig-Thiang:Quantization_coarse_cohomology,Kubota-Ludewig-Thiang:Delocalized_helical_surfaces,Kordyukov-Manuilov:Vanishing_theorem_magnetic,Rossi-Panati:Algebraic_localization_Roe_triviality,Ludewig-Thiang:Large-scale_quantization_1}.
Related ideas about coarse-graining and topological phases had also been
circulated earlier,
see~\citelist{\cite{Gross-Nesme-Vogts-Werner:Quantum_walks}*{Section 8.2}
\cite{Kitaev:Anyons}*{Appendix C}}, as well as T-duality ideas in \cite{Mathai-Thiang:T-duality}. As sketched in \cite{Witten:Lectures_on_topological_phases}*{Section 2.4--2.5}, for the integer quantized Hall conductance, the lattice-dependent TKNN--Kubo--Chern number \cite{TKNN:Quantized_Hall_conductance} approach (see \cref{rem:physics.Chern.Hall}) should be invariant under coarse-graining, so as to make contact with an effective Chern--Simons theory model in the large-scale limit. We also mention that symmetry breaking from a group $G$
to an arbitrary (possibly trivial) subgroup was considered in the context
of equivariant coarse homology theories in
\cite{Bunke-Ludewig:Breaking_symmetries}, with applications to spectral
theory of invariant differential operators on manifolds with boundary.

The general idea is that topological phases described by K-theory of 
Roe algebras remain protected even after
perturbation by disorder of a rather general nature. 
Due to this reason, Ewert and Meyer called such topological phases
``strong''. We prefer to call them \emph{coarse} here, so as to
maintain an a priori distinction from the \emph{strong} (\( G\)-invariant) phases in the sense of our \cref{dfn:weak_strong}.

\medskip
Let \( \mathcal{H} \) be an ample \( X \)-\( \Gamma \)-module and \(
\mathfrak{S} \) the symmetry breaking poset of \( \Gamma \). For each \(
G\in\mathfrak{S} \),
forgetting \( G \)-invariance gives
an inclusion of \Cst-subalgebras of \( \Bdd(\mathcal{H}) \), 
\[ 
  \iota_{1\backslash G}\colon
  \Cst(X,\mathcal{H})^G\hookrightarrow \Cst(X,\mathcal{H}).
\]
If \( G\prec H \), then clearly \( \iota_{1\backslash
G}=\iota_{1\backslash H}\circ\iota_{H\backslash G} \). So the family of group
homomorphisms 
\begin{equation}\label{eq:equiv.to.coarse}
\K_i(\iota_{1\backslash G})\colon \K_i(\Cst(X,\mathcal{H})^G)\to \K_i(\Cst(X,\mathcal{H})),\qquad G\in\mathfrak{S},
\end{equation}
is compatible with symmetry breaking in the sense of
\eqref{eq:consistent.symm.break}. Provided that the abelian group
$R\defeq\K_i(\Cst(X,\mathcal{H}))$ is known to be reduced (which
is typically the case), the
maps \eqref{eq:equiv.to.coarse} descend to maps
\begin{equation}\label{eq:strong.to.coarse.comparison.maps}
{\rm sK}_i(\iota_{1\backslash G})\colon {\rm sK}_i(\Cst(X,\mathcal{H})^G)\to \K_i(\Cst(X,\mathcal{H})),\qquad G\in\mathfrak{S}
\end{equation}
consistent with symmetry breaking.

\begin{definition}\label{dfn:strong.to.coarse.comparison}
Let $\mathcal{H}$ be an ample $X$-$\Gamma$ module satisfying the
\refbasicsetup.
Assume that $\K_i(\Cst(X,\mathcal{H}))$ is a reduced abelian group. 
\begin{enumerate}
\item The maps \eqref{eq:strong.to.coarse.comparison.maps} are called the
\emph{strong-to-coarse comparison maps}. 
\item The universal map
\begin{equation}\label{eq:strong.colimit.to.coarse.comparison.map}
\sK_i(\iota_{1\backslash \mathfrak{S}})\colon \sK_i(\Cst(X,\mathcal{H})^\mathfrak{S})\to \K_i(\Cst(X,\mathcal{H}))
\end{equation}
is called the \emph{strong-colimit-to-coarse comparison map}.
\end{enumerate}
\end{definition}

Our goal is to analyze the comparison maps of \cref{dfn:strong.to.coarse.comparison}, particularly, their surjectivity. 
Note that the comparison maps
\eqref{eq:strong.to.coarse.comparison.maps} and
\eqref{eq:strong.colimit.to.coarse.comparison.map}
make sense whether or not $\mathcal{H}$ is ample. For our analysis below,
it is convenient to use an ample $\mathcal{H}$ in the target Roe algebra
$\Cst(X,\mathcal{H})$, so that it is independent of the choice of
$\mathcal{H}$ (see \cref{rmk:equivariant_coarse_equivalence}).

Unfortunately, there
are very few methods to compute the K(O)-theory of Roe algebras.  The
most powerful general method is probably the coarse Baum--Connes conjecture
\cite{Higson-Roe:Coarse_Baum-Connes}, which has been verified for large
classes of spaces $X$.  For example, it has been employed to prove the
gaplessness of certain magnetic Hamiltonians on hyperbolic half-planes
\cite{Ludewig-Thiang:Gapless_hyperbolic_half-plane} and helical surfaces
\cite{Kubota-Ludewig-Thiang:Delocalized_helical_surfaces}.  In
\cref{sec:comparison.strong.coarse} below, we will provide a new kind of
application of the Baum--Connes conjecture.

\subsubsection{Mal'cev completion}
\label{sec:Malcev_completion}

Recall from \cref{thm:main_computation_theorem} that the colimit
$\K_i(\Cst(X,\mathcal{H})^\mathfrak{S})$ is controlled by the discrete
group $\Gamma$ (and 2-cocycle $\sigma$), whereas the choice of $X$ is
secondary. With this flexibility in $X$ in mind, let us recall a well-known
procedure to construct a \emph{Lie group} $\Gamma_\R$ from $\Gamma$, valid
for a fairly broad range of $\Gamma$. The homogeneity of the choice
$X=\Gamma_\R$ will, in turn, be exploited for computation.
 
Let \( \Gamma \) be a finitely generated, torsion free, nilpotent discrete
group. Then it is known that \( \Gamma \) is isomorphic to a cocompact lattice
in a connected, simply connected, nilpotent Lie group \( \Gamma_\R \), see
\cite{Raghunathan:Discrete_subgroups_Lie_groups}*{Theorem 2.18}.  The Lie
group \( \Gamma_\R \) is called the (real) \emph{Mal'cev completion} of \(
\Gamma \). We recall its construction, following
\cite{Dekimpe:Intra-nilmanifolds}*{Section 2}.

Under the above assumptions on
$\Gamma$, there exists a finite decreasing series of normal
subgroups of \( \Gamma \),
\[ 
  \Gamma=\Gamma_1\supsetneqq \Gamma_2\supsetneqq \dots\supsetneqq
  \Gamma_k\supsetneqq \Gamma_{k+1}=1,
\]
such that \( \Gamma_i/\Gamma_{i+1}\simeq\Z \) and \(
[\Gamma,\Gamma_{i}]\subseteq \Gamma_{i+1} \). 
A \emph{Mal'cev basis} of \( \Gamma \) for such a series is 
a set of elements
\[ 
  \{a_1,\dots,a_k\},\quad a_i\in \Gamma_i,
\]
such that for each \( i \), 
the image of \( a_i \) in \( \Gamma_i/\Gamma_{i+1}\simeq\Z \) is
a generator. Fix a Mal'cev basis of \( \Gamma \). Then every element \(
a\in\Gamma \) can be uniquely expressed as
\begin{equation}\label{eq:Malcev_basis_expansion}
  a=a_1^{x_1}a_2^{x_2}\dots a_k^{x_k}.
\end{equation}
where \( x_1,\dots,x_k\in\Z \).

Moreover, there exist polynomials with \( \Q \)-coefficients, 
\begin{gather*} 
p_1=0,\quad q_1=0,\\
p_i(x_1,x_2,\dots,x_{i-1},y_1,y_2,\dots,y_{i-1}),\quad
q_i(x_1,x_2,\dots,x_{i-1},z),\quad 2\leq i\leq k,
\end{gather*}
such that for $x_1,x_2, \dots, x_k,y_1,\dots,y_k,z\in \Z$ we have the
following equations
\begin{equation}\label{eq:Malcev_basis_polynomial}
\begin{gathered}
\prod_{i=1}^ka_i^{x_i}\prod_{i=1}^ka_i^{y_i}=\prod_{i=1}^ka_i^{x_i+y_i+p_i(x_1,x_2,\dots,x_{i-1},y_1,y_2,\dots,y_{i-1})},\\
(\prod_{i=1}^ka_i^{x_i})^z=\prod_{i=1}^ka_i^{zx_i+q_i(x_1,x_2,\dots,x_{i-1},z)}.
\end{gathered}
\end{equation}

\begin{definition}\label{def:real_Malcev_completion}
The \emph{real Mal'cev completion} of \( \Gamma \), denoted by \( \Gamma_\R
\), is the group generated by elements of the form
\eqref{eq:Malcev_basis_expansion} for \( x_1,\dots,x_k\in\R \), and
products given by \eqref{eq:Malcev_basis_polynomial} for the polynomials \(
p_i\) determined by \( \Gamma \) and its Mal'cev basis.
\end{definition}

The simplest example is \( (\Z^d)_\R\simeq\R^d \). The basic non-abelian
example is the (3-dimensional) \emph{integer Heisenberg group}
\begin{equation}\label{eq:Heisenberg_integer_lattice}
\HH\defeq \begin{pmatrix}1 & \Z & \Z \\ 0 & 1 & \Z \\ 0 & 0 & 1\end{pmatrix}=\left\{\begin{pmatrix}1 & x & z \\ 0 & 1 & y \\ 0 & 0 & 1\end{pmatrix}\,:\, x,y,z\in \Z\right\},
\end{equation}
which has $\HH_\R$ being the Heisenberg Lie group
\[
\HH_\R\defeq \begin{pmatrix}1 & \R & \R \\ 0 & 1 & \R \\ 0 & 0 & 1\end{pmatrix}=\left\{\begin{pmatrix}1 & x & z \\ 0 & 1 & y \\ 0 & 0 & 1\end{pmatrix}\,:\, x,y,z\in \R\right\}.
\]

\subsubsection{Strong-to-coarse comparison map}\label{sec:comparison.strong.coarse}

For the twisted case of the next result, we need $\sigma$ to be a magnetic
2-cocycle, as defined in \cref{dfn:magnetic.2.cocycle} of
\cref{sec:magnetic_translations}.

\begin{theorem}\label{thm:existence_strong_phase}
Let \( \Gamma \) be a finitely generated, torsion
free, nilpotent discrete group, and \( \mathfrak{S} \) be its symmetry breaking poset. 
Let \( \mathcal{H} \) be an ample complex $\sigma$-twisted \( X \)-\(
\Gamma \)-module satisfying \refbasicsetup. 
We assume that $\sigma$ is a magnetic 2-cocycle for $(\Gamma_\R,\Gamma)$.
Then for every \( G\in\mathfrak{S} \), 
the strong-to-coarse comparison maps
\eqref{eq:strong.to.coarse.comparison.maps} are surjective for every degree \( i \).
\end{theorem}

\begin{proof}
For any $\sigma$-twisted $X$-$\Gamma$ module $\mathcal{H}$ in the \refbasicsetup, \cref{thm:main_computation_theorem} gives a natural identification of the symmetry breaking directed systems
\begin{align*}
\K_i(\iota_{H\backslash G})\colon \K_i(\Cst(X,\mathcal{H})^G)&\to \K_i(\Cst(X,\mathcal{H})^H),\\
\phi_{H\backslash G}\colon \K_i(\Cstr(G,\overline{\sigma}))&\to \K_i(\Cstr(H,\overline{\sigma})).
\end{align*}
Also, the choice of ample $\mathcal{H}$ is unimportant, in view of \cref{rmk:equivariant_coarse_equivalence}. So, without loss of generality, we may take $X=\Gamma_\R$ to be the Mal'cev completion equipped with a left-invariant Riemannian metric. We will also take $\mathcal{H}$ to be the Hilbert space of $L^2$-sections of $\mathcal{S}\otimes \mathcal{L}$, where $\mathcal{S}\to \Gamma_\R$ is the spinor bundle, and $\mathcal{L}\to\Gamma_\R$ is a Hermitian line bundle whose covariant derivative commutes with a $\sigma$-projective representation of $\Gamma$ (as described in \cref{sec:magnetic_translations}).

As \( \Gamma_\R \) is a connected, simply connected nilpotent Lie group, it is diffeomorphic to $\R^d$ (\cref{lem:nilpotent.is.Rn}), therefore,
\[
  \K_i(\Gamma_\R)\simeq\K_i(\R^d)\simeq
  \begin{cases}
  \Z, & \text{if \( i-d\equiv 0\mod 2 \);} \\
  0, & \text{if \( i-d\equiv 1\mod 2 \).} 
  \end{cases}
\]
The assembly map of the Baum--Connes conjecture \cite{Higson-Roe:Coarse_Baum-Connes}
\begin{equation}\label{eq:BC.conjecture}
  \mu\colon \K_i(\Gamma_\R)\to\K_i(\Cst(\Gamma_\R,\mathcal{H})),
\end{equation}
relates $\K_i(\Gamma_\R)$ to $\K_i(\Cst(\Gamma_\R,\mathcal{H}))$.
It follows from \cref{cor:contractible_nilpotent_coarse_Baum-Connes} that
$\mu$ is an isomorphism $\Z\xrightarrow{\sim}\Z$ if $i-d\equiv 0\mod 2$, and \( 0\xrightarrow{\sim}0 \) if $i-d\equiv 1\mod 2$. 

Surjectivity of \( \sK_{i}(\iota_{1\backslash G}) \) is thus automatic for
\( i-d\equiv 1\mod 2 \). It remains to prove surjectivity for \( i=d \). To
this end, we describe a generator for each side of
\eqref{eq:BC.conjecture}, as follows. It is well-known that the spin Dirac
operator on Euclidean $\R^d$ represents a generator of
$\K_d(\R^d)\simeq\Z$.  This follows from, say, an induction on the
dimension \( d \) using \cite{Higson-Roe:Analytic_K-homology}*{Proposition
11.2.5}.  As $\R^d$ is contractible, there is a unique (trivial) spin
structure, and by interpolating between the Euclidean metric and the metric
of \( \Gamma_\R \), we get a concordance of their spinor bundles in the
sense of \cite{Higson-Roe:Analytic_K-homology}*{Definition 11.2.6}, and by
\cite{Higson-Roe:Analytic_K-homology}*{Proposition 11.2.7}, the spin Dirac
operator $D_\mathcal{S}$ for $\Gamma_\R$ represents the same $\K$-homology
class as the spin Dirac operator on Euclidean $\R^d$. Thus \(
\K_d(\Gamma_\R)\simeq\Z \) is also generated by the K-homology class \(
[D_\mathcal{S}] \).  We may replace $D_\mathcal{S}$ by the twisted spin
Dirac operator $D_{\mathcal{S}\otimes\mathcal{L}}$, which has the same
principal symbol so it represents the same $\K$-homology class as the
untwisted one (e.g., \cite{Higson-Roe:Analytic_K-homology}*{Exercise
10.9.5}). By definition, the assembly map \( \mu \) sends
$[D_{\mathcal{S}\otimes\mathcal{L}}]$ to its \emph{coarse index} \(
\ind_\mathrm{c}(D_{\mathcal{S}\otimes\mathcal{L}})
\in\K_d(\Cst(\Gamma_\R,\mathcal{H})) \), so the latter is a generator
as well. 

As the metric on $\Gamma_\R$ is left-invariant, the (proper) translation action $\Gamma\curvearrowright \Gamma_\R$ is isometric and orientation-preserving, and it lifts to a unitary representation $T$ of $\Gamma$ on
$L^2(X;\mathcal{S})$ commuting with the (Levi--Civita) spin connection, thus with $D_\mathcal{S}$. The twisted $D_{\mathcal{S}\otimes\mathcal{L}}$ is obtained from $D_\mathcal{S}$ by tensoring the spin connection with the connection $\nabla$ on $\mathcal{L}$. If we tensor $T_g$ with the unitary magnetic translation operators $U_g$ of \cref{sec:magnetic_translations}, we get (after adjusting phases of $U_g$ if necessary) a $\sigma$-projective unitary representation of $\Gamma$ on $L^2(X;\mathcal{S}\otimes\mathcal{L})$, which commutes with $D_{\mathcal{S}\otimes\mathcal{L}}$. We remark that $\mathcal{H}=L^2(X;\mathcal{S}\otimes\mathcal{L})$ becomes an ample $\sigma$-twisted $\Gamma_\R$-$\Gamma$-module in this way.

Now, $\Gamma_\R$ is a complete Riemannian spin manifold, so
$D_{\mathcal{S}\otimes\mathcal{L}}$ is essentially self-adjoint on the
(dense, $\Gamma$-invariant) subspace of compactly-supported smooth sections
of $\mathcal{S}\otimes\mathcal{L}$ (see
\cite{Wolf:Essential_self-adjointness}). The coarse index
class $\ind_\mathrm{c}(D_{\mathcal{S}\otimes\mathcal{L}})$ has an operator
representative in $\Cst(\Gamma_\R,\mathcal{H})$ constructed by
continuous functional calculus
\cite{Roe:Index_theory_coarse_geometry}*{Chapter 3}. This representative
commutes with the twisted $\Gamma$-representation, thus it actually lies in
$\Cst(\Gamma_\R,\mathcal{H})^\Gamma$. Therefore,
$\ind_\mathrm{c}(D_{\mathcal{S}\otimes\mathcal{L}})$ lifts to a
$\Gamma$-equivariant class in
\( \K_i(\Cst(\Gamma_\R,\mathcal{H})^\Gamma) \). Likewise, there is also a
lift to \( \K_i(\Cst(\Gamma_\R,\mathcal{H})^G) \) for every
$G\in\mathfrak{S}$.

We have verified that the maps
\[
\K_i(\iota_{1\backslash G})\colon \K_i(\Cst(\Gamma_\R,\mathcal{H})^G)\to \K_i(\Cst(\Gamma_\R,\mathcal{H}))\simeq\Z,\qquad G\in\mathfrak{S},
\]
are surjective. As $\K_i(\Cst(\Gamma_\R,\mathcal{H}))\simeq\Z$ is a reduced group, these maps descend to surjective comparison maps,
\[
{\rm sK}_i(\iota_{1\backslash G})\colon {\rm
sK}_i(\Cst(\Gamma_\R,\mathcal{H})^G)\to
\K_i(\Cst(\Gamma_\R,\mathcal{H}))\simeq\Z,\qquad G\in\mathfrak{S}.\qedhere
\]
\end{proof}

\begin{remark}\label{rem:other.examples.surjective.comparison}
The proof of \cref{thm:existence_strong_phase} also works if $\Gamma$ is a
lattice in some contractible Lie group $Y$ (with left-invariant metric) for
which the assembly map is an isomorphism. It also works for the surface
groups $\Gamma_g$ (fundamental group of compact orientable surface of genus
$g\geq 2$) acting on the (negatively-curved) hyperbolic plane, because in
this case, the coarse assembly map is an isomorphism $\Z\to\Z$ in degree
$i=0$ \cite{Higson-Roe:Coarse_Baum-Connes}*{Corollary 7.4}.

However, in view of
\cref{rmk:equivariant_coarse_equivalence,rem:micro.geom.irrelevant}, we
prefer to provide a concrete,
yet fairly general criterion on the abstract group $\Gamma$ without a
priori knowledge of a ``nice'' manifold $X$ on which $\Gamma$ acts. 
\end{remark}

\begin{remark}
In \cref{sec:Euclidean.magnetic.translations}, it is shown that \emph{any} 2-cocycle $\sigma$ on $\Gamma\simeq\Z^d$ is a magnetic 2-cocycle for $(\R^d,\Z^d)$.

For general $\Gamma$, the 2-cocycles on $\Gamma$ are harder to understand. We do not know general criteria for a 2-cocycle on $\Gamma$ to be a magnetic 2-cocycle. Certainly, it is necessary that $\sigma$ can be deformed to a cohomologically trivial 2-cocycle, since the connection 1-form can be taken to 0. 
\end{remark}

\begin{remark}
For the basic non-abelian example, $\Gamma=\HH$, it follows immediately
from \cref{thm:existence_strong_phase} that all the strong quotients
\eqref{eq:strong.quotient} contain a $\Z$-summand, identified with
$\K_1(\Cst(\HH_\R))$ by the strong-to-coarse comparision maps.
In \cref{sec:Heisenberg}, we will carry out a more detailed
computation, and show, in particular, that these strong quotients are
isomorphic to $\Z$. Thus, we may say that strong topological invariants
coincide with coarse ones in this case.
\end{remark}

\begin{remark}\label{rmk:untwisted_real_abelian_strong_to_coarse}
\cref{thm:existence_strong_phase} still holds if \( \Gamma \) is a
finitely generated abelian group, and \( \mathcal{H} \) is an ample,
untwisted real \( X \)-\( \Gamma \)-module satisfying \refbasicsetup. This
result relies on a real version of the Kasparov product
\cite{Kasparov:Equivariant_KK-theory} and follows essentially from
\cite{Ewert-Meyer:Coarse_geometry}*{Theorem 7 and Corollary 5}.

We identify \( \Gamma \) with \( \Z^d \) by choosing a basis.
There is a commutative diagram of KO-theory groups for each \( i \) and
each \( d \) (compare the commutative diagram of
\cite{Ewert-Meyer:Coarse_geometry}*{Theorem 7}):
\[ \begin{tikzcd}
\KO_i(\Cst(\R^d)^{\Z^d}_\R) \arrow[rr, "\alpha_{i,d}"] 
\arrow[rd, "\KO_{i}(\iota_{1\backslash \Z^d})"'] 
&& \KO_{i-d}(\R) \\ 
&  \KO_i(\Cst(\R^d)_\R) \arrow[ru, "\alpha'_{i,d}"'] 
\end{tikzcd} \]
Here, \( \alpha_{i,d} \) is induced by the ``Dirac element''
of the real \Cst-algebra 
\( \Cstr(\Z^d)_\R\simeq\Cont(\T^d,\tau_{\T^d})
\) in Kasparov's bivariant K-theory (see
\citelist{\cite{Kasparov:Equivariant_KK-theory}*{Definition and Lemma 4.2}
\cite{Bourne-Carey-Lesch-Rennie:KO-valued_spectral_flow}*{Lemma 9.1}}). 
For \( i=d \), one checks that \( \alpha_{d,d} \) maps the class \(
w_{1,\emptyset} \) (the ``strong'' generator) to a generator of \(
\KO_0(\R)\simeq\Z \), and that \( \alpha_{d,d}' \) is an isomorphism \(
\Z\xrightarrow{\sim}\Z \). This implies that \( \KO_i(\iota_{1\backslash
\Z^d}) \) sends
the class of \( w_{1,\emptyset} \) to a generator of
\( \KO_d(\Cst(\R^d)_\R)\simeq\Z \).

The K-theory
maps \( (\alpha_{i,d})_{i\in\Z} \) and \( (\alpha'_{i,d})_{i\in\Z} \) 
are induced by Kasparov products, hence they give 
\( \KO_*(\R) \)-module homomorphisms
\[
\alpha_{*,d}\colon \KO_*(\Cst(\R^d)^{\Z^d}_\R)\to\KO_{*-d}(\R^d)
\quad\text{and}\quad 
\alpha'_{*,d}\colon \KO_*(\Cst(\R^d)_\R)\xrightarrow{\sim}\KO_{*-d}(\R). 
\]
Thus
\( \KO_*(\Cst(\R^d)_\R) \) is a free \( \KO_*(\R) \)-module generated by a
single generator in degree \( +d \), and \( \KO_0(\iota_{1\backslash \Z^d})
\) maps \( w_{1,\emptyset} \) to this free \( \KO_*(\R) \)-module
generator.  Therefore, every element of \( \KO_*(\Cst(\R^d)_\R) \) belongs
to the image of the \( \KO_*(\R) \)-module map \( \KO_*(\iota_{1\backslash
\Z^d}) \). Since \( \KO_i(\Cst(\R^d)_\R)\simeq\KO_{i-d}(\R) \) 
is reduced, it follows that \(
\mathrm{sKO}_i(\iota_{1\backslash \Z^d}) \) is surjective for every \( i
\).
\end{remark}

\subsubsection{A non-surjective comparison map}
\label{sec:crystallographic.generalization}
A \( d \)-dimensional crystallographic group is a discrete cocompact
subgroup of \( \Isom(\R^d) \), the isometry group of the \(d\)-dimensional
Euclidean space. It is known that every \(d\)-dimensional crystallographic
group $\widetilde{\Gamma}$ is isomorphic to an extension of a finite group
$F$ by a lattice $\Gamma\subseteq\R^d$, i.e., there is a short exact
sequence of groups
\begin{equation}\label{eq:crystal.extension}
\begin{tikzcd}
1 \arrow[r] & \underbrace{\Gamma}_{\simeq\Z^d} \arrow[r] 
& \widetilde{\Gamma} \arrow[r] & \underbrace{F}_{\text{finite}}
\arrow[r] & 1.
\end{tikzcd}
\end{equation}
Crystallographic groups are of
great importance in solid-state physics, and we would like an analogue of
\cref{thm:existence_strong_phase} for this case. If the
crystallographic group $\widetilde{\Gamma}$ acts on $\R^2$ by
orientation-preserving isometries, we can regard the spin Dirac operator on
$\R^2$ as a $\widetilde{\Gamma}$-invariant operator, and the proof of
\cref{thm:existence_strong_phase} works in the same way.
However, if $\widetilde{\Gamma}$ has orientation-reversing elements, the
surjectivity of the strong-to-coarse comparison map can fail, as the
following example shows.

Let \(\widetilde{\Gamma}\defeq\Z\rtimes\Z\) be the semidirect product given by the unique non-trivial action of \(\Z\) on \(\Z\). It fits into the non-split short exact sequence
\[
\begin{tikzcd}[row sep=0]
1\ar[r] & \Z^2 \ar[r] & \Z\rtimes\Z \ar[r] & \Z/2 \ar[r] & 1\\
{} & (a,b)\ar[r,mapsto] & (a,2b) & {} & {} \\
{} & {} & (a,b) \ar[r,mapsto] & b\!\!\!\mod 2\Z. & {}
\end{tikzcd}
\]
The group generator $(1,0)\in\Z\rtimes\Z$ acts on the Euclidean plane
$X=\R^2$ as $(x,y)\mapsto (x+1,y)$, whereas the generator
$(0,1)\in\Z\rtimes\Z$ acts as the orientation-reversing glide reflection
$(x,y)\mapsto (-x,y+\tfrac{1}{2})$. This exhibits $\Z\rtimes\Z$ as a 2D
crystallographic group, commonly called $\mathsf{pg}$ in the physics
literature. Below, we explain why the strong-to-coarse comparison map is
\emph{not} surjective in this case.

There is an isomorphism
\begin{equation}\label{eq:pg.crossed.product}
\Cstr(\Z\rtimes\Z)\simeq\Cstr(\Z)\rtimes_\alpha\Z,
\end{equation}
where $\alpha$ is the automorphism of $\Cstr(\Z)$ mapping $w_1=\rho_\Z(1)$ to $w_1^{-1}=\rho_\Z(-1)$. Thus,
\[
\K_1(\alpha)\colon \K_1(\Cstr(\Z))\to\K_1(\Cstr(\Z)),\qquad [w_1]\mapsto -[w_1],
\]
and obviously,
\[
\K_0(\alpha)\colon \K_0(\Cstr(\Z))\to\K_1(\Cstr(\Z)),\qquad [1]\mapsto [1].
\]
Thus, the Pimsner--Voiculescu sequence \cite{Pimnser-Voiculescu:Exact_sequence} for the crossed product \eqref{eq:pg.crossed.product} is
\[
\begin{tikzcd}
\Z\simeq\K_0(\Cstr(\Z))\ar[r,"0"] & \K_0(\Cstr(\Z)) \ar[r] & \K_0(\Cstr(\Z)\rtimes_\alpha\Z)\ar[d,"\partial_{\rm PV}"]\\
\K_1(\Cstr(\Z)\rtimes_\alpha\Z)\ar[u,"\partial_{\rm PV}"] & \K_1(\Cstr(\Z))\ar[l] & \K_1(\Cstr(\Z))\simeq\Z\ar[l,"\times 2"]
\end{tikzcd}.
\]
It follows from exactness that
\[
\K_0(\Cstr(\Z\rtimes\Z))\simeq\K_0(\Cstr(\Z)\rtimes_\alpha\Z)\simeq\Z,
\]
with generator $[1]$. 

Now, symmetry breaking from $\Z\rtimes\Z$ to its index-2 subgroup $\Z^2$
effects $[1]\mapsto 2\cdot[1]$, and $[1]$ is a weak class in
$\K_0(\Cstr(\Z^2))$ (see \cref{sec:complex_d-dim}). It follows that
$\wK_0(\Cstr(\Z\rtimes\Z))=\K_0(\Cstr(\Z\rtimes\Z))$ 
so
\[
  \sK_0(\Cstr(\Z\rtimes\Z))=0.
\]
However, for $X=\R^2$, we know that $\K_0(\Cst(X))\simeq\Z$. Thus the
strong-to-coarse comparison map cannot be surjective (it is the zero
map). The proof in \cref{thm:existence_strong_phase} does not work because
$(0,1)\in\Z\rtimes\Z$ does not preserve the orientation of $\R^2$, so the
Dirac operator is not $\Z\rtimes\Z$-invariant. The strong-to-coarse
comparison map only becomes surjective after we break the symmetry group
from $\Z\rtimes\Z \to \Z^2$.

\medskip
Nevertheless, the following colimit analogue of
\cref{thm:existence_strong_phase} holds, in particular, for all
crystallographic groups.

\begin{corollary}
\label{cor:existence.strong.colimit}
Let \( \widetilde{\Gamma} \) be an extension of a finite group $F$ by a
finitely generated, torsion-free, nilpotent discrete group $\Gamma$.  Let
\( \mathcal{H} \) be an ample complex $\sigma$-twisted \( X \)-\(
\widetilde{\Gamma} \)-module satisfying \refbasicsetup.  We assume that
$\sigma$ is a magnetic 2-cocycle for $(\Gamma_\R,\widetilde{\Gamma})$.
Then the strong-colimit-to-coarse comparison maps is surjective in
every degree \( i \).
\end{corollary}
\begin{proof}
Note that $\sigma$ restricts to a magnetic 2-cocycle for $\Gamma$, and that
$\mathcal{H}$ is also an ample $\sigma$-twisted \( X \)-\( \Gamma \)-module
satisfying \refbasicsetup. So the conclusion of
\cref{thm:existence_strong_phase} applies, saying that for all finite index
$G\subseteq \Gamma$, the strong-to-coarse comparison maps are surjections.
Write $\mathfrak{S}$ for the symmetry breaking poset $\Gamma$. By the
universal property of the strong colimit, we also have
\[
\sK_i(\iota_{1\backslash\mathfrak{S}})\colon \sK_i(\Cst(X,\mathcal{H})^\mathfrak{S})\to \K_i(\Cst(X,\mathcal{H}))
\]
being a surjection.

Write $\widetilde{\mathfrak{S}}$ for the symmetry breaking poset of $\widetilde{\Gamma}$. Note that $\widetilde{\mathfrak{S}}$ has $\mathfrak{S}$ as a cofinal subset. It follows that the (strong) colimit over $\widetilde{\mathfrak{S}}$ is isomorphic to that over $\mathfrak{S}$, thus
\[
\sK_i(\iota_{1\backslash\widetilde{\mathfrak{S}}})\colon \sK_i(\Cst(X,\mathcal{H})^{\widetilde{\mathfrak{S}}})\to \K_i(\Cst(X,\mathcal{H}))
\]
is a surjection as well.
\end{proof}

\begin{remark}
Let \( H=-\Delta+V \) be the Schr\"odinger operator acting on \( L^2(\R^2)
\), for a bounded, \( \widetilde{\Gamma}\)-invariant potential function \(
V \), where \( \widetilde{\Gamma} \) is a 2D crystallographic group as in
\eqref{eq:crystal.extension}. When the Fermi energy $E\in\R$ is in a
spectral gap of $H$, the Fermi projection $P_{\leq
E}=\chi_{(-\infty,E]}(H)$ is $\widetilde{\Gamma}$-invariant and belongs to
the equivariant Roe algebra \( \Cst(\R^2)^{\widetilde{\Gamma}} \). It is
convenient to temporarily consider \( H \), thus also $P_{\leq E}$, as
being only \( \Gamma \)-symmetric for the finite-index abelian subgroup \(
\Gamma\simeq\Z^2 \). Then, as mentioned in the Introduction, we may apply
the Bloch transform to obtain (after Fourier transform) a vector bundle
$\mathcal{E}$ over the Brillouin torus $\widehat{\Gamma}\simeq\T^2$, and
this has a first Chern class in $H^2(\T^2)\simeq\Z$.  This class is
identified with the reduced $\K$-theory class in $\widetilde{\K}^0(\T^2)$
under the Chern character, and generates the strong summand in
$\K^0(\T^2)\simeq\K_0(\Cstr(\Z^2))$ in our terminology. 

Thus far, symmetry under the lifts of $g\in F$ in $\widetilde{\Gamma}$ has
not been accounted for. This extra symmetry constrains $\mathcal{E}$ to be
(twisted) $F$-equivariant, where the $F$-action on $\T^2$ is that
dual to the $F$-action on $\Z^2$ associated to the extension
\eqref{eq:crystal.extension}; see
\cite{Freed-Moore:Twisted_equivariant_matter} for details. When $F$ has an
orientation-reversing element, the action of this element on $\T^2$ is
orientation-reversing. Consequently, the first Chern class of
$\mathcal{E}$, thus the $\widetilde{\K}^0(\T^2)$ class, vanishes. In other
words, after breaking the symmetry $\widetilde{\Gamma}$ to
$\Gamma\simeq\Z^2$, the class of $P_{\leq E}$ becomes trivial in the strong
quotient, thus it must vanish in $\K_0(\Cst(\R^2))$.  
\end{remark}

\section{Non-abelian example: Integer Heisenberg group}
\label{sec:Heisenberg}

In this section, we shall study weak and strong phases for the 
basic non-abelian example: the integer Heisenberg group \( \HH \) (see
\eqref{eq:Heisenberg_integer_lattice}).
Compared to the abelian case of 
$\Gamma\simeq\Z^d$, there are several complications:

\begin{enumerate}
\item Finite index subgroups of $\HH$ are not mutually isomorphic. For example, if $m_1,m_2,M\in\N_{\geq 1}$, then
\begin{equation}\label{eq:Heisenberg_integer_basic_subgroups}
\begin{pmatrix}1 & m_1\Z & \frac{m_1m_2}{M}\Z \\ 0 & 1 & m_2\Z \\ 0 & 0 & 1\end{pmatrix}
\end{equation}
are finite index subgroups of $\HH$, but they are non-isomorphic for
distinct $M$. In fact, their group $\Cst$-algebras have non-isomorphic
$\K$-theories, because an extra $M$-torsion $\K$-theory
class appears, see
\cite{Hannabuss-Mathai-Thiang:T-duality_parametrised}*{Section 5.1.1}.

\item While a general finite index subgroup $G\subseteq \HH$ is abstractly isomorphic to one of the form \eqref{eq:Heisenberg_integer_basic_subgroups}, see
\cite{Lee-Packer:Twisted_group_Cst-algebras}*{Section 1, Page 93}, the symmetry breaking maps
$\phi_{H\backslash G}$ are defined by the subgroup structure, and this is very complicated in general.
\end{enumerate}

Nevertheless, as will be shown in \cref{prop:Heis.cofinal}, there is a
cofinal subset of the poset of all finite-index subgroups of \( \HH \),
such that every  \( G\subseteq \HH \) belonging to this cofinal
subset has the following convenient properties.
First, \( G \) is isomorphic to
the ``standard'' integer Heisenberg group $\HH$, so that no extra torsion
K-theory elements appear. Second, their associated Pimsner--Voiculescu
sequences
allow us to explicitly describe the K-theory generators for these \( \Cst(G)
\),
as well the symmetry breaking maps between them.

\subsection{K-theory of integer Heisenberg group algebra}

Let $\Gamma=\HH$ be the integer Heisenberg group 
(see \eqref{eq:Heisenberg_integer_lattice}), and write
\[
\mathfrak{H}=\{G\subseteq \HH\,:\,G\;\text{has\;finite\;index\;in}\;\HH\}
\]
for its symmetry breaking poset. 

\begin{proposition}\label{prop:Heis.cofinal}
The collection 
\begin{equation}\label{eq:Heis_cofinal_subset}
\left\{m\HH\defeq \begin{pmatrix} 
1 & m\Z & m^2\Z \\ 0 & 1 & m\Z \\ 0 & 0 & 1\end{pmatrix}\,:\,m\in\N_{\geq 1}\right\}
\end{equation}
is a cofinal subset of $\mathfrak{H}$.
\end{proposition}
\begin{proof}
We recall the parametrization of general finite index subgroups of $\HH$, taken from
\cite{Lightwood-Sahin-Ugarcovici:Heisenberg_odometers}*{Proposition 3.1 and
3.2}. The integer Heisenberg group \( \HH \) is isomorphic to
an extension of \( \Z^2 \) by its center 
\( Z(\HH)\simeq\Z \), i.e.~there is a short exact sequence
\[ 
  \begin{tikzcd}[row sep=0]
  1 \arrow[r] & \Z \arrow[r] & \HH \arrow[r,"\varpi"] & \Z^2 \arrow[r] & 1
  \\
  & z \arrow[r, mapsto] & \left(\begin{smallmatrix}
  1 & 0 & z \\ 0 & 1 & 0 \\ 0 & 0 & 1
  \end{smallmatrix}\right)
  & \\
  & & \left(\begin{smallmatrix}
  1 & x & z \\ 0 & 1 & y \\ 0 & 0 & 1
  \end{smallmatrix}\right) \arrow[r, mapsto] & (x,y).
  \end{tikzcd}
\]
Let $G$ be a finite index subgroup of $\HH$. Then $\varpi(G)$ is a finite-index subgroup of $\Z^2$, so $\varpi(G)=A_G\Z^2$ for some
$A_G\in\Mat_2(\Z)$ with $\det(A_G)\neq 0$. Setting $m_G\in \N_{\geq 1}$ to
be the finite index of $G\cap Z(\HH)$ in $Z(\HH)$, and 
\begin{align*}
i_G\colon A_G\Z^2 &\to \Z\\
(x,y)&\mapsto \min\{z\geq 0\,:\,(x,y,z)\in G\},
\end{align*}
we have the following presentation of $G$,
\begin{equation}\label{eq:finite.index.Heis.general}
G=\left\{\begin{pmatrix} 1 & x & i_G(x,y)+m_Gk\\ 0 & 1 & y \\ 0 &
  0 & 1 \end{pmatrix}\,:\,(x,y)\in A_G\Z^2, k\in\Z\right\}.
\end{equation}
Furthermore, the parameters $m_G, i_G, A_G$ satisfy
\begin{equation}\label{eq:normal.Heis.parameters}
i_G(x,y)+i_G(x^\prime,y^\prime)+xy^\prime =i_G(x+x^\prime,y+y^\prime)\;\mod\,m_G,\quad (x,y)\in A_G\Z^2.
\end{equation}

Let $G$ be a finite index subgroup of $\HH$, as described by
\eqref{eq:finite.index.Heis.general}. Since \( \det(A_G)\neq 0 \),
it has an inverse \( A_G^{-1}\in\GL_2(\Q) \).
Set $\ell_G=|\det(A_G)|\in \N_{\geq 1}$. Then $A_G^{-1}(\ell_G,0)\in \Z^2$
and $A_G^{-1}(0,\ell_G)\in \Z^2$. Therefore, $(\ell_G,0), (0,\ell_G)\in
A_G\Z^2$, so \eqref{eq:finite.index.Heis.general} says that
\[
\begin{pmatrix} 
1 & \ell_G & i_G(\ell_G,0) \\ 0 & 1 & 0 \\ 0 & 0 & 1
\end{pmatrix},\;
\begin{pmatrix} 
1 & 0 & i_G(0,\ell_G) \\ 0 & 1 & \ell_G \\ 0 & 0 & 1 
\end{pmatrix}\quad\in G.
\]
Because of \eqref{eq:normal.Heis.parameters}, $i_G(m_G\ell_G,0)=m_G\cdot i_G(\ell_G,0)\mod m_G$ holds. Therefore $i_G(m_G\ell_G,0)$, and similarly $i_G(0,m_G\ell_G)$, is an integer multiple of $m_G$. Then it follows from \eqref{eq:finite.index.Heis.general} that
\[
\begin{pmatrix} 
1 & m_G\ell_G & 0 \\ 0 & 1 & 0 \\ 0 & 0 & 1
\end{pmatrix},\;
\begin{pmatrix} 
1 & 0 & 0 \\ 0 & 1 & m_G\ell_G \\ 0 & 0 & 1
\end{pmatrix}\quad\in G.
\]
The two elements above generate the subgroup $m_G\ell_G\HH\subseteq G$. As $G\in\mathfrak{H}$ was arbitrary, this construction verifies the cofinality of \eqref{eq:Heis_cofinal_subset}.
\end{proof}

Now we describe the K-theory generators of \( \Cstr(m\HH) \).
In order to do this, notice that every finite
index subgroup of the form \( m\HH \) as in
\eqref{eq:Heis_cofinal_subset} is isomorphic to $\HH$, via
\[
  \HH\xrightarrow{\sim}m\HH,\quad (x,y,z)\mapsto (mx,my,m^2z).
\]
This induces an isomorphism \( \Cstr(\HH)\xrightarrow{\sim}\Cstr(m\HH)
\) for each \( m\in\N_{\geq 1} \). 

It is well-known that
$\Cstr(\HH)$ is isomorphic to
a crossed product \Cst-algebra as follows. 
Consider the automorphism $\alpha$
of $\Cstr(\Z\times \Z)\simeq \Cstr(\Z)\otimes\Cstr(\Z)$,
defined by
\[
\alpha(1\otimes w_1)=1\otimes w_1,\qquad \alpha(w_1\otimes 1)= w_1\otimes w_1^*,
\]
where $1$ denotes $1_{\ell^2(\Z)}$ and $w_1=\rho_\Z(1)$. 
We may write $\alpha=\Ad_{v}$, 
where $v$ is the unitary operator on $\ell^2(\Z\times\Z)$ for the
change-of-basis matrix $\left(\begin{smallmatrix} 1 & 0 \\ -1 & 1
\end{smallmatrix}\right)\in\GL(2,\Z)$. 
The reduced \emph{crossed product} $\Cst$-algebra $\Cstr(\Z\times \Z)\rtimes_{\alpha}\Z$ is the $\Cst$-subalgebra of $\Bdd(\ell^2(\Z\times\Z)\otimes\ell^2(\Z))$ generated by
\[
(w_1\otimes 1)\otimes 1,\qquad (1\otimes w_1)\otimes 1,\qquad v\otimes w_1.
\]
There is a canonical inclusion homomorphism
\begin{equation}\label{eq:inclusion.into.crossed.product}
\begin{aligned}
j\colon \Cstr(\Z\times\Z) &\to \Cstr(\Z\times\Z)\rtimes_{\alpha}\Z\\
A & \mapsto A\otimes 1.
\end{aligned}
\end{equation}
There is also the $\Z$-equivariant inclusion homomorphism,
\begin{equation}\label{eq:jprime.M.inclusion}
\begin{aligned}
j^\prime\colon \Cstr(\Z\times \Z)\simeq \Cstr(\Z)\rtimes_{\rm id}\Z &\to \Cstr(\Z\times \Z)\rtimes_{\alpha} \Z\\
w_1\otimes 1 &\mapsto v\otimes w_1,\\
1\otimes w_1 &\mapsto j(1\otimes w_1)=(1\otimes w_1)\otimes 1.
\end{aligned}
\end{equation}

Below, we will describe how to identify $\Cstr(m\HH)$ with
$\Cstr(\Z\times\Z)\rtimes_\alpha\Z$. To keep track of the various copies of
$\Z$, we will write $j_{xz}, j^\prime_{yz}$ instead of $j, j^\prime$.
The identification is
\begin{equation}\label{eq:Heis.crossed.product}
\begin{tikzcd}[cramped, row sep=0]
\Cstr(m\HH)\ar[r,"\sim"] & \Cstr(\HH)\ar[r,"\sim"] &  \Cstr(\Z_x\times\Z_z)\rtimes_{\alpha}\Z_y\\
U_m\defeq \rho_{m\HH}(m,0,0)\ar[r,mapsto] & \rho_\HH(1,0,0)\ar[r,mapsto] & j_{xz}(w_1\otimes 1)=(w_1\otimes 1)\otimes 1\\
V_m\defeq \rho_{m\HH}(0,m,0)\ar[r,mapsto] & \rho_\HH(0,1,0)\ar[r,mapsto] & j^\prime_{yz}(w_1\otimes 1)=v\otimes w_1\\
W_m\defeq \rho_{m\HH}(0,0,m^2) \ar[r,mapsto] & \rho_\HH(0,0,1)\ar[r,mapsto] & j_{xz}(1\otimes w_1)=(1\otimes w_1)\otimes 1=j^\prime_{yz}(1\otimes w_1).
\end{tikzcd}
\end{equation}
Note that \eqref{eq:Heis.crossed.product} indeed an isomorphism of
\Cst-algebras because of the relations
\begin{equation}\label{eq:Heis.group.commutator}
U_mV_mU_m^*V_m^*=W_m,\qquad m\in \N_{\geq 1},
\end{equation}
and
\begin{align*}
&j_{xz}(w_1\otimes 1)\circ j^\prime_{yz}(w_1\otimes 1)\circ j_{xz}(w_1^*\otimes 1) \circ j^\prime_{yz}(w_1^*\otimes 1)\\
&=((w_1\otimes 1)\otimes 1)\circ (v\otimes w_1)\circ  ((w_1^*\otimes 1)\otimes 1)\circ (v^*\otimes w_1^*)\\
&=((w_1\otimes 1)\otimes 1)\circ ((w_1^*\otimes w_1)\otimes 1)\\
&=j_{xz}(1\otimes w_1).
\end{align*}
For later use, we will also need the alternative crossed product structure,
\begin{equation}\label{eq:alternative.crossed.product}
\begin{tikzcd}[cramped]
\Cstr(m\HH)\ar[r,"\sim"] & \Cstr(\HH)\ar[r,"\sim"] & \Cstr(\Z_y\times\Z_z)\rtimes_{\alpha^{-1}}\Z_x,
\end{tikzcd}
\end{equation}
where we note that the \emph{inverse} automorphism $\alpha^{-1}$ appears on
the right side because of \eqref{eq:Heis.group.commutator} with the roles
of $U,V$ swapped. Under \eqref{eq:alternative.crossed.product}, with
$j_{yz}$ denoting the inclusion of $\Cstr(\Z_y\times\Z_z)$ into the
alternative crossed product, we
have another way of writing $V_m$ and $W_m$, namely,
\begin{equation}\label{eq:Vm.in.alternative.crossed.product}
\begin{tikzcd}[row sep=0]
V_m \defeq\rho_{m\HH}(0,m,0) \arrow[r, mapsto] & \rho_{\HH}(0,1,0)
\arrow[r, mapsto] & j_{yz}(w_1\otimes 1), \\
W_m \defeq\rho_{m\HH}(0,0,m) \arrow[r, mapsto] & \rho_{\HH}(0,0,1)
\arrow[r, mapsto] & j_{yz}(1\otimes w_1).
\end{tikzcd}
\end{equation}

The $\K$-theory of $\Cstr(m\HH)\simeq\Cstr(\HH)$ can be computed explicitly
using the Pimsner--Voiculescu sequence for the crossed product
$\Cstr(\Z\times\Z)\rtimes_\alpha\Z$. The result, which can be found in
\cite{Anderson-Paschke:Rotation_algebra}*{Proposition 1.4}, is given below.
\begin{proposition}\label{prop:Heis.standard.K.groups}
The abelian group $\K_0(\Cstr(\Z_x\times \Z_z)\rtimes_{\alpha}\Z_y)$ is freely generated by
\[
[1],\qquad (j_{xz})_*\big([w_1]\otimes [w_1]\big),\qquad (j^\prime_{yz})_*\big([w_1]\otimes [w_1]\big).
\]
The abelian group $\K_1(\Cstr(\Z_x\times \Z_z)\rtimes_{\alpha}\Z_y)$ is freely generated by
\[
[j_{xz}(w_1\otimes 1)],\qquad [j_{yz}^\prime(w_1\otimes 1)],\qquad [\mathcal{X}],
\]
where $\mathcal{X}$ is a unitary in $\Mat_2(\Cstr(\Z_x\times \Z_z)\rtimes_{\alpha}\Z_y)$ constructed in \cite{Anderson-Paschke:Rotation_algebra}*{pp.~8} (denoted $V_a$ there).
\end{proposition}

\begin{remark}
Above, the external product $[w_1]\otimes[w_1]$ is one of two generators of $\K_0(\Cstr(\Z\times\Z))\simeq\K_0(\Cstr(\Z)\otimes \Cstr(\Z))\simeq\Z\oplus\Z$ (the other one being the class of the identity element), see \eqref{eq:d-dim_KO-generators}, \eqref{eq:strong.basis}. In \cite{Anderson-Paschke:Rotation_algebra}, the so-called Bott projection $P_{\rm Bott}\in\Mat_2(\Cont(\T^2))\simeq\Mat_2(\Cstr(\Z\times\Z))$ was used as an explicit representative for a $\K_0(\Cstr(\Z\times\Z))$ class. It may be checked that $[w_1]\otimes [w_1]=[P_{\rm Bott}]-[1]$, but we will not need this.
\end{remark}

For convenience, we will write $X_m\in\Cstr(m\HH)$ for the operator corresponding to \mbox{$\mathcal{X}\in \Cstr(\Z_x\times\Z_z)\rtimes_\alpha\Z_y$} under the isomorphism \eqref{eq:Heis.crossed.product}. It will also be convenient to write
\[
\Cstr(m\HH)\simeq \Cstr(m\Z_x\times m^2\Z_z)\rtimes_\alpha m\Z_y \simeq \Cstr(m\Z_y\times m^2\Z_z)\rtimes_{\alpha^{-1}} m\Z_x,
\]
and still write $j_{xz}, j_{yz}^\prime$ for the respective inclusions into the crossed products.

\begin{corollary}\label{cor:Heis.standard.K.groups}
Let $m\in\N_{\geq 1}$. The abelian group $\K_0(\Cstr(m\HH))$ is freely generated by
\[
[1],\qquad j_{xz}\big([w_{m}]\otimes [w_{m^2}]\big),\qquad j_{yz}\big([w_{m}]\otimes [w_{m^2}]\big),
\]
while the abelian group $\K_1(\Cstr(m\HH))$ is freely generated by
\[
[U_m],\qquad [V_m],\qquad [X_m].
\]
\end{corollary}

\subsection{Symmetry breaking for integer Heisenberg group}
\begin{proposition}\label{prop:Heis.symm.breaking.maps}
Let $\mathcal{H}$ be any
$X$-$\HH$-module as in the \refbasicsetup.
For $m\in\N_{\geq 1}$, the weak subgroup ${\rm wK}_*(\Cst(X,\mathcal{H})^{m\HH})\simeq
{\rm wK}_*(\Cstr(m\HH))$ is free abelian on the generators
listed in
\cref{cor:Heis.standard.K.groups}), with $[X_m]$ excluded. Furthermore, under the
symmetry breaking maps $\phi_{n\HH\backslash m\HH}$, $m\mid n$, we have
\begin{align*}
[1]&\mapsto \left(\frac{n}{m}\right)^4[1],\\
[U_m]&\mapsto \left(\frac{n}{m}\right)^3[U_n],\\
[V_m]&\mapsto \left(\frac{n}{m}\right)^3[V_n],\\
(j_{xz})_*\big([w_{m}]\otimes [w_{m^2}]\big) 
&\mapsto \frac{n}{m}(j_{xz})_*\big([w_{n}]\otimes [w_{n^2}]\big),\\
(j_{yz}^\prime)_*\big([w_{m}]\otimes [w_{m^2}]\big)&\mapsto \frac{n}{m}(j_{yz}^\prime)_*\big([w_{n}]\otimes [w_{n^2}]\big).
\end{align*}
The strong quotient ${\rm sK}_1(\Cst(X,\mathcal{H})^{m\HH})\simeq {\rm sK}_1(\Cstr(m\HH))$ is free abelian on a single generator $q_{m\HH}[X_m]$, where $[X_m]$ is the remaining generator of $\K_*(\Cstr(m\HH))$. Furthermore, the strong symmetry breaking maps
\[
{\rm sK}_1(\iota_{n\HH\backslash m\HH})\colon {\rm sK}_1(\Cst(X,\mathcal{H})^{m\HH})\to {\rm sK}_1(\Cst(X,\mathcal{H})^{n\HH}),\qquad m\mid n,
\]
are isomorphisms $\Z\to\Z$.

\end{proposition}
\begin{proof}
It will suffice to consider $m=1$ and $N=\frac{n}{m}$.
\begin{itemize}
\item $[1]$ case. It is clear that breaking $\HH$ to $N\HH$ maps $[1_{\ell^2(\HH)}]$ to $N^4\cdot [1_{\ell^2(N\HH)}]$.
\item {} $[U_1], [V_1]$ cases.
As usual, we make the Hilbert space identifications
\begin{equation}\label{eq:Heis.Hilbert.decomposition}
\begin{aligned}
\ell^2(\HH)&\simeq \big(\ell^2(\Z_x)\otimes \ell^2(\Z_z)\big)\otimes \ell^2(\Z_y),\\
\ell^2(N\HH) & \simeq \big(\ell^2(N\Z_x)\otimes \ell^2(N^2\Z_z)\big)\otimes \ell^2(N\Z_y),
\end{aligned}
\end{equation}
under which
\begin{equation}\label{eq:tensor.form.U}
\begin{aligned}
U_1\defeq \rho_{\HH}(1,0,0) &\simeq (w_1\otimes 1_{\ell^2(\Z_z)})\otimes 1_{\ell^2(\Z_y)},\\
U_N\defeq \rho_{N\HH}(1,0,0) &\simeq (w_N\otimes 1_{\ell^2(N^2\Z_z)})\otimes 1_{\ell^2(N\Z_y)}.
\end{aligned}
\end{equation}
Consider the symmetry breaking
\[
\HH=\begin{pmatrix}1 & \Z & \Z \\ 0 & 1 & \Z\\ 0 & 0 & 1\end{pmatrix}\prec \begin{pmatrix} 1 & N\Z & \Z \\ 0 & 1 & \Z\\ 0 & 0 & 1 \end{pmatrix}\prec \begin{pmatrix} 1 & N\Z & N^2\Z \\ 0 & 1 & N\Z\\ 0 & 0 & 1 \end{pmatrix}=N\HH.
\]
For the first symmetry breaking step, only the tensor factor $\ell^2(\Z_x)$ is involved, and the 1D symmetry breaking result \eqref{eq:KO.symm.breaking.d1} for $\phi_{N\Z_x\backslash \Z_x}$ gives $[w_1]\mapsto [w_N]$. For the second symmetry breaking step, the $1_{\ell^2(\Z_z)}\otimes1_{\ell^2(\Z_y)}$ tensor factor is mapped to $1_{\ell^2(N^2\Z_z)^{\oplus N^2}}\otimes 1_{\ell^2(N\Z_y)^{\oplus N}}$. In total, 
\[
\phi_{N\HH\backslash \HH}\colon [U_1]\mapsto \big[\big(w_N\otimes 1_{\ell^2(N^2\Z_z)^{\oplus N^2}}\big)\otimes 1_{\ell^2(N\Z_y)^{\oplus N}}\big]=N^3[U_N].
\]

Unlike formula \eqref{eq:tensor.form.U} for $U_1$, under the Hilbert space isomorphism \eqref{eq:Heis.Hilbert.decomposition}, the unitary $V_1$ does not have a simple tensor decomposition formula. Instead, let us consider the map
\begin{equation}\label{eq:Heis.real.auto}
\varsigma_1\colon \HH\to\HH,\qquad (x,y,z)\mapsto (y,x,xy-z),
\end{equation}
which is an automorphism squaring to the identity. Restricting $\varsigma_1$ to $N\HH$, we also have an automorphism
\[
\varsigma_N\colon  N\HH\to N\HH.
\]
We denote the induced $\Cst$-algebra automorphisms of $\Cstr(\HH)$ and $\Cstr(N\HH)$ with the same symbols $\varsigma_1$ and $\varsigma_N$ respectively. We have $\varsigma_1(U_1)=V_1$ and $\varsigma_N(U_N)=V_N$. Moreover, the automorphisms are compatible with symmetry breaking,
\begin{align*}
\phi_{N\HH\backslash \HH}[V_1]&=\phi_{N\HH\backslash \HH}\circ (\varsigma_1)_*[U_1]\\
&=(\varsigma_N)_*\circ\phi_{N\HH\backslash \HH}[U_1]\\
&=(\varsigma_N)_*\big(N^3[U_N]\big)=N^3[V_N].
\end{align*}

\item {} $(j_{xz})_*\big([w_{1}]\otimes [w_{1}]\big)$ and $(j_{yz}^\prime)_*\big([w_{1}]\otimes [w_{1}]\big)$ cases. We break the symmetry in two steps,
\begin{equation}\label{eq:symm.break.steps.for.weak.Bott}
\HH=\begin{pmatrix}1 & \Z & \Z \\ 0 & 1 & \Z\\ 0 & 0 & 1\end{pmatrix}\prec \begin{pmatrix} 1 & \Z & \Z \\ 0 & 1 & N\Z\\ 0 & 0 & 1 \end{pmatrix} \prec \begin{pmatrix} 1 & N\Z & N^2\Z \\ 0 & 1 & N\Z\\ 0 & 0 & 1 \end{pmatrix}=N\HH.
\end{equation}
For the first symmetry breaking step of
\eqref{eq:symm.break.steps.for.weak.Bott}, we pick a section $s\colon N\Z_y\to \Z_y$, thus a section
\[
{\rm id}\times s\colon  \begin{pmatrix} 1 & \Z & \Z \\ 0 & 1 & N\Z\\ 0 & 0 & 1 \end{pmatrix} \to \begin{pmatrix} 1 & \Z & \Z \\ 0 & 1 & \Z\\ 0 & 0 & 1 \end{pmatrix}.
\]
The symmetry breaking map is $\K_*(\Ad_{\eta_{{\rm id}\times s}})$ composed with stabilization. We have
\[
\Ad_{\eta_{{\rm id}\times s}}(A\otimes 1_{\ell^2(\Z)})=A\otimes {\rm Ad}_{\eta_s}(1_{\ell^2(\Z)})=A\otimes 1_{\ell^2(N\Z)}^{\oplus N}.
\]

Applying the last observation to $(j_{xz})_*\big([w_{1}]\otimes [w_{1}]\big)$, we have
\begin{equation}\label{eq:partial.symm.break.Bott.gen}
\begin{aligned}
\K_*(\Ad_{\eta_{{\rm id}\times s}})\big((j_{xz})_*([w_1]\otimes [w_1])\big)
&=\K_*(\Ad_{\eta_{{\rm id}\times s}})\big(([w_1]\otimes [w_1])\otimes [1_{\ell^2(\Z)}]\big)\\
&=([w_1]\otimes [w_1])\otimes [1_{\ell^2(N\Z)}^{\oplus N}]\\
&=N\cdot ([(w_1]\otimes [w_1])\otimes \big[1_{\ell^2(N\Z)}\big].
\end{aligned}
\end{equation}
Next, the second symmetry breaking step of \eqref{eq:symm.break.steps.for.weak.Bott} takes
\[
\K_*\left(\Cstr\begin{pmatrix} 1 & \Z & \Z \\ 0 & 1 & N\Z\\ 0 & 0 & 1 \end{pmatrix}\right)\simeq \K_*(\Cstr(\Z\times \Z)\rtimes N\Z) \to\K_*(\Cstr(N\Z\times N^2\Z)\rtimes_\alpha N\Z),
\]
and it is induced by $\Ad_{t\times {\rm id}}$, where $t\colon N\Z\times N^2\Z\to \Z\times \Z$ is a section. When applied to the class \eqref{eq:partial.symm.break.Bott.gen}, we have
\begin{align*}
\K_*(\Ad_{t\times{\rm id}})\big(N\cdot ([(w_1]\otimes [w_1])\otimes \big[1_{\ell^2(N\Z)}\big]\big)
&=N\cdot\K_*(\Ad_t)([w_1]\otimes [w_1])\otimes \big[1_{\ell^2(N\Z)}\big]\\
&=N\cdot\big([w_N]\otimes [w_{N^2}]\big)\otimes \big[1_{\ell^2(N\Z)}\big]\\
&=N\cdot (j_{xz})_*([w_N]\otimes [w_{N^2}])
\end{align*}
where in the second equality, we used the K\"unneth isomorphism (\cref{lem:real_Kunneth_breaking_formula}) together with the 1D symmetry breaking formula \eqref{eq:KO.symm.breaking.d1}. This gives the scaling-by-$N$ result for $\phi_{N\HH\backslash \HH}$ applied to $(j_{xz})_*([w_1]\otimes [w_1])$.

As for $(j^\prime_{yz})_*([w_1]\otimes [w_1])$, we use the alternative crossed product identification \eqref{eq:alternative.crossed.product}, to convert it into $j_{yz}([w_1]\otimes[w_1])$ (see \eqref{eq:Vm.in.alternative.crossed.product}). Then we similarly deduce that it scales by $N$ under $\phi_{N\HH\backslash \HH}$.
\end{itemize}

The above calculations show that five of the six free generators of
$\K_*(\Cstr(m\HH))$ are weak, so they are annihilated by the map $q_{m\HH}$
to the strong quotient. So the strong quotient
${\rm sK}_*(\Cst(X,\mathcal{H})^{m\HH})\simeq{\rm
sK}_*(\Cstr(m\HH))$ is either trivial or is singly-generated by
$q_{m\HH}[X_m]$, where $[X_m]$ is the last remaining generator of
$\K_*(\Cstr(m\HH))$. 

Without loss of generality, we may assume that $\mathcal{H}$ is ample. By \cref{thm:existence_strong_phase}, the comparison map is surjective from ${\rm sK}_1(\Cst(X,\mathcal{H})^{m\HH})$ to $\Z$. It follows that ${\rm sK}_1(\Cst(X,\mathcal{H})^{m\HH})$ must be isomorphic to $\Z$. This also shows ${\rm wK}_*(\Cstr(\Z^d,\sigma_\Theta))$ has no further generators besides the five mentioned above.

As the above conclusion holds for every $m$, consistency under
symmetry breaking (see \eqref{eq:consistent.symm.break.strong}) implies
that the strong symmetry breaking maps ${\rm sK}_1(\iota_{n\HH\backslash
m\HH})$ must be isomorphisms $\Z\xrightarrow{\sim}\Z$.
\end{proof}

\begin{theorem}\label{thm:Heis.symm.breaking}
Let $\HH$ be the integer Heisenberg group, and $\mathfrak{H}$ be its symmetry breaking poset. For any $X$-$\HH$-module $\mathcal{H}$ in the \refbasicsetup, we have
\[
\K_i(\Cst(X,\mathcal{H})^{\mathfrak{H}})\simeq \colim_{m\in
\N}\K_i(\Cstr(m\HH))\simeq 
\begin{cases}
    \Q^3, & i=0,\\
\Q^2\oplus \Z, & i=1.
\end{cases}
\]
\end{theorem}
\begin{proof}
The first isomorphism is due to \cref{thm:main_computation_theorem}, \cref{lem:cofinal_isomorphism} and \cref{prop:Heis.cofinal}. In \cref{prop:Heis.symm.breaking.maps}, the weak symmetry breaking maps $\phi_{m\HH\backslash n\HH}, n\mid n$, were found to diagonal matrices
\[
\mathrm{diag}\big((n/m)^4, (n/m)^3, (n/m)^3, (n/m), (n/m)\big)
\]
with respect to suitable free abelian bases for the weak subgroups. So the weak colimit is isomorphic to $\Q^3$ in degree $1$ and $\Q^2$ in degree $0$. As for the extra $\Z$-summand in degree $1$, this comes from the last part of \cref{prop:Heis.symm.breaking.maps} which implies that the strong colimit is isomorphic to $\Z$.
\end{proof}

\section{Symmetry breaking in Euclidean space: complex twisted case}
\label{sec:Euclidean.twisted}
In this section, we study symmetry breaking for $\Gamma\simeq\Z^d$ in the
twisted case. The tensor product factorization and K\"unneth theorem method
no longer works. Instead, we will use the Pimsner--Voiculescu exact
sequence \cite{Pimnser-Voiculescu:Exact_sequence}, combined with
\cref{thm:existence_strong_phase}, to compute the weak subgroups and
strong quotients, as well as their colimits.

\subsection{Symmetry breaking computation}\label{sec:Euclidean.twisted.computation}
Fix a dimension $d\geq 1$, and a standardized 2-cocycle $\sigma_\Theta$ on $\Z^d$ (Eq.~\eqref{eq:standard_2_cocycle}), specified by a real antisymmetric matrix $(\Theta_{jk})_{j,k=1,\ldots,d}$. We continue to write $\sigma_\Theta$ for the restricted 2-cocycle on subgroups of $\Z^d$.
Let $\{e_1,\ldots,e_d\}$ denote the standard basis elements of $\Z^d$. For a multi-index
\[
I=\{i_1,\ldots,i_{|I|}\}\subseteq\{1,\ldots,d\},
\]
we write $\Z^I\subseteq\Z^d$ for the rank-$|I|$ subgroup with generators $\{e_i\}_{i\in I}$, and $I^\mathrm{c}$ for the complementary multi-index.

The twisted group $\Cst$-algebra $\Cstr(\Z^I,\sigma_\Theta)$ is generated by the unitary operators
\[
U^I_j=\rho_{\Z^I}^{\sigma_\Theta}(e_j),\qquad j\in I,
\]
which, by formula \eqref{eq:standard_2_cocycle}, satisfy
\begin{equation}\label{eq:higher_NC_torus_conjutation_rule}
U^I_k U^I_j (U^I_k)^*=\exp(2\pi i\Theta_{kj}) U^I_j,\qquad j,k\in I.
\end{equation} 
So $\Cstr(\Z^I,\sigma_\Theta)$ is an \emph{$|I|$-dimensional noncommutative
torus} in the sense of
\cites{Rieffel:Higher-dimensional_NC_tori,Elliott:Projective_representation_torsion-free}.

The unitaries $U^I_j, j\in I\setminus\{j_{|I|}\}$ generate the
$\Cst$-subalgbera
\mbox{$\Cstr(\Z^{I\setminus\{j_{|I|}\}},\sigma_\Theta)\otimes
1_{\ell^2(\Z)}$.} Eq.~\eqref{eq:higher_NC_torus_conjutation_rule} specifies
how conjugation by $U_{j_{|I|}}$ implements an automorphism of this
subalgebra. Thus $\Cstr(\Z^{I},\sigma_\Theta)$ is a (reduced) crossed
product of $\Cstr(\Z^{I\setminus\{j_{|I|}\}},\sigma_\Theta)$ by a group
$\Z$ of automorphisms.
Iterating, $\Cstr(\Z^I,\sigma_\Theta)$ is an $|I|$-fold \emph{crossed
product} of the algebra $\C$ by $\Z$.

Let us consider
\[
  \Cstr(\Z^d,\sigma_\Theta)\defeq \Cstr(\Z^{\{1,\ldots,d\}},\sigma_\Theta)=\Cstr(\Z^{\{1,\ldots,d-1\}},\sigma_\Theta)\rtimes\Z^{\{d\}}.
\]
Above, the automorphism is labelled by the real parameters $\Theta_{dj}, j=1,\ldots,d-1$, as in \eqref{eq:higher_NC_torus_conjutation_rule}, thus it is homotopic to the identity automorphism, and induces the identity map on $\K_*(\Cstr(\Z^{\{1,\ldots,d-1\}},\sigma_\Theta))$. Consequently, the associated PV exact sequence simplifies to the short exact sequences
\begin{equation}\label{eq:higher_NC_PV_SES}
\begin{tikzcd}
0\ar[r] & \K_i(\Cstr(\Z^{\{1,\ldots,d-1\}},\sigma_\Theta)) \ar[r,"(\blank)\otimes 1"] & \K_i\big(\Cstr(\Z^{\{1,\ldots,d\}},\sigma_\Theta)\rtimes\Z^{\{d\}}\big)\ar[d,"\partial_{{\rm PV}}"] \\
{} & 0 & \K_{i-1}(\Cstr(\Z^{\{1,\ldots,d-1\}},\sigma_\Theta))\ar[l]
\end{tikzcd}
\end{equation}
for $i=0,1$.
We know that $\K_0(\C)\simeq\Z$, generated by $[1]$. Then, by an inductive argument, \begin{equation}\label{eq:complex_K_groups_torus}
\K_*(\Cstr(\Z^d,\sigma_\Theta))\simeq \Z^{2^d}.
\end{equation}

Below, we shall describe a basis $\{w_I\}_{I\subseteq\{1,\ldots,d\}}$ for $\K_*(\Cstr(\Z^d,\sigma_\Theta))$. For each multi-index $I\subseteq\{1,\ldots,d\}$, consider $\Cstr(\Z^{I^\mathrm{c}},\sigma_\Theta)$ as an $|I^\mathrm{c}|$-fold crossed product, and set
\begin{equation}\label{eq:PV.lift.iterated}
w^{I^\mathrm{c}}_\emptyset\in \K_{|I^\mathrm{c}|}(\Cstr(\Z^{I^\mathrm{c}},\sigma_\Theta))
\end{equation}
to be an $|I^\mathrm{c}|$-fold iterated PV lift of $[1]\in\K_0(\C)$. Then define
\begin{equation}\label{eq:higher_NC_basis_choice}
w_I\defeq w^{I^\mathrm{c}}_\emptyset\otimes 1_{\ell^2(\Z^I)}\in\K_{|I^\mathrm{c}|}(\Cstr(\Z^d,\sigma_\Theta)).
\end{equation}
Note that:
\begin{itemize}
\item If $d\in I$, then $w_I$ factorizes as
\begin{equation*}
w_I=\big(w^{I^\mathrm{c}}_\emptyset\otimes 1_{\ell^2(\Z^{I\setminus\{d\}})}\big)\otimes 1_{\ell^2(\Z^{\{d\}})}.
\end{equation*}
So $w_I$ is obtained from a $\K_*(\Cstr(\Z^{\{1,\ldots,d-1\}},\sigma_\Theta))$ class by applying $(\blank)\otimes 1$ in \eqref{eq:higher_NC_PV_SES}.
\item Otherwise, if $d\in I^\mathrm{c}$, we have the crossed product
\[
\Cstr(\Z^{I^\mathrm{c}},\sigma_\Theta)=\Cstr(\Z^{I^\mathrm{c}\setminus\{d\}},\sigma_\Theta)\rtimes\Z^{\{d\}}.
\]
The homomorphism
\[
(\blank)\otimes 1_{\ell^2(\Z^I)}\colon \Cstr(\Z^{I^\mathrm{c}\setminus\{d\}},\sigma_\Theta)\to \Cstr(\Z^{\{1,\ldots,d-1\}},\sigma_\Theta)
\]
is $\Z^{\{d\}}$-equivariant. So naturality of PV sequences gives
\[
\partial_{{\rm PV}} w_I=\partial_{{\rm PV}} w_\emptyset^{I^\mathrm{c}}\otimes \big[1_{\ell^2(\Z^I)}\big]\in \K_*(\Cstr(\Z^{\{1,\ldots,d-1\}},\sigma_\Theta)).
\]
\end{itemize}
Inductively, and from our definition of $w_\emptyset^{I^\mathrm{c}}$, we see that $w_I$ is obtained from $[1]\in\K_0(\C)$ by taking PV lifts $|I^\mathrm{c}|$ times and $|I|$-fold applications of $(\cdot)\otimes 1_{\ell^2(\Z)}$. Recalling the short exact sequences \eqref{eq:higher_NC_PV_SES}, iterated $d$ times, we conclude that these $w_I, I\subseteq\{1,\ldots,d\}$ provide a free abelian basis for $\K_*(\Cstr(\Z^d,\sigma_\Theta))$.

\begin{remark}
A similar iterated-PV construction of a basis for
$\K_*(\Cstr(\Z^d,\sigma))$ was given in
\cite{Prodan-SBaldes:Complex_topological_insulators}*{Section 4.2.3}. We stress that many choices are involved in constructing the basis elements $w_I$. This is because the splitting of the PV short exact sequences are not canonical, so the $w_\emptyset^{I^\mathrm{c}}$ in \eqref{eq:PV.lift.iterated} are not canonical. In fact, it is not easy to specify any $i=0$ PV lift's representative explicitly.

For our basis \eqref{eq:higher_NC_basis_choice}, the $(\blank)\otimes
1_{\ell^2(\Z^I)}$ tensor factor indicates that the basis elements $w_I$
with $I\neq \emptyset$ are ``stacked topological phases'' with ``stacking
directions'' $I$. This terminology is borrowed from physics, where ``weak''
topological phases are often described as being connected to ``stacked''
phases.
\end{remark}

For $m\in\N_{\geq 1}$, there is likewise a basis $\{w_{m,I}\}_{I\subseteq\{1,\ldots,d\}}$ for $\K_*(\Cstr(\Z^d,\sigma_\Theta))\simeq \Z^{2^d}$.

\begin{proposition}\label{prop:strong.basis}
The weak subgroup ${\rm wK}_*(\Cstr(m\Z^d,\sigma_\Theta))$ is free abelian on the generators $w_{m,I}, I\neq \emptyset$, while the strong quotient ${\rm sK}_*(\Cstr(m\Z^d,\sigma_\Theta))$ is freely generated by the image of $w_{m,\emptyset}$. Furthermore, for $m\mid n$, the strong symmetry breaking maps ${\rm sK}_d(\Cstr(m\Z,\sigma_\Theta))\to {\rm sK}_d(\Cstr(n\Z,\sigma_\Theta))$ are isomorphisms $\Z\to\Z$.
\end{proposition}
\begin{proof}
As usual, we may replace $(m,n)$ by $(1,N)$ where $N=\frac{n}{m}$.

First, suppose $I\neq \emptyset$. For the subgroup $\Z^{I^\mathrm{c}}\times N\Z^I$, we can construct a basis for $\K_*(\Cstr(\Z^{I^\mathrm{c}}\times N\Z^I,\sigma_\Theta))$ by the same kind of formula as \eqref{eq:higher_NC_basis_choice},
\[
w_{N,J}\defeq w^{J^\mathrm{c}}_\emptyset\otimes 1_{\ell^2(\Z^{J\cap I^\mathrm{c}})}\otimes 1_{\ell^2(N\Z^{J\cap I})}\in\K_*(\Cstr(\Z^d,\sigma_\Theta)),\qquad J\subseteq\{1,\ldots,d\}.
\]
In particular, for $J=I$, we have
\[
w_{N;I}=w^{I^\mathrm{c}}_\emptyset\otimes 1_{\ell^2(N\Z^{I})}
\]
It is clear that symmetry breaking from $\Z^d=\Z^{I^\mathrm{c}}\times \Z^I$ to $\Z^{I^\mathrm{c}}\times N\Z^I$ maps
\begin{equation}\label{eq:partial.weak.symm.breaking}
w_I\mapsto N^{|I|}\cdot w_{N,I}.
\end{equation}
The right side is $N$-divisible (since $|I|>0$). As $N$ was arbitrary, this implies that $w_I$ maps to a divisible element in the symmetry breaking colimit, i.e., $w_I$ is weak. So all ``stacked'' basis elements $w_I, I\neq\emptyset$ map to zero in the strong quotient ${\rm sK}_*(\Cstr(\Z^d,\sigma_\Theta))$.

Thus ${\rm sK}_*(\Cstr(\Z^d,\sigma_\Theta))$ is generated by the image
$q_{\Z^d}(w_\emptyset)$ of the ``unstacked'' degree-$d$ basis element
$w_\emptyset$. Apply \cref{thm:existence_strong_phase} to deduce
that the comparison map 
\[
{\rm sK}_d(\iota_{1\backslash \Z^d})\colon {\rm sK}_d(\Cstr(\Z^d,\sigma_\Theta))\simeq \K_d(\Cst(\R^d)^{\Z^d})\to \K_d(\Cst(\R^d))\simeq \Z
\]
is a surjection (using an ample $\sigma_\Theta$-twisted
$\R^d$-$\Z^d$-module for the Roe algebras). It follows that $q_{\Z^d}(w_\emptyset)$ is a free generator for ${\rm sK}_d(\Cstr(\Z^d,\sigma_\Theta))\simeq\Z$. This also shows that no multiple of $w_\emptyset$ is weak, thus ${\rm wK}_*(\Cstr(\Z^d,\sigma_\Theta))$ has no further generator besides the ``stacked'' ones mentioned earlier.

The same conclusion holds for every $m\in\N_{\geq 1}$. By consistency,
Eq.~\eqref{eq:consistent.symm.break.strong}, the symmetry breaking maps
must be isomorphisms $\Z\to\Z$.
\end{proof}

\begin{proposition}\label{prop:higher_NC_torus_symm_break_matrix}
For $N\in\N_{\geq 1}$, let $\{w_I\}_{\emptyset\neq I\subseteq\{1,\ldots,d\}}$ and
$\{w_{N,I}\}_{\emptyset\neq I\subseteq\{1,\ldots,d\}}$ be bases for
${\rm wK}_*(\Cstr(\Z^d,\sigma_\Theta))$ and ${\rm wK}_*(\Cstr(N\Z^d,\sigma_\Theta))$ respectively,
as in \cref{prop:strong.basis}. Order these bases in decreasing order of
$|I|$. Then the symmetry breaking homomorphism $\phi_{N\Z^d\setminus \Z^d}$ on the weak subgroups is represented as an upper-triangular $(2^d-1)\times (2^d-1)$ integer matrix
$\Phi_N$ with non-zero diagonal entries.
\end{proposition}
\begin{proof}
We can carry out the symmetry breaking in two steps,
\[
\Z^d\to \Z^{I^\mathrm{c}}\times N\Z^I,\qquad\mathrm{then}\qquad \Z^{I^\mathrm{c}}\times N\Z^I\to N\Z^{I^\mathrm{c}}\times N\Z^I=N\Z^d.
\]
For the first symmetry breaking step, \eqref{eq:partial.weak.symm.breaking} says that
\[
w_I\mapsto N^{|I|}\cdot w_{N,I}=N^{|I|}\cdot w^{I^\mathrm{c}}_\emptyset\otimes 1_{\ell^2(N\Z^I)}\in \K_*\big(\Cstr(\Z^{I^\mathrm{c}}\times N\Z^I,\sigma_\Theta)\big).
\]
The second symmetry breaking step only involves the $I^\mathrm{c}$ indices. By replacing $\Z^d$ with $\Z^{I^\mathrm{c}}$, we know from the proof of \cref{prop:strong.basis} that
\[
\phi_{N\Z^{I^\mathrm{c}}\setminus \Z^{I^\mathrm{c}}}\colon w^{I^\mathrm{c}}_\emptyset\mapsto A\pm w^{I^\mathrm{c}}_{N,\emptyset},
\]
with $A$ being some weak element of $\K_*(\Cstr(N\Z^{I^\mathrm{c}},\sigma_\Theta))$, necessarily of the form
\[
A=\sum_{\emptyset\neq J\subseteq I^\mathrm{c}} n_J\cdot w^{I^\mathrm{c}}_{N,J}.
\]
Overall, the total symmetry breaking map $\phi_{N\Z^d\setminus \Z^d}$ takes
\[
w_I \mapsto N^{|I|}\cdot\Big(\big(\pm w^{I^\mathrm{c}}_{N,\emptyset}+A\big)\otimes 1_{\ell^2(N\Z^I)}\Big)\\
=\pm N^{|I|}\cdot w_{N,I}\; + N^{|I|}\cdot\sum_{\emptyset\neq J\subseteq I^\mathrm{c}}n_J\cdot w_{N,J\cup I}.
\]
On the right side, the extra weak terms have multi-indices $J\cup I$ of cardinality strictly greater than $|I|$. Therefore, the $I$-th column of the matrix $\Phi_N$ has $I$-th entry being $\pm N^{|I|}$ and zero entries afterwards.
\end{proof}

\begin{theorem}\label{thm:complex_twisted_colimit}
Let $\Gamma\simeq\Z^d$ be a free abelian group on $d$ generators, and $\mathfrak{Z}^d$
be its symmetry breaking poset. Let $\sigma$ be a 2-cocycle on $\Gamma$.
For any $\sigma$-twisted $X$-$\Gamma$-module $\mathcal{H}$ in
the \refbasicsetup, we have
\begin{equation}\label{eq:complex_twisted_colimit}
\K_i(\Cst(X,\mathcal{H})^{\mathfrak{Z}^d})\simeq\colim_{m\in \N}\K_i(\Cstr(m\Z^d,\sigma))\simeq\begin{cases}
\Q^{2^{d-1}-1}\oplus\Z, & \text{if \( i-d \) is even;} \\
\Q^{2^{d-1}}, & \text{if \( i-d \) is odd.}
\end{cases}
\end{equation}
\end{theorem}

\begin{proof}
By \cref{thm:main_computation_theorem} and \cref{lem:cofinal_subseteq_R}, we have
\begin{equation*}
\K_i(\Cst(X,\mathcal{H}))^{\mathfrak{Z}^d}\simeq
\colim_{G\in\mathfrak{Z}^d}\K_i(\Cstr(G,\sigma))\simeq
\colim_{m\in\N}\K_i(\Cstr(m\Z^d,\sigma)).
\end{equation*}
By \cref{prop:cohomologous_symm_breaking} and \cref{lem:Zd_standard_cocycle}, we may replace $\sigma$ with a cohomologous $\sigma_\Theta$ of the standardized form \eqref{eq:standard_2_cocycle}.

By \cref{prop:higher_NC_torus_symm_break_matrix}, the symmetry breaking maps $\phi_{n\Z^d\setminus m\Z^d}$ on the weak subgroups are represented by $(2^d-1)\times (2^d-1)$ upper-triangular integer matries $\Phi_{n,m}$ with non-zero diagonals. Note that $\Phi_{n,m}$ may be regarded as elements of \mbox{${\rm GL}(2^d-1,\Q)$}, so they are injective. Hence, the weak colimit \( W \),
\[
W\defeq \colim_{m\in\N}{\rm wK}_*(\Cstr(\Z^d,\sigma_\Theta)),
\]
is torsion-free. So \( W \) is a torsion-free divisible abelian group, whence a \( \Q
\)-vector space.
We claim that \( W \) has $\Q$-dimension \( 2^{d}-1 \). To this end, note that
\[
{\rm wK}_*(\Cstr(m\Z^d,\sigma_\Theta))\otimes_\Z\Q\simeq\Q^{2^d-1},\qquad m\in\N,
\]
and that
\[ 
  \Phi_{n,m}\otimes_\Z\Q\colon
  {\rm wK}_*(\Cstr(m\Z^d,\sigma_\Theta))\otimes_\Z\Q\to {\rm
  wK}_*(\Cstr(n\Z^d,\sigma_\Theta))\otimes_\Z\Q,\qquad \text{for \( m\mid n
  \),} 
\]
are isomorphisms of \( \Q \)-vector spaces, \( \Q^{2^d-1}\to\Q^{2^d-1} \). So the colimit of these rationalized maps is isomorphic to $\Q^{2^d-1}$.
In general, taking the colimit of a diagram of abelian groups commutes with
tensor products (see \cref{lem:tensor_product_colimit_Ab}). 
Thus we have
\begin{align*} 
\Q^{2^d-1}&\simeq\colim({\rm wK}_*\big(\Cstr(m\Z^d,\sigma)\otimes_\Z\Q\big)\\
  &\simeq\big(\colim {\rm wK}_*(\Cstr(m\Z^d,\sigma))\big)\otimes_\Z\Q\\
  &\simeq W\otimes_\Z\Q\simeq W.
\end{align*}
Note that this calculation can be split into the even and odd $\K$-theory parts, because the symmetry breaking maps respect the $\K$-theory degree.

Finally, \cref{prop:strong.basis} implies that the
strong quotient is isomorphic to $\Z$. This accounts for the extra $\Z$
summand in 
\(\colim_{m\in\N}\K_i(\Cstr(m\Z^d,\sigma))\).
\end{proof}

\subsection{Strong/coarse topological invariant of Landau
levels}\label{sec:Landau_level}
Consider the magnetic Laplacian on the Euclidean plane $X\simeq \R^2$, for the \emph{uniform} magnetic 2-form $\mathscr{F}=b\cdot\mathrm{vol}_X, b>0$. In the \emph{symmetric gauge}, the connection 1-form is
\[
\mathscr{A}=\frac{b}{4i}(\bar{z}dz-z d\bar{z}),
\]
in terms of the complex coordinate $z$. The \emph{Landau operator} is the magnetic Laplacian,
\[
H_b=(d-i\mathscr{A})^*(d-i\mathscr{A}).
\]
A brief computation shows that the \emph{twisted Dirac operator} on $X$,  
\begin{equation}\label{eq:twisted_Dirac_operator}
\slashed{D}_b=\begin{pmatrix}0 & (\slashed{D}_b)_-\\ (\slashed{D}_b)_+ & 0 \end{pmatrix}=-2i\begin{pmatrix} 0 & \partial -\frac{b}{4}\bar{z} \\ \bar{\partial}+\frac{b}{4}z & 0\end{pmatrix},
\end{equation}
where $\partial=\frac{1}{2}(\partial_x-i\partial_y),
\bar{\partial}=\frac{1}{2}(\partial_x+i\partial_y)$, squares to
\[
\slashed{D}_b^2=\begin{pmatrix} H_b-b & 0 \\ 0 & H_b+b\end{pmatrix}.
\]
(This is an example of the Schr\"{o}dinger--Lichnerowicz identity.)

As $H_b+b$ is strictly positive, $\ker(\slashed{D}_b)_-=0$, while
\[
\ker(H_b-b)=\ker \slashed{D}_b^2=\ker \slashed{D}_b=\ker (\slashed{D}_b)_+=\overline{\mathrm{span}}\left\{\binom{z\mapsto z^me^{-b|z|^2/4}}{0}\right\}.
\]
This is the eigenspace for the lowest eigenvalue of $H_b$, and is called the \emph{Lowest Landau Level} (LLL) eigenspace. We have just identified it with a twisted Dirac kernel. The interval $(b,3b)$ is a spectral gap of $H_b$. The eigenprojection
\[
P_{{\rm LLL},b}\defeq P_{\ker (H_b-b)}
\]
is obtained by continuous functional calculus of $\slashed{D}_b$. So $P_{{\rm LLL},b}$ lies in the Roe algebra $\Cst(X)$, and its $\K_0$-class $[P_{{\rm LLL},b}]$ is the \emph{coarse Dirac index} \cite{Roe:Index_theory_coarse_geometry} generating $\K_0(\Cst(X))\simeq\Z$. So the $P_{{\rm LLL},b}$ has a non-trivial coarse topological invariant. 

Let $\Gamma$ be a lattice in the translation group. The uniform $\mathscr{F}$ is $\Gamma$-invariant, and as explained in \cref{sec:magnetic_translations}, there is a projective unitary representation of $\Gamma$ on $L^2(X)$ making it a $\sigma$-twisted $X$-$\Gamma$-module, and which commutes with $(d-i\mathscr{A})$, thus with $H_b$. Thus we may consider $[P_{\mathrm{LLL},b}]$ as a class $[P_{\mathrm{LLL},b}]_\Gamma\in \K_0(\Cst(X)^\Gamma)\simeq\Z\oplus \Z$. This also holds for any finite index subgroup $G\subseteq\Gamma$, and we note that $G\simeq\Z^2$.

In \cref{sec:Euclidean.twisted.computation}, we described a basis $\{w_{\{1,2\}}, w_\emptyset\}$ for $\K_0(\Cst(X)^G)\simeq \K_0(\Cstr(\Z^2,\sigma))\simeq\Z\oplus\Z$, and explained that $w_{\{1,2\}}$ generates the weak subgroup while $q_G(w_\emptyset)$ generates the strong quotient. In the proof of \cref{prop:strong.basis}, we saw that the strong-to-coarse comparison map is an isomorphism $\Z\to \Z$. Thus, in the diagram
\begin{equation}\label{eq:K_computation_Landau}
\begin{tikzcd}[column sep=large]
\underbrace{\K_0(\Cst(X)^{G})}_{\mathbb{Z}\oplus \mathbb{Z}}
\ar[r,"q_G"']\ar[rr,"\K_0(\iota_{1\backslash G})",bend left=20] & \underbrace{{\rm sK}_0(\Cst(X)^G)}_{\mathbb{Z}} \ar[r,"\sim","{\rm sK}_0(\iota_{1\backslash G})"'] & \underbrace{\K_0(\Cst(X))}_{\mathbb{Z}}\\
{[P_{\mathrm{LLL},b}]_G}=(r,c)
\ar[r, mapsto] & c
\ar[r, mapsto] &{[P_{\mathrm{LLL},b}]}=\pm 1
\end{tikzcd},
\end{equation}
we must have $c=\pm 1$ for the LLL projection. We stress that this argument
holds independently of the choice of lattice $\Gamma$ and $G\subseteq
\Gamma$.

\begin{remark}\label{rem:physics.Chern.Hall}
If $b\in 2\pi \Z$ (``integral flux condition''), then
$\sigma\colon \Gamma\times\Gamma\to{\rm U}(1)$ will be cohomologically trivial.
The same triviality holds for the restriction \( \sigma|_G\colon G\times
G\to\mathrm{U}(1) \) to any finite-index subgroup \( G\subseteq
\Gamma \).
If $b\in 2\pi \Q$ (``rational flux quantisation''),
then there exists a finite-index subgroup \( G\subseteq \Gamma \) such
that \( \sigma|_{G\times G} \) is cohomologically trivial.
Then we have
\[
\K_0(\Cst(X)^G)\xrightarrow{\sim} \K_0(\Cstr(G))
\underset{\sim}{\xrightarrow{\mathrm{Fourier}}}
\K^0(\widehat{G})\simeq\Z\oplus \Z.
\]
On the right side, the reduced part
$\widetilde{\K}^0(\widehat{G})\simeq\Z$ of the topological $\K$-group
corresponds to $c$ in \eqref{eq:K_computation_Landau}. 

In physics, this $c$
is called the ``Chern number'', introduced in
\cite{TKNN:Quantized_Hall_conductance}. While it is often stated that $c=\pm 1$
for a Landau level projection, explicit calculations that verify this do
not seem to be readily available. 
That $c\neq 0$ was  mentioned in 
\cite{Dubrovin-Novikov:Ground_states}.
More recently, a computation that $c=\pm 1$, in
Landau gauge (i.e.~standard form \eqref{eq:standard_2_cocycle} for $\sigma$) for a unit flux square lattice, was carried out in
\cite{dNittis-Gomi-Moscolari:Geometry_non-abelian_Landau_levels}*{Corollary
3.3}. 
Another computation using the Atiyah--Singer family index
theorem can be found in
\cite{Bunke-Ludewig:Breaking_symmetries}*{Section 7, Proposition 7.1(3)}.
We have provided a different and more general argument,
\eqref{eq:K_computation_Landau}, using the coarse geometric perspective
together with the strong-to-coarse comparison map.
\end{remark}

\section{Conclusion and outlook}
We have seen that whenever a $\Gamma$-equivariant index makes sense in $\K_*(\Cst(X,\mathcal{H})^\Gamma)$, we automatically have an entire diagram of symmetry breaking maps
\[
\K_*(\iota_{H\backslash G})\colon \K_*(\Cst(X,\mathcal{H})^G)\to
\K_*(\Cst(X,\mathcal{H})^H),
\quad G,H\in\mathfrak{S},\ G\prec H.
\]
By considering the whole symmetry breaking diagram instead of focusing only on a single $G$, we discover the canonical weak/strong dichotomy for equivariant indices at each $G\in\mathfrak{S}$. We also learn that in the typical situation where $\K_*(\Cst(X,\mathcal{H}))$ is a reduced group, only strong equivariant indices can survive when regarded as non-equivariant coarse indices.

For abelian $\Gamma\simeq\Z^d$, we explained in
\cref{thm:existence_strong_phase} and
\cref{rmk:untwisted_real_abelian_strong_to_coarse}
that the strong-to-coarse comparison maps
\eqref{eq:strong.to.coarse.comparison.maps} are isomorphisms in both the
real and complex twisted cases. In general, it may not be an isomorphism,
or even a surjection (see
\cref{sec:crystallographic.generalization}). 
Then we must pass to the strong symmetry breaking colimit.
In \cref{thm:existence_strong_phase} and \cref{cor:existence.strong.colimit}, we showed that for
finitely-generated, (virtually) torsion-free and nilpotent groups, the
strong(-\emph{colimit})-to-coarse comparison maps in complex $\K$-theory are
surjective. It would be useful to extend the class of groups for which
this surjectivity conclusion holds, and also to understand when it holds in the real $\K$-theory setting (besides the abelian case explained in \cref{rmk:untwisted_real_abelian_strong_to_coarse}).

Finally, let us mention that we have been using the directed poset
$\mathfrak{S}$ of \emph{all} finite index subgroups of $\Gamma$. In
principle, we could also consider some smaller directed poset
$\mathfrak{S}^\prime$, and define the weak/strong dichotomy with respect to
the $\mathfrak{S}^\prime$-shaped symmetry breaking diagram. For example,
this could be relevant when considering symmetry breaking for a
crystallographic group, in a manner which maintains the presence of the
point group $F$.

\appendix
\section{Magnetic translations}
\subsection{Magnetic translation operators}\label{sec:magnetic_translations}
Let $X$ be a manifold with an action of a group $\Gamma$ by diffeomorphisms. 
Assume that
\begin{enumerate}
\item $X$ is connected and simply-connected;
\item $\mathcal{L}\to X$ is a trivializable Hermitian line bundle;
\item There is a connection $\nabla$ on $\mathcal{L}$ whose curvature is $\Gamma$-invariant.
\end{enumerate}

Picking a trivialization $\mathcal{L}\simeq X\times\C$, we identify 
$\nabla$ with
$d-i\mathscr{A}$ for some globally-defined connection 1-form
$\mathscr{A}\in\Omega^1(X)$. (We follow the physics convention and drop the
$i$ factor in $\mathscr{A}$.) Write
$\mathscr{F}=d\mathscr{A}\in\Omega^2(X)$ for the curvature 2-form. Its
$\Gamma$-invariance implies that
\[
d(\mathscr{A}-g^*\mathscr{A})=d\mathscr{A}-g^*d\mathscr{A}=\mathscr{F}-g^*\mathscr{F}=0,\qquad g\in \Gamma.
\]
As $X$ is simply-connected, the closed 1-form $\mathscr{A}-g^*\mathscr{A}$ admits an antiderivative $I_g\in\Omega^0(X)$. For example, we can take
\begin{equation}\label{eq:magnetic_antiderivative}
I_g(x)\defeq I_g^{(\mathscr{A})}(x)=\int_\mathcal{O}^x(\mathscr{A}-g^*\mathscr{A}),\qquad g\in \Gamma,\;x\in X,
\end{equation}
where $\mathcal{O}\in X$ is some choice of origin, and the integral can be taken over any smooth path joining $\mathcal{O}$ to $x$. 

The \emph{magnetic translation operator by $g\in \Gamma$} is defined to be
\begin{equation}\label{eq:magnetic.trans}
U_{g}\defeq U_{g}^{(\mathscr{A})}=e^{iI_{g^{-1}}}\cdot (g^{-1})^*,\qquad g\in \Gamma,
\end{equation}
acting on smooth functions $X\to\C$. The operators $U_g$ are readily
checked to commute with $d-i\mathscr{A}$. (The appearance of $g^{-1}$ in
formula \eqref{eq:magnetic.trans} makes it a \emph{left} action.) The
magnetic translations \eqref{eq:magnetic.trans} satisfy
\begin{equation}\label{eq:magnetic.translation.projective}
U^{(\mathscr{A})}_{g}U^{(\mathscr{A})}_{h}=\sigma(g,h)U^{(\mathscr{A})}_{gh},\qquad g, h\in \Gamma,
\end{equation}
where
\begin{equation}\label{eq:magnetic_cocycle}
\begin{aligned}
\sigma(g^{-1},h^{-1})&=\exp\left(i(I_{g}+g^*I_{h}-I_{hg})\right)\\
&=\exp\left(\int_{\mathcal{O}}^x(\mathscr{A}-g^*\mathscr{A})+\int_{\mathcal{O}}^{g\cdot x}(\mathscr{A}-h^*\mathscr{A})-\int_{\mathcal{O}}^x(\mathscr{A}-(hg)^*\mathscr{A})\right)\\
&=\exp\left(i\int_{\mathcal{O}}^{g\cdot\mathcal{O}} (\mathscr{A}-h^*\mathscr{A})\right)
\end{aligned}
\end{equation}
are \emph{$x$-independent} phases. Therefore, $g\mapsto
U_g^{(\mathscr{A})}$ defines a \emph{projective} unitary representation of
$\Gamma$, with $2$-cocycle $\sigma$ given by
\eqref{eq:magnetic_cocycle}.

\begin{definition}\label{dfn:magnetic.2.cocycle}
A 2-cocycle $\sigma\colon \Gamma\times\Gamma\to{\rm U}(1)$ is said to be a \emph{magnetic 2-cocycle for $(X,\Gamma)$}, if it is cohomologous to the 2-cocycle \eqref{eq:magnetic.translation.projective} for some magnetic translation operators $\{U^{(\mathscr{A})}_g\}_{g\in\Gamma}$ as constructed in \eqref{eq:magnetic.trans} above.
\end{definition}
\begin{remark}
If $X$ is Riemmanian and $\Gamma$ acts by isometries, then the magnetic translations $U_g$ are \emph{unitary} operators on the $L^2$-sections of $\mathcal{L}\to X$.
\end{remark}

Let us discuss the effect of various choices made in the construction of
magnetic translations. First, there is a phase ambiguity in each $U_g$ due
to the choice of origin $\mathcal{O}$ in formula
\eqref{eq:magnetic_antiderivative}. This does not affect the cohomology
class of $\sigma$.  Next, a global gauge transformation (changing the
trivialization $\mathcal{L}\simeq X\times \C$) means that we apply the
operator $V$ of multiplication by some smooth function $X\to {\rm U}(1)$.
As $X$ is simply-connected, we can write this function as $\exp(i\Lambda)$
for some smooth $\Lambda\colon X\to\mathbb{R}$. The connection 1-form is
transformed to $\mathscr{A}^\prime=\mathscr{A}+d\Lambda$. As in
\eqref{eq:magnetic_antiderivative} and \eqref{eq:magnetic.trans},
write 
\begin{equation}\label{eq:gauge_transformed_antiderivative}
I^\prime_g(x)=\int_\mathcal{O}^x(\mathscr{A}^\prime-g^*\mathscr{A}^\prime)=I_g(x)+(\Lambda-g^*\Lambda)(x)-(\Lambda-g^*\Lambda)(\mathcal{O}).
\end{equation}
Then the modified magnetic translation operators
\[
U_g^\prime\defeq 
U_g^{(\mathscr{A}^\prime)}=e^{iI^\prime_{g^{-1}}}\cdot (g^{-1})^*
\]
are related to the old ones as follows,
\begin{align*}
    \Ad_V(U_{g^{-1}})=e^{i\Lambda}U_{g^{-1}}e^{-i\Lambda}
    &=e^{i\Lambda}e^{iI_g}e^{-ig^*\Lambda}g^*\\
    &=e^{i(\Lambda-g^*\Lambda)(\mathcal{O})}e^{iI_g^\prime(x)}g^* & (\mathrm{Eq.}\,\eqref{eq:gauge_transformed_antiderivative})\\
    &=e^{i(\Lambda-g^*\Lambda)(\mathcal{O})}U^\prime_{g^{-1}},
\end{align*}
i.e.,
\[
\Ad_V(U_g)=\omega(g)\cdot U^\prime_{g},\qquad g\in \Gamma,
\]
for some phases $\omega(g)=e^{i(\Lambda-(g^{-1})^*\Lambda)(\mathcal{O})}$. This verifies that the 2-cocycles for $U_g$ and for $U_g^\prime$ are cohomologous.

\subsection{Standarized 2-cocycles in abelian case}
If $\Gamma$ is a (discrete) group, a $2$-cocycle $\sigma\colon \Gamma\times\Gamma\to\mathrm{U}(1)$ can be viewed as a way to specify a ``twisted'' group law on $\mathrm{U}(1)\times \Gamma$, namely,
\begin{equation}\label{eq:twisted_group_composition}
(\lambda_1,g_1)\cdot(\lambda_2,g_2)=(\lambda_1\lambda_2 \sigma(g_1,g_2), g_1g_2),\qquad \lambda_1,\lambda_2\in \mathrm{U}(1),\, g_1, g_2\in \Gamma.
\end{equation}
The group with composition law \eqref{eq:twisted_group_composition} is denoted $\mathrm{U}(1)\times_\sigma \Gamma$, and we view it as a central extension of $\Gamma$ by ${\rm U}(1)$, i.e., a short exact sequence
\begin{equation}\label{eq:central_extension_sequence}
\begin{tikzcd}
1\ar[r] &\mathrm{U}(1)\ar[r] &\mathrm{U}(1)\times_\sigma \Gamma \ar[r] &\Gamma \ar[r] & 1.
\end{tikzcd}
\end{equation}
We say that the central extensions $\mathrm{U}(1)\times_\sigma \Gamma$ and
$\mathrm{U}(1)\times_{\sigma^\prime} \Gamma$ are isomorphic if there is an
isomorphism $\varphi\colon \mathrm{U}(1)\times_\sigma \Gamma\to \mathrm{U}(1)\times_{\sigma^\prime} \Gamma$ such that the diagram
\begin{equation}\label{eq:CE_cohom_cocycles}
\begin{tikzcd}
& & \mathrm{U}(1)\times_\sigma\Gamma\ar[dr]\ar[dd,"\sim","\varphi"'] & & \\
1\ar[r] & \mathrm{U}(1)\ar[ur]\ar[dr] & &\Gamma\ar[r] & 1 \\
& & \mathrm{U}(1)\times_{\sigma^\prime} \Gamma\ar[ru] & &
\end{tikzcd}
\end{equation}
commutes. This occurs iff $\sigma$ and $\sigma^\prime$ are cohomologous (\cref{dfn:cohomologous_cocycles}).

\begin{lemma}\label{lem:Zd_standard_cocycle}
Let $\sigma\colon \Z^d\times\Z^d\to\mathrm{U}(1)$ be a 2-cocycle. There is
a unique real antisymmetric $d\times d$ matrix $\Theta$ such that $\sigma$
is cohomologous to
\begin{equation}\label{eq:standard_2_cocycle}
\sigma_\Theta\colon (x,y)\mapsto \exp\left(2\pi i \sum_{1\leq j<k \leq d} \Theta_{jk}x_jy_k\right),\qquad x,y\in\Z^d.
\end{equation}
\end{lemma}
\begin{proof}
This result can be found in
\cite{Backhouse-Bradley:Projective_representations_space_groups}*{Theorem
4.1}. As mentioned in \cite{Hannabuss-Mathai-Thiang:T-duality_NC}, it also
follows from \cite{Kleppner:Multipliers_abelian_groups}*{Lemma 7.2} and
references therein, and we provide a full account here for convenience.

For any (discrete) group $\Gamma$, the antisymmetrization of a 2-cocycle
$\sigma\colon \Gamma\times\Gamma\to\mathrm{U}(1)$ is denoted
\begin{equation}\label{eq:AS_cocycle}
\sigma^{(2)}(x,y)\defeq \frac{\sigma(x,y)}{\sigma(y,x)},\qquad x,y\in\Gamma.
\end{equation}
If $\sigma$ is cohomologous to $1$, then the formulae \eqref{eq:cohomologous_cocycles} and \eqref{eq:AS_cocycle} show that $\sigma^{(2)}=1$. 

We claim that the converse also holds when $\Gamma$ is abelian. In this case, the antisymmetrization $\sigma^{(2)}$ is easily checked to be a bicharacter, i.e., a homomorphism with respect to each argument.
It follows that $\sigma^{(2)}$ is also a 2-cocycle.  Now suppose $\sigma$
has $\sigma^{(2)}=1$. Then $\sigma$ is symmetric, implying that
$\mathrm{U}(1)\times_\sigma \Gamma$ is abelian (see formula
\eqref{eq:twisted_group_composition}). So
\eqref{eq:central_extension_sequence} is a short exact sequence in the
category of abelian groups.  As $\mathrm{U}(1)$ is a divisible group, it is
an injective object in this category, meaning that
\eqref{eq:central_extension_sequence} admits a splitting homomorphism
$s\colon \Gamma\to \mathrm{U}(1)\times_\sigma \Gamma$. Then
$\mathrm{U}(1)\times_\sigma \Gamma$ is isomorphic to the direct product
extension $\mathrm{U}(1)\times \Gamma$ with trivial 2-cocycle, i.e.,
$\sigma$ is cohomologous to $1$.

We have shown, for abelian $\Gamma$, that a 2-cocycle $\sigma$ is cohomologous to $1$ iff $\sigma^{(2)}=1$. Equivalently, two 2-cocycles $\sigma,\sigma^\prime$ on an abelian group $\Gamma$ are cohomologous iff their antisymmetrizations $\sigma^{(2)},\sigma^{\prime(2)}$ coincide. 

Now specialize to the case of $\Gamma=\Z^d$. Given an arbitrary 2-cocycle $\sigma$ on $\Z^d$, we saw that $\sigma^{(2)}$ is an antisymmetric bicharacter. By considering the behaviour on the standard basis of $\Z^d$, we must have
\[
\sigma^{(2)}(x,y)=\exp\big(2\pi i \sum_{j,k=1}^d \Theta_{jk}x_jy_k\big),
\]
for some real antisymmetric matrix $\Theta=(\Theta_{jk})_{j,k=1,\ldots,d}$. For this matrix $\Theta$, consider $\sigma_\Theta$ given in formula \eqref{eq:standard_2_cocycle}. It is a bicharacter (thus it is indeed a 2-cocycle), and its antisymmetrization is also
\begin{align*}
\sigma_\Theta^{(2)}(x,y)&=\exp\big(2\pi i \sum_{1\leq j<k\leq d}\Theta_{jk}x_j y_k\big)\exp\big(-2\pi i \sum_{1\leq k<j\leq d}\Theta_{kj}y_k x_j\big)\\
&=\exp\big(2\pi i \sum_{j,k=1}^d\Theta_{jk}x_j y_k\big). & (\Theta_{jk}=-\Theta_{kj})
\end{align*}
As $\sigma$ and $\sigma_\Theta$ have the same antisymmetrizations, they are cohomologous.
\end{proof}

\subsection{Magnetic translations in Euclidean space}\label{sec:Euclidean.magnetic.translations}
Let $\Gamma\simeq \Z^d$ be a rank-$d$ lattice of translations of $d$-dimensional Euclidean space $X$.
Let $\Theta$ be a real $d\times d$ antisymmetric matrix, and consider the connection 1-form
\begin{equation}\label{eq:standard.connection.form}
\mathscr{A}=-2\pi \sum_{1\leq j<k\leq d} \Theta_{jk} x_kdx_j,
\end{equation}
where $x_j, j=1,\ldots,d$ are linear coordinates for $X$ adapted to a basis $\{e_j\}_{j=1,\ldots,d}$ for $\Gamma\simeq\Z^d$. (So $e_j$ increases the $x_j$ coordinate by $1$). The curvature 2-form is
\[
\mathscr{F}=d\mathscr{A}=2\pi \sum_{1\leq j<k\leq d} \Theta_{jk}dx_j\wedge dx_k,
\]
and is clearly $\Gamma$-invariant.
The magnetic translations operators for this $\mathscr{A}$ have a 2-cocycle, denoted $\sigma_\mathscr{A}$, whose formula is given by \eqref{eq:magnetic_cocycle}. Explicitly,
\begin{align*}
\sigma_\mathscr{A}(e_j,e_k)&=\exp\left(i\int_\mathcal{O}^{-e_j\cdot\mathcal{O}}(\mathscr{A}-(-e_k)^*\mathscr{A}\right)\\
&=
\begin{cases}\exp\left(-2\pi i\Theta_{jk}\int_\mathcal{O}^{-e_j\cdot\mathcal{O}} dx_j\right)=\exp(2\pi i\Theta_{jk}), & 1\leq j<k\leq d,\\
1, & \mathrm{otherwise}.
\end{cases}\\
&=\sigma_\Theta(e_j,e_k),
\end{align*}
where $\sigma_\Theta$ is the standardized 2-cocycle of \eqref{eq:standard_2_cocycle}. This shows that $\sigma_\Theta$ is a magnetic 2-cocycle in the sense of \cref{dfn:magnetic.2.cocycle}.

By \cref{lem:Zd_standard_cocycle}, any 2-cocycle $\sigma$ on $\Z^d$ is cohomologous to some standardized $\sigma_\Theta$, thus any $\sigma$ is a magnetic 2-cocycle.

\section{Categorical terminology}
\label{app:categorical_terminology}

The terminology of directed systems and their
colimits varies in the literature (see
\cref{sec:remarks_categorical_terminology}). In this appendix, we
establish the terminology adopted throughout this paper.

\subsection{Diagrams and their colimits}
Let \( (I,\prec) \) be a poset, i.e.~a set \( I \)
with a partial order \( \prec \). It can be viewed as a category, whose
objects are elements in \( I \); and for every pair of objects
\( (a,b) \) in \( I \), there is a unique morphism \( a\to b \) iff
\( a\prec b \). 

An \( I \)-shaped diagram in a category \(
\mathfrak{C} \) is a functor \( F\colon I\to\mathfrak{C} \).   
Set
\[ 
  c_i\defeq F(i),\,i\in I;\quad 
  \phi_{j\backslash i}\defeq F(i\to j),\,i\prec j,
\]
then an \( I \)-shaped diagram \( F \) can
be concretely described by the following data:
\begin{itemize}
\item a collection of objects \( (c_i)_{i\in I} \) in
\( \mathfrak{C} \) labelled by \( I \);
\item a collection of morphisms
\[ 
  \set*{\phi_{j\backslash i}\colon c_i\to c_j\,:\,i,j\in I,\,i\prec
  j},
\]
such that \( \phi_{k\backslash  j}\circ\phi_{j\backslash
i}=\phi_{k\backslash i} \) holds if \( i\prec j\prec k \). 
\end{itemize}
We call \( \phi_{j\backslash i} \) the \emph{connecting
morphisms} of the diagram. 
The notation $\phi_{j\backslash i}$ is motivated by the coset construction
of our symmetry breaking morphisms $\phi_{H\backslash G}$, see
\eqref{eq:phi_HG}, when a symmetry group $G$ is broken to a subgroup $H$.
When the connecting morphisms are clear from the context, then we may also
denote an \( I \)-shaped diagram by its underlying objects \( (c_i)_{i\in
I} \). 

\begin{definition}
Let \( F\colon I\to\mathfrak{C} \) be an \( I
\)-shaped diagram in \( \mathfrak{C} \). 
The \emph{colimit} of \( F \)  
is an object \( \colim F \) in \( \mathfrak{C} \),
together with universal morphisms
\[ 
  \phi_{\infty\backslash i}\colon c_i\to \colim F,\quad \text{for each}\ i\in I
\]
that is characterised by the following universal property: if \(
c \) is an object in \( \mathfrak{C} \) and \( \psi_i\colon c_i\to c \) are
morphisms satisfying
\[ 
  \psi_{i}=\psi_{j}\circ\phi_{j\backslash i},\quad \text{for each \( i\prec j \),} 
\]
then there is a unique morphism \(
h\colon \colim F\to c \) such that 
\( \psi_i=h\circ\phi_{\infty\backslash i} \) 
for each \( i\in I \), see     
\eqref{fig:colimit}.
\end{definition}

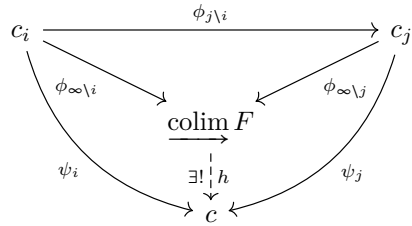
\begin{figure}[h!]
\begin{tikzcd}[column sep=1.5cm]
c_i \arrow[rr, "\phi_{j\backslash i}"] \arrow[rd, "\phi_{\infty\backslash i}"'] 
\arrow[rdd, "\psi_i"', bend right] && 
c_j \arrow[ld, "\phi_{\infty\backslash j}"] \arrow[ldd, "\psi_j", bend
left] \\
& \colim F \arrow[d, dashed, "\exists !"', "h"] & \\ 
& c &
\end{tikzcd}
\caption{Universal property of the colimit}
\end{figure}\label{fig:colimit}
In the concrete description of
\( F \) by objects and morphisms in \( \mathfrak{C} \),
we also write \( \colim_{i\in I}F(i) \) for \( \colim F \).
If a diagram \( F\colon I\to\mathfrak{C} \) has a colimit, then the colimit
is unique up to canonical isomorphism, due to the universal property
\eqref{fig:colimit}.

Let \( F,G\colon I\to \mathfrak{C} \) be diagrams with the same shape poset
\( I \). A \emph{natural transformation} \( \eta\colon F\Rightarrow G \) is
a collection of morphisms \( \{\eta_{i}\colon F(i)\to
G(i)\}_{i\in I} \), such that
the following diagram commutes for every \( i\prec j \):
\[ \begin{tikzcd}
F(i) \arrow[r, "F(i\to j)"] \arrow[d, "\eta_i"] 
& F(j) \arrow[d, "\eta_j"] \\
G(i) \arrow[r, "G(i\to j)"] & G(j).
\end{tikzcd} \]
If $F,G$ have colimits, then a natural transformation $\eta\colon F\Rightarrow G$
induces a unique morphism between their colimits.
A natural transformation \( \eta\colon F\Rightarrow G \) is a \emph{natural
isomorphism} if \( \eta_i \) is an isomorphism in \( \mathfrak{C} \) for
each \( i\in I \). 

For diagrams with \emph{directed} posets
as their shapes \( I \), their colimits may be computed by passing to some
suitable subdiagram.

\begin{definition}\label{def:direct_colimit}
\begin{enumerate}
\item 
A poset \( (I,\prec) \) is called \emph{directed}, if for
every pair of elements \( a,b\in I \), there exists \( c\in
I \) such that \( a\prec c \) and \( b\prec c \).
\item 
A \emph{directed system} in a category \( \mathfrak{C} \) is an \(
I \)-shaped diagram \( F\colon I\to\mathfrak{C} \) 
for a directed poset \( I \).
\item 
A \emph{direct colimit} is the colimit of a directed system.
\end{enumerate}
\end{definition}

\begin{definition}[\cite{MacLane:Categories_for_the_working_mathematician}*{Section IX.3}]\label{dfn:cofinal}
Let \( (I,\prec) \) be a directed poset. A subset \(
J\subseteq I \) is called \emph{cofinal}, if for
every \( i\in I \), there exists \( j\in J \) such that \( i\prec j \).   
\end{definition}

If \( J \) is a cofinal subset of a directed poset \(
(I,\prec) \), then the restriction of \( \prec \) to \(
J \) turns it into a directed poset as well. The inclusion of \(
J \) into \( I \) is a functor when $(I,\prec)$ and $(J,\prec)$ are viewed
as categories. Given an \( I \)-shaped directed system
\( F\colon I\to\mathfrak{C} \), then it restricts to a \(
J \)-shaped directed system defined by   
\[ 
  F|_{J}\colon J\hookrightarrow I\xrightarrow{F}\mathfrak{C}.
\]
If the colimits of both \( F\colon I\to\mathfrak{C} \) and \(
F|_J\colon J\to\mathfrak{C} \) exist, then there is a
canonical map between the direct colimits:
\begin{equation}\label{eq:cofinal_canonical_map}
h\colon \colim F|_J\to\colim F.
\end{equation}
by their universal properties.

\begin{lemma}[\cite{MacLane:Categories_for_the_working_mathematician}*{Section
IX.3, Theorem 1}]\label{lem:cofinal_isomorphism}
Let \( J\subseteq I \) be a cofinal subset of a
directed poset \( (I,\prec) \). Let \( F\colon
I\to\mathfrak{C} \) be an \( I \)-shaped diagram
whose direct colimit exists. Then the direct colimit
\[ 
  \colim F|_{J}
\]
also exists, and the canonical map \( h \) in
\eqref{eq:cofinal_canonical_map} is an isomorphism.
\end{lemma}

\subsection{Colimits of \Cst-algebras and abelian groups}
By a directed system of \Cst-algebras, we shall mean a directed system in
the category \( \catCst \) whose objects are
\Cst-algebras; and morphisms are \st-homomorphisms between them. 
The category \( \catCst \) 
is closed under taking (small) colimits. In particular, for every poset 
\( I \), the colimit of the diagram \( F\colon I\to\catCst \)
exists.

If \( F \) is a directed system of
\Cst-algebras with their connecting morphisms being 
injective \st-homomorphisms, then
the direct colimit can be described as the closure of their union.

\begin{lemma}[\cite{Blackadar:K-theory}*{Section
3.3}]\label{lem:direct_colimit_Cst-subalgebras}
Let \( (A_i)_{i\in I} \) be a directed system of \Cst-algebras, 
such that every \( A_i \) is a \Cst-subalgebra of a fixed
\Cst-algebra \( B \); and that the connecting morphisms are given by
inclusions \( A_i\hookrightarrow A_j \) for \( i\prec j \). Then
\[ 
  \colim_{i\in I}A_i\simeq
  \overline{\bigcup_{i\in I}A_i}
\]
where the closure is taken inside \( B \). The universal morphisms \(
\phi_{i}\colon A_i\to\colim_{i\in I}A_i \) are given by inclusion maps as
well.
\end{lemma}

For each \( n \), the K-theory group \( \K_n \) is a functor from \(
\catCst \) to the category of abelian groups \( \catAb \).    
It is well-known that \( \K_n \)
commutes with direct colimits. That is, let
\( (A_i)_{i\in I} \) be a directed system of \Cst-algebras.
Then
\[ 
  \colim_{i\in I}\K_n(A_i)\simeq\K_n(\colim_{i\in I}A_i)
\]
for the K-theory functors \( \K_n \). We note that this ``continuity''
property of $\K_n$ may not hold for a
general colimit whose shape is not a directed poset.

\begin{remark}
The categorical statements about \Cst-algebras and K-theory 
hold for real \Cst-algebras and KO-theory, see
\cite{Boersema-Schochet:Real_K-theory}.
\end{remark}

Now let \( \catAb \) be the category whose objects are abelian groups and
morphisms are group homomorphisms. 

\begin{lemma}\label{lem:concrete_colimit_limit_in_Ab}
Let \( (G_i)_{i\in I} \) be a diagram of abelian groups, whose connecting
morphisms are given by \( \phi_{j\backslash i}\colon G_i\to G_j \). 
Then its colimit \( \colim_{i\in I}G_i \) is isomorphic to
\[ 
  \bigoplus_{i\in I}G_i/R,
\]
where \( R\leq\bigoplus_{i\in I}G_i \) is the subgroup generated by all
elements of the form
\[ 
  g_j-\phi_{j\backslash i}(g_i),\quad i\prec j.
\]
\end{lemma}

The following two lemmas are used in the paper.
\cref{lem:tensor_product_colimit_Ab} is due to the fact that a tensor
product of abelian groups 
is a cokernel, hence commutes with any other colimit.
\cref{lem:direct_colimit_Ab_exact} is Grothendieck's Axiom (AB5) for
abelian categories \citelist{\cite{Grothendieck:Tohoku}
\cite{Weibel:Homological_algebra}*{A.4.6}}, which is satisfied by the
category \( \mathfrak{Ab} \) (more generally, the category of modules over
a commutative unital ring \( R \)).

\begin{lemma}\label{lem:tensor_product_colimit_Ab}
Let \( (G_i)_{i\in I} \) be a diagram of abelian groups, and \( H \) be
another abelian group. Then
\[ 
  (\colim_{i\in I}G_i)\otimes_\Z H\simeq\colim_{i\in I}(G_i\otimes_\Z H).
\]
\end{lemma}
As a special case, 
let \( H=\Q \), and let 
\( (G_i)_{i\in I} \) be the K-theory groups \( \K_n(A_i) \)  
of a directed system of
\Cst-algebras \( (A_i)_{i\in I} \). Then \(
G_i\otimes_{\Z}\Q=\K_n(A_i)\otimes_\Z\Q \)
is the rational K-theory groups of \( A_i \) (see
\cite{Cuntz-Meyer-Rosenberg:K-theory}*{Definition 8.1}). Then
\cref{lem:tensor_product_colimit_Ab} implies that rational K-theory also
commutes with direct colimits of \Cst-algebras.

Below, we denote the trivial abelian group by \( 0 \).

\begin{lemma}\label{lem:direct_colimit_Ab_exact}
Let \( (G_i)_{i\in I}, (H_i)_{i\in I} \) and \( (K_i)_{i\in I} \) be
directed systems of abelian groups, such that for every \( i,j\in I \) with \( i\prec j \), 
\[ 
\begin{tikzcd}
0 \arrow[r] & K_i \arrow[r] \arrow[d] & G_i \arrow[r] \arrow[d]  & H_i
\arrow[r] \arrow[d] & 0 \\
0 \arrow[r] & K_j \arrow[r] & G_j \arrow[r] & H_j \arrow[r] & 0 
\end{tikzcd}
\]
is a commutative diagram whose rows are short exact sequences, and vertical
maps are the connecting morphisms of these directed systems. Then there is
a short exact sequence of abelian groups
\[ 
  \begin{tikzcd}
0 \arrow[r] & \colim\limits_{i\in I}K_i \arrow[r] & \colim\limits_{i\in I}G_i \arrow[r] &
\colim\limits_{i\in I} H_i \arrow[r] & 0. 
\end{tikzcd} 
\]
\end{lemma}

\subsection{Remarks on categorical terminology}
\label{sec:remarks_categorical_terminology}
Below we collect a few categorical terminologies in the literature and clarify the relation with our usage. 
\begin{enumerate}
\item In the field of operator algebras, a \emph{directed poset} 
is often called a \emph{directed set}. A
\emph{directed system} is also called an \emph{inductive system}.
In some literature, the shape poset \( (I,\prec) \) of a directed system 
is required to be isomorphic to the totally ordered poset \( (\N,\leq) \).
\item The \emph{direct colimit} is sometimes called a \emph{direct limit} or
\emph{inductive limit}. We remark, however, that this terminology is not
compatible with the modern categorical language of ``limits'' as a distinct
concept from ``colimits''. 
\item One could also consider colimits of $I$-shaped diagrams with $I$
being a more general (small) category. The term ``direct limit'' is
sometimes used as an synonym for a colimit, whether or not $I$ is a
directed poset. Directed colimits in our sense are special instances of
\emph{filtered colimits} (see
\cite{MacLane:Categories_for_the_working_mathematician}*{Section IX.1}).
Taking the filtered colimit of \Cst-algebras commutes with the K-theory
functor, whereas this does not hold for the colimit of a diagram of
\Cst-algebras with more general shape. 
\item The term \emph{cofinal} is referred to as \emph{final} in
\cite{MacLane:Categories_for_the_working_mathematician}*{Section IX.3}. We
prefer to work with ``co'' here, because \cref{lem:cofinal_isomorphism}
shows that restricting a diagram to a cofinal shape poset 
preserves its colimit.
\end{enumerate}

\section{Real \Cst-algebras and KO-theory}
\label{sec:real_Cst-algebras}
We briefly recall the definition of real \Cst-algebras and their KO-theory.
A comprehensive review is \cite{Boersema-Schochet:Real_K-theory}. We
highlight the \( \KO_*(\R) \)-module structure of the KO-theory \(
\KO_*(A) \) of a real \Cst-algebra \( A \), needed for the analysis of
symmetry breaking in \cref{sec:real_d-dim}.

\subsection{Real \Cst-algebras}
A real \Cst-algebra is a norm-closed \st-subalgebra of \(
\Bdd(\mathcal{H}_\R) \) for a real Hilbert space \( \mathcal{H}_\R \). 

The complexification of a real \Cst-algebra \( A\subseteq
\Bdd(\mathcal{H}_\R) \) is the complex \Cst-algebra \( A_\C\defeq
A\otimes_\R\C\subseteq \Bdd(\mathcal{H}_\R\otimes_\R\C) \). Thus 
every real \Cst-algebra \( A \) can be described as a 
real \st-subalgebra of its complexification.

Let \( X \) be a locally compact Hausdorff space, then \( \Co(X)_\R \), the
continuous real-valued functions on \( X \) that vanish at infinity, is a
real \Cst-algebra. More generally, 
an \emph{involutive space} is a space \( X \) equipped with a 
\( \Z/2 \)-action, i.e., a homeomorphism \( \tau\colon X\to X \) such that
\( \tau^2=\id_X \). An involutive space yields a real \Cst-algebra  
\[ 
  \Co(X,\tau)\defeq
  \set*{f\in\Co(X)_\C\,:\,\overline{f(x)}=f(\tau(x))}.
\]
As a special case, \( \Co(X)_\R=\Co(X,\id_X) \) for the trivial involution
\( \id_X\colon X\to X \). The real Gelfand--Naimark Theorem due to Arens
and Kaplansky \cite{Arens_Kaplansky:Topological_representation} says that
every commutative real \Cst-algebra arsies as \( \Co(X,\tau) \) for some
involutive space \( (X,\tau) \). Namely, $X$ is the character space of
the complex commutative $\Cst$-algebra $\Co(X)_{\C}$, and complex
conjugation induces an involution $\tau$ on $X$.

Let \( \I\defeq (-\pi,\pi) \) with involution \( \tau_{\I}(x)=-x \).
Then \( \I \) is \( \Z/2 \)-equivariantly homeomorphic to the real
line \( \R \) with involution \( \tau_\R(x)=-x \).

Given a real \Cst-algebra \( A \), we write
\begin{equation}\label{eq:suspension_desuspension_real}
\begin{aligned} 
\mathrm{S}A\defeq \Co(\I,\id)\otimes_\R A,&\quad
\mathrm{V}A\defeq \Co(\I,\tau_{\I})\otimes_\R A,\\
\mathrm{S}^nA\defeq
\underbrace{\mathrm{S}\mathrm{S}\dots\mathrm{S}}_{\text{\( n \)-fold}}A,
&\quad \mathrm{V}^nA\defeq
\underbrace{\mathrm{V}\mathrm{V}\dots\mathrm{V}}_{\text{\( n \)-fold}}A,
\end{aligned}
\end{equation}
We call the real \Cst-algebra \( \mathrm{S}A \) the \emph{suspension} of
\( A \), and \( \mathrm{V}A \) the \emph{desuspension} of \( A \).
Both \( \mathrm{S}A \) and \( \mathrm{V}A \) have complexification \(
\Co(\R)_\C\otimes_\C A_\C \), which is the usual suspension of a complex
\Cst-algebra \( A \). We note, however, that \( \mathrm{S}A \) and \(
\mathrm{V}A \) are not isomorphic as real \Cst-algebras, and they have
different KO-theory, see \cref{lem:KO-theory_suspension_desuspension}.

\subsection{\KO-theory}
Let \( A \) be a real \Cst-algebra. There is a sequence of abelian groups
\[ 
  \KO_i(A),\quad i\in\Z,
\]
called the KO-theory groups of \( A \). 
If \( A \) is a unital real
\Cst-algebra, then \( \KO_0(A) \) is the Grothendieck group of the monoid
of equivalence classes of projections in \( \Mat_{\infty}(A) \). If \( A \)
is non-unital, then one defines
\[ 
  \KO_0(A)\defeq \ker(\KO_0(q)\colon \KO_0(A^+)\twoheadrightarrow
  \KO_0(\R)),
\]
where \( A^+ \) is the real unitisation of \( A \).

Higher KO-theory groups are defined using suspensions: we define
\[
  \KO_n(A)\defeq \KO_0(\mathrm{S}^nA),
\]
here we recall that \( \mathrm{S}^nA \) is  
the \( n \)-fold suspension of \( A \). 
The real Bott periodicity of KO-theory (see
\citelist{\cite{Bott:Periodicity} \cite{Wood:Banach_algebras_periodicity}})
says that
\[ 
  \KO_i(A)\simeq\KO_{i+8}(A),\quad \text{for every \( i\in \Z \)}. 
\]

The KO-theory of a real \Cst-algebra \( A \) is sometimes more conveniently
described as a \( \KO_*(\R) \)-module. In order to do this,
define the \( \Z \)-graded abelian group  
\[
  \KO_*(A)\defeq \bigoplus_{i\in\Z}\KO_i(A).
\]
The external tensor product of real \Cst-algebras turns \(
\KO_*(\R) \) into a unital, \( \Z \)-graded commutative ring, and \(
\KO_*(A) \) into a graded \( \KO_*(\R) \)-module.  
The exterior product on K-theory induces a graded \( \KO_*(\R)
\)-module structure on \( \KO_*(A) \). Moreover, every \st-homomorphism 
$A\to B$
induces a degree zero module map between the graded modules 
\( \KO_*(A) \to\KO_*(B) \).

The \( \Z \)-graded 
ring \( \KO_*(\R) \) can be described as (see 
\citelist{\cite{Karoubi:K-theory}*{Theorem 5.19}
\cite{Boersema-Schochet:Real_K-theory}*{Theorem 22.2}}):
\begin{equation}\label{eq:KO-ring_of_R}
  \Z[\eta,\alpha,\beta,\beta^{-1}]/(2\eta,\eta^3,\eta\alpha,\alpha^2-4\beta),
\end{equation} 
where
\[ 
  \deg(\eta)=1,\quad\deg(\alpha)=4,\quad\deg(\beta)=8.
\]
In particular, we have
\[ 
  \KO_i(\R)\simeq\begin{cases} 
  \Z, & \text{if}\ i\equiv 0,4\hphantom{,6,7}\mod 8;\\
  \Z/2, & \text{if}\ i\equiv 1,2\hphantom{,6,7}\mod 8;\\
  0, & \text{if}\ i\equiv 3,5,6,7 \mod 8,
  \end{cases}
\]
and a generator of the abelian group \( \KO_i(\R) \) is given by the
product of \( \beta^k \) with the generator of \( \KO_*(\R) \)  
in degree \( i-8k \), where \( k
\) is the unique integer such that \( i-8k\in\{0,1,\dots,7\} \).

\begin{lemma}[\cite{Boersema:Real_Cst-algebras_Kuenneth_formula}*{Corollary
5.7}]\label{lem:KO-theory_suspension_desuspension}
For the suspension/desuspension real \Cst-algebras \(\mathrm{S}A\)
and \(\mathrm{V}A\) as in \eqref{eq:suspension_desuspension_real},
there are isomorphisms of \( \KO_*(\R) \)-modules
\[ 
  \KO_*(\mathrm{S}A)\simeq\KO_{*+1}(A),\quad 
  \KO_*(\mathrm{V}A)\simeq\KO_{*-1}(A).
\]
\end{lemma}
Let \( A \) and \( B \) be real \Cst-algebras and \( f\colon A\to B \) be a
\st-homomorphism. Then \( f \) induces maps
\[
  \KO_i(f)\colon \KO_i(A)\to\KO_i(B)
\]
in KO-theory for each \( i \). These maps \( \KO_i(f) \) are, in
particular, components of a degree-zero
\( \KO_*(\R) \)-module map between the corresponding \(
\KO_*(\R) \)-modules \( \KO_*(A) \) and \( \KO_*(B) \).
Such a description is particularly useful 
  when \( \KO_*(A) \) is a finitely-generated 
  free \( \KO_*(\R) \)-module, in which case, the
induced maps \( \KO_i(f) \)  
are completely determined by their behaviour on
generators.

\begin{remark}\label{rmk:K-ring_of_C}
For a complex \Cst-algebra \( A \), one can also 
describe the \( \Z \)-graded abelian
group \( \K_*(A) \) as a module over the unital, graded commutative ring
\( \K_*(\C) \). The \( \Z \)-graded ring \( \K_*(\C) \) is the ring of Laurent polynomials
\[ 
  \Z[\beta,\beta^{-1}],\quad \deg(\beta)=2,
\]
where multiplication with \( \beta \) encodes the complex Bott periodicity
isomorphisms \( \K_n(A)\simeq\K_{n+2}(A) \). Thus
\( \K_*(A) \) being a \( \K_*(\C) \)-module just says that it is given by a
\( \Z/2 \)-graded abelian group \( \K_0(A)\oplus\K_1(A) \) together with
Bott periodicity.

Similar to the case of real \Cst-algebras,
a \st-homomorphism \( f\colon A\to B \) between complex \Cst-algebras
gives a degree-zero \( \K_*(\C) \)-module map \( \K_*(A)\to\K_*(B) \).
The paragraph above then says that this is the same data as a pair of group
homomorphisms \( \K_0(A)\to\K_0(B) \) and \( \K_1(A)\to\K_1(B) \).
\end{remark}

\section{Assembly maps for contractible nilpotent Lie groups}
\label{sec:coarse_Baum-Connes}

The coarse Baum--Connes conjecture \cite{Higson-Roe:Coarse_Baum-Connes}
says that for a metric space \( X \) satisfying certain properties, 
the \emph{assembly map}
\begin{equation}\label{eq:assembly_map}
  \mu\colon \K_i(X)\to\K_i(\Cst(X))
\end{equation}
is an isomorphism for every \( i \). 
Here \( \K_i(X) \) is the (analytic) K-homology of the
space \( X \), and \( \K_i(\Cst(X)) \) is the K-theory of the Roe algebra
of \( X \).

The assembly map \( \mu \) can be presented in different ways depending on
the model of K-homology. The original formulation uses Paschke duality
\cite{Higson-Roe:Coarse_Baum-Connes}. In
\cite{Yu:Localization_algebras_coarse_Baum-Connes}, the assembly map 
is formulated using the localisation algebra of \( X \). 

The left side of \( \mu \) encodes the topological (local) data of the
space \( X \), whereas the right side encodes the coarse (large-scale)
data. It is known that \( \mu \) 
fails to be an isomorphism for a general metric space
\( X \). Nevertheless, we have the following landmark result due to Yu:

\begin{lemma}[\citelist{\cite{Higson-Roe:Coarse_Baum-Connes}*{Proposition 3.8} \cite{Yu:Finite_asymptotic_dimension}}]
\label{lem:Yu_theorem}
Let $X$ be a metric simplicial complex which is 
uniformly contractible, has bounded coarse geometry, 
and has finite asymptotic dimension. Then the assembly map
\eqref{eq:assembly_map} is an isomorphism for all \( i \). 
\end{lemma}

Definitions of these properties are given below. The goal of this appendix
is to prove that every contractible nilpotent Lie group \( X \) satisfies
all these properties, whence we may use the coarse Baum--Connes conjecture
to compute \( \K_i(\Cst(X)) \).  

Let \( (X,d) \) be a metric space. 
A subset \( \Lambda\subseteq X \) is called \emph{\( \varepsilon \)-separated} if \( d(x,y)>\varepsilon \)
for every pair of distinct elements \( x,y\in\Lambda \).  
Write \( \Ball(x,r)\defeq \{y\in X\;:\;
d(x,y)<r\} \) the open \( r \)-ball centered at \( x \). 

\begin{definition}[see e.g.~\cite{Higson-Roe:Coarse_Baum-Connes}]
Let \( X \) be a metric space.  
\begin{enumerate}
\item $X$ is \emph{uniformly contractible}, if for any~$r>0$, there
exists~$R>r$, such that the
inclusion~$\Ball(x,r)\hookrightarrow\Ball(x,R)$ is nullhomotopic for
all~$x\in X$.
\item \( X \) has \emph{bounded coarse geometry}, if there exists \(
\varepsilon>0 \), such that the following holds: 
for every \( r>0 \), there exists \( n_r>0
\) such that the intersection of any 
\( \varepsilon \)-separated subset with any \( r \)-ball \( \Ball(x,r) \)
has at most cardinality \( n_r \). 
\item The \emph{asymptotic dimension} of \( X \), denoted by \(
\operatorname{asdim}(X) \), is the smallest integer \(
n\) such that: for any \( r>0 \), there exists a uniformly bounded cover \(
\mathfrak{U}=\{U_i\}_{i\in I}\) of \( X \), such that any \( r \)-ball \(
\Ball(x,r)
\) in \( X \) intersects at most \( (n+1) \)-open sets belonging to \(
\mathfrak{U} \).   
\end{enumerate}
\end{definition}

For the Lemmas below, we recall that a Lie group $X$ can be equipped
with a left-invariant Riemannian metric, whence it is geodesically
complete, thus a proper metric space by the Hopf--Rinow theorem, see
\cite{Jost:Riemannian_geometry}*{Theorem 1.4.8}.

\begin{lemma}\label{lem:uniformly_contractible_Lie_group}
Let \( X \) be a contractible Lie group, equipped with a
left-invariant metric. Then \( X \) is uniformly contractible. 
\end{lemma}

\begin{proof}
Let \( e\in X \) be its unit and let \( r>0 \). 
By left invariance, it suffices to show that \( \Ball(e,r) \) is
contractible inside \( \Ball(e,R) \) for some \( R>r \).  

Since \( X \) is contractible, there exists a continuous map \(
H\colon X\times[0,1]\to X \) such that \( H(g,0)=g \) for all \( g\in X \)
and \( H(g,e)=e \). Let \( K_r\defeq \overline{\Ball(x,r)} \). Then \( K_r
\) is compact because the metric of \( X \) is proper. Thus
\[ 
  C_r\defeq H(K_r\times[0,1])\subseteq X
\]
is compact.

A compact set in a metric space is always bounded. Thus there exists \( R>0
\) such that \( C_r\subseteq \Ball(e,R) \). In particular, \( R \) depends
only on \( r \). Therefore \( H \) restricts to a homotopy between the
inclusion \( \Ball(x,r)\hookrightarrow \Ball(x,R) \) and the identity map.
\end{proof}

\begin{lemma}\label{lem:bounded_coarse_geometry_Lie_group}
Let \( X \) be a Lie group equipped with a left invariant
metric. Then \( X \) has bounded coarse geometry. 
\end{lemma}

\begin{proof}
Let \( e\in X \) be its unit. 
We must show that there exists \( \varepsilon>0 \), such that given any \(
r>0 \), then there exists a constant \( n_r
\) depending only on \( r \), such that any \( r \)-ball \( \Ball(x,r) \)
and any \( \varepsilon \)-separated subset 
\( \Lambda \) intersects at at most \( n_r\)
points. In the following we take \( \varepsilon=1 \).   

By left invariance, it suffices to prove for the \( r \)-balls \( \Ball(e,r)
\) centered at \( e \) and for all \( r>0 \). 
Since \( K_r\defeq \overline{\Ball(x,r)} \) is
compact, it is covered by finitely many \( \tfrac{1}{2} \)-balls,
i.e.~there exists \( x_1,\dots,x_n\in X \) such that
\[ 
  K_r\subseteq \bigcup_{i=1}^n\Ball(x_i,\tfrac{1}{2}).
\]
If \( \Lambda \) is \( 1 \)-separated, then the intersection of 
\( \Lambda \) with \( \Ball(x_i,\tfrac{1}{2}) \) has at most one point.
Thus the intersection of \( \Lambda \) with \( K_r \), hence with \(
\Ball(x,r) \), has at most \( n \)
points, where \( n \) depends only on \( r \).     
\end{proof}

\begin{lemma}[\cite{Carlsson-Goldfarb:Homological_coherence}*{Theorem 3.5}]
\label{lem:asdim_nilpotent_Lie_group}
Let \( X \) be a simply connected, 
nilpotent Lie group, equipped a 
left-invariant metric. Then \( \operatorname{asdim} X=\dim X
\).  
\end{lemma}

\begin{lemma}[\cite{Knapp:Lie_groups_beyond_introduction}*{Theorem 1.104}]\label{lem:nilpotent.is.Rn}
A connected, simply connected nilpotent Lie group is diffeomorphic to \( \R^n \) for some \( n \).
\end{lemma}

Combining
\cref{lem:Yu_theorem,lem:uniformly_contractible_Lie_group,lem:bounded_coarse_geometry_Lie_group,lem:asdim_nilpotent_Lie_group,lem:nilpotent.is.Rn},
we have the following

\begin{corollary}\label{cor:contractible_nilpotent_coarse_Baum-Connes}
Let \( X \) be a connected, simply-connected nilpotent Lie group,
equipped with a left-invariant metric. Then the
assembly map \( \mu\colon \K_i(X)\to\K_i(\Cst(X)) \) of
\eqref{eq:assembly_map} is an isomorphism for
every \( i \). 
\end{corollary}

\bibliography{references} 
\end{document}